\documentclass[prodmode,acmtocl]{acmsmall}

\usepackage[ruled]{algorithm2e}

\SetAlFnt{\small}
\acmVolume{1}
\acmNumber{1}
\acmArticle{1}
\acmYear{2016}
\acmMonth{1}

\doi{0000001.0000001}

\issn{1234-56789}

\usepackage{amsmath,amssymb}
\usepackage{proof}
\usepackage{stmaryrd}
\usepackage{latexsym}
\usepackage{url}
\usepackage{mycommands}
\usepackage{color}
\usepackage{xspace}
\usepackage{tikz}
\usepackage{relsize}
\RequirePackage{txfonts}
\usepackage{microtype}
\usepackage{hyperref}

\begin{document}

% Page heads
\markboth{B. Lellmann and E. Pimentel}{Modularisation of sequent calculi for normal and non-normal modalities
}

% Title portion
\title{Modularisation of sequent calculi for normal and non-normal modalities}
\author{Bj\"{o}rn Lellmann
\affil{Institute of Computer Languages, TU Wien, Austria}
Elaine Pimentel
\affil{Departamento de Matem\'{a}tica, UFRN, Brazil}}

\begin{abstract}
%!TEX root = TOCL.tex

In this work we explore the connections between (linear) nested
sequent calculi and ordinary sequent calculi for normal and non-normal
modal logics. By proposing local versions to ordinary sequent rules we
obtain linear nested sequent calculi for a number of logics, including
to our knowledge the first nested sequent calculi for a large class of
simply dependent multimodal logics, and for many standard non-normal
modal logics. The resulting systems are modular and have separate left
and right introduction rules for the modalities, which makes them
amenable to specification as bipole clauses. While this
granulation of the sequent rules introduces more choices for proof
search, we show how linear nested sequent calculi can be restricted to
blocked derivations, which directly correspond to ordinary sequent
derivations.

%%% Local Variables: 
%%% mode: latex
%%% TeX-master: "TOCL"
%%% End: 

\end{abstract}



%
% End generated code
%

 %Modal logic, proof theory}

\acmformat{Bj\"{o}rn Lelllmann and Elaine Pimentel, 2016. Modularisation of sequent calculi for normal and non-normal modalities.}

\begin{bottomstuff}
\includegraphics[height=2ex]{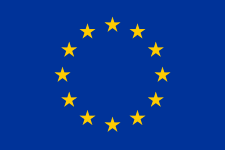} This project has received funding from the European Union's Framework
Programme for Research and Innovation Horizon~2020 (2014-2020) under
the Marie Sk{\l}odowska-Curie Grant Agreement No.~[660047] and from
the WWTF project MA16-028.
This work is supported by the EU under Project GetFun and the Brazilian research agencies CNPq and CAPES.
\end{bottomstuff}

\maketitle

\section{Introduction}\label{sec:intro}
%!TEX root = TOCL.tex

One of the main research topics in proof theory is the proposal of
suitable frameworks for logical systems. Determining which properties
should be taken into account for calling a framework \emph{suitable}
depends on the intended application. For example, \emph{simple}
frameworks are easy to understand and handle, hence this can be a
desirable characteristic.  Another highly desirable property is
\emph{analyticity}. Analytic calculi consist solely of rules that
compose the formulae to be proved in a stepwise manner. As a result,
derivations in an analytic calculus possess the subformula property:
every formula that appears (anywhere) in the derivation must be a
subformula of the formulae to be proved. This is a powerful
restriction on the form of the proofs and can be exploited to prove
important meta-logical properties of the formalised logics such as
consistency, decidability and interpolation.  Also, a framework is
often required to be amenable for \emph{smooth extensions}, in order
to avoid the necessity of a fresh start every time new axioms are
added to the base logic.

Maybe the best known formalism for proposing proof systems is
Gentzen's \emph{sequent calculus}~\cite{gentzen35}.  Due to its
simplicity, sequent calculus appears as an ideal tool for proving
meta-logical properties. However, it is neither expressive enough for
constructing analytic calculi for many logics of interest, nor
scalable in order to capture large classes of logics in a uniform and
systematic way.

In the case of \emph{modal logics}, the limitation of the sequent
framework is glaring.  Undoubtedly, there are sequent calculi for a
number of modal logics exhibiting many good properties (such as
analyticity)\footnote{Analyticity in sequent calculus systems is often
  guaranteed by proving {\em cut elimination}.}, which can be used in
complexity-optimal decision procedures. However, their construction
often seems ad-hoc; they are usually not \emph{modular}, in the sense
that the addition of a single property usually implies a reworking of
the whole system to obtain cut elimination; and they mostly lack
properties such as separate left and right introduction rules for the
modalities, which are relevant from the point of view of
proof-theoretic semantics and facilitate closer connections to natural
deduction systems.
  
These problems are often connected to the fact that the modal rules in
such calculi usually introduce more than one connective at a time. For
example, in the standard presentation of the rule
\[
\infer[\mathsf{k}]{\Gamma',\Box \Gamma\seq\Box A,\Delta}{\Gamma\seq A} 
\]
for modal logic $\K$~\cite{Chellas:1980fk}, the context $\Gamma$
contains an arbitrary finite number of formulae, each of which is
prefixed with a box in the conclusion.  Thus, if the formulae in
$\Gamma$ really are considered to form part of the context, then this
context is not kept intact when passing over to the premiss. Moreover,
the context $\Box \Gamma$ in the conclusion places a severe
restriction on the side formulae, in that only modalised formulae can
appear. Thus, the rule is not \emph{local} in the sense that it does
not only decompose the principal formula $\Box A$. Alternatively, the
$\mathsf{k}$ rule can also be seen as an infinite set of rules
\[
\left\{\; \vcenter{\infer[\mathsf{k}_n]{\Gamma', \Box B_1,\dots, \Box
      B_n \seq \Box A, \Delta}{B_1,\dots, B_n \seq A}} \mid n \geq
  0\;\right\}
\]
each with a fixed number of principal formulae. While from this point
of view the rules $\mathsf{k}_n$ could be considered local because
they do not place any restriction on the side formulae in
$\Gamma',\Delta$, they explicitly introduce boxed formulae on both
sides of the sequent arrow, and hence explicitly discard the
distinction between left and right rules for the modal
connective. Thus, both of these perspectives are somewhat
dissatisfying. For a more detailed discussion see,
e.g.,~\cite{Wansing:2002fk}.

One way of solving this problem is to consider extensions of the
sequent framework that are expressive enough for capturing these
modalities using separate left and right introduction rules.  This is
possible e.g.\ in the frameworks of \emph{labelled
  sequents}~\cite{Negri:2011} or in that of \emph{nested sequents} or
\emph{tree-hypersequents}~\cite{Bull:1992,DBLP:journals/sLogica/Kashima94,Brunnler:2009kx,Poggiolesi:2009vn,Strasburger:2013}.
In the labelled sequent framework, the trick is accomplished by
explicitly mentioning the Kripke-style relational semantics of normal
modal logics in the sequents. In the nested or tree-hypersequent
framework in contrast, intuitively, a single sequent is replaced with
a tree of sequents, where successors of a sequent are interpreted
under a modality. The modal rules of these calculi govern the transfer
of (modal) formulae between the different sequents, and it can be
shown that it is sufficient to transfer only one formula at a
time. However, the price to pay for this added expressivity is that
the obvious proof search procedure is of suboptimal complexity since
it constructs potentially exponentially large nested
sequents~\cite{Brunnler:2009kx}.

In this work, we reconcile the added superior expressiveness and
modularity of nested sequents over ordinary sequents with the computational behaviour of the
standard sequent framework by proposing the concept of block form
derivations for \emph{linear nested sequents}.
Linear nested sequents~\cite{Lellmann:2015lns} (short: $\LNS$) is
a restricted form of nested sequents where the tree-structure is
restricted to that of a line.
In $\LNS$, a list of standard
sequents is separated by the \emph{nesting operator} $\lns$, with the
head of the list interpreted in the usual way and the tail
interpreted (recursively) under a modal operator.  
The logical rules then act on
the elements of the list, possibly moving formulas from one element to another.
This finer way of representing systems enables both locality and
modularity by decomposing standard sequent rules into smaller components.
For example, the modal rule $\mathsf{k}$ in the linear nested setting is
decomposed into the two rules
\[
     \infer[\Box_L]{\mathcal{S}\con{\Gamma,\Box A \seq \Delta\lnes{\Sigma
          \seq \Pi}}}{\mathcal{S} \con{\Gamma \seq
        \Delta\lnes{\Sigma,A \seq \Pi}}}
           \qquad
           \infer[\Box_R]{\mathcal{G} \lnes{\Gamma\seq \Delta,\Box A}}{\mathcal{G}
      \lnes {\Gamma\seq\Delta\lnes{\;\seq A}}}
\]
Note that different connectives are introduced one at a time by
different (context free) rules, and this entails locality.  Moreover,
decomposing the sequent rules enables modularity since now extensions
of, \eg, the modal system $\K$ are obtained by adding the respective
(local) modal rules.

However, locality has a collateral side effect: more choices on the
application of rules. This may cause an explosion in the proof
space. In order to obtain a better control of proofs, we propose a
proof strategy based on \emph{blocks of applications} of modal rules.
The result is a notion of normal derivations in the linear nested
setting, which directly correspond to derivations in the standard
sequent setting.

Since we are interested in the connections to the standard sequent
framework, we concentrate on logics which have a standard sequent
calculus. Examples include normal modal logic $\K$ and extensions of
it, in particular the family of simply dependent multimodal
logics~\cite{Demri:2000}, as well as several non-normal modal logics,
i.e., standard extensions of \emph{classical modal
  logic}~\cite{Chellas:1980fk}. Notably, we obtain the first nested
sequent calculi for the logics of the \emph{modal tesseract} (see
Fig.~\ref{fig:modal-tesseract}).

Finally, while more expressive formalisms such as $\LNS$ enable
calculi for a broader class of logics, the greater bureaucracy makes
it harder to prove meta-logical properties, such as analyticity
itself.  Since a specific logic gives rise to specific sets of rules
in different calculi, it is important to determine whether there is a
\emph{general} methodology for determining/analysing such meta-level
properties. This is the role of \emph{logical frameworks} in proof
theory, where proof systems are adequately embedded into a meta-level
formal system so that object-level properties can be uniformly proven.
Since logical frameworks often come with automated procedures, the
meta-level machinery can be used for proving properties of the
embedded systems \emph{automatically}.
In~\cite{DBLP:journals/tcs/MillerP13} \emph{bipoles} and the
\emph{focusing proof strategy}~\cite{Andreoli:1992} in linear
logic~\cite{Girard:1987uq} were used in order to specify sequent
systems.  By interpreting object-level inference rules as meta-level
bipoles, focusing forces a one-to-one correspondence between the
application of rules and the derivation of formulae.  In this work, we
show that this bipole/focusing approach can be extended to linear
nested systems. Such specification allows for the proposal of a
general theorem prover (POULE available at
\url{http://subsell.logic.at/nestLL/}), parametric in the theory,
profiting from the modularity of the specified systems.
  
It should be noted that some preliminary results on linear nested
systems for various modal systems were presented
in~\cite{DBLP:conf/lpar/LellmannP15}. In the present paper we give
many more examples and refine several technical details. The new
contributions with respect to~\cite{DBLP:conf/lpar/LellmannP15} are:
(1) generalisation of the results on simply dependent bimodal logics
to large family of logics in Sec.~\ref{sec:simply-depend-modal}; (2)
introduction of modular linear nested sequent calculi for the
non-normal modal logics of the modal tesseract in
Sec.~\ref{sec:dosens-principle}; (3) definition of a notion of normal
forms for linear nested sequents, via the concept of \emph{modal block
  forms}; this allows for a modular way of translating modal sequent
into linear nested sequent systems; (4) automatic generation of
labelled systems for all the logics in the modal tesseract; and
finally (5) discussion on some other possible approaches for focusing
in modal systems, especially the ones proposed
in~\cite{DBLP:conf/fossacs/ChaudhuriMS16}.

The rest of the paper is organized as follows. In
Section~\ref{sec:nested} we introduce the concept of linear nested
sequents ($\LNS$). In Section~\ref{sec:examples} we show that the
linear nested sequent framework is a good formalism for a large class
of modal systems, showing non trivial extensions of multimodal $\K$ as
well as a large class of non-normal modal
logics. Section~\ref{sec:dosens-principle} also presents local systems
for non-normal logics, but by modifying the structural rules of the
system, instead of their logical rules. In both Sections we make use
of auxiliary structural operators. Since locality often entails less
efficient systems, in Section~\ref{sec:focused} we propose a notion of
``normal proofs'' in $\LNS$ derivations, hence showing how to reduce
the proof space and consequently optimize proof search. Since modal
connectives presented in this work are uniquely defined by the modal
rules, we can specify such rules as bipoles. We show the specification
process in Section~\ref{sec:labelled}, by first proposing labelled
sequent versions for $\LNS$ systems and then showing how to to
generate bipole clauses in linear logic which adequately correspond to
$\LNS$ modal rules.  Finally, in Section~\ref{sec:conc}, we conclude
by pointing out some future work.

%%% Local Variables: 
%%% mode: latex
%%% TeX-master: "TOCL"
%%% End: 

\section{Linear nested sequent systems}\label{sec:nested}
%!TEX root = TOCL.tex

As an intermediate between the efficiency of the ordinary sequent
framework and the expressiveness of the nested sequent
framework~\cite{Bull:1992,DBLP:journals/sLogica/Kashima94,Brunnler:2009kx,Poggiolesi:2009vn,Strasburger:2013}
we consider calculi in the \emph{linear nested sequent}
framework~\cite{Lellmann:2015lns}. This is essentially a reformulation
of Masini's 2-sequents~\cite{Masini:1992} in the nested sequent
framework, where the tree structure of nested sequents is restricted
to that of a line. The benefit is that this framework exhibits the
structure essential to obtain modular calculi, i.e., the nesting of
sequents, while retaining a very close connection to the ordinary
sequent framework and offering advantages in terms of efficiency. A
similar approach was followed with the $\mathsf{G}$-$\mathsf{CK}_n$
sequents for constructive modal logic of~\cite{Mendler:2011kxb} which
moreover also add some form of focusing to the linear structure.

In the following, we consider a \emph{sequent} to be a pair
$\Gamma \seq \Delta$ of multisets of formulae and adopt the standard
conventions and notations for formulae, multisets, and proof systems
(see e.g.~\cite{Negri:2011}).  A linear nested sequent then is simply
a finite list of sequents. As noted in~\cite{Lellmann:2015lns}, this
data structure matches exactly that of a \emph{history} in a backwards
proof search in an ordinary sequent calculus, a fact we will heavily
use in what follows.

\begin{definition}
The set $\LNS$ of \emph{linear nested sequents} is given recursively by:
\begin{enumerate}
  \item if  $\Gamma\seq \Delta$ is a sequent then $\Gamma\seq \Delta\in\LNS$
  \item if $\Gamma\seq \Delta$ is a sequent and $\mathcal{G}\in\LNS$
    then $\Gamma\seq\Delta\lns\mathcal{G}\in\LNS$.
\end{enumerate}
We will write $\mathcal{S}\{\Gamma\seq \Delta\}$ for denoting a
\emph{context} $\mathcal{G}\lns\Gamma\seq \Delta\lns\mathcal{H}$ where
each of $\mathcal{G},\mathcal{H}$ is a linear nested sequent or empty
(omitting the $\lns$ symbol in the latter case).  We call each sequent
in a linear nested sequent a \emph{component} and slightly abuse
notation and abbreviate ``linear nested sequent'' to $\LNS$.  The
standard interpretation for linear nested sequents for modal logic
$\K$ is given by:
\begin{align*}
  \iota_\Box(\Gamma \seq \Delta) \defs & \Land \Gamma \iimp \Lor \Delta\\
  \iota_\Box(\Gamma \seq \Delta \lns \mathcal{G}) \defs & \Land \Gamma \iimp
                                                    \Lor \Delta \lor\Box
                                                    \iota_\Box(\mathcal{G})
\end{align*}
As usual, we take a conjunction and disjunction over an empty multiset
to be $\top$ and $\bot$, respectively.
\end{definition}

Thus, the nesting operator $\lns$ of linear nested sequents is
interpreted as a structural connective for the modal box on the right
hand side of a sequent. Note that this is essentially the standard
interpretation of the brackets $\nes{.}$ of nested sequents using the
two-sided sequents of~\cite{Bull:1992} instead of the single-sided
formulation of~\cite{Brunnler:2009kx}. Since we only consider
\emph{linear} nested sequents, we use $\lns$ instead of iterated
brackets to increase readability.

\begin{example}\label{ex:lns-interpretation}
  Consider the logic $\K$.
    \begin{enumerate}
    \item The formula interpretation of the linear nested sequent
      $\;\seq A \lns A\seq \;$ is $\top \to A \lor \Box (A \to \bot)$
      which is equivalent to $A \lor \Box \neg A$.\label{ex:triv}
    \item The formula interpretation of the linear nested sequent
      $\Box A\seq \; \lns \; \seq \; \lns \;\seq A$ is \\
      $\Box A \to
      \bot \lor \Box( \top \to \bot \lor \Box(\top \to A))$ which is equivalent to $\Box A
      \to \Box \Box A$.\label{ex:density}
    \end{enumerate}
\end{example}

\begin{remark}
  It is worth noting that while the structure of a linear nested
  sequent as a list of ordinary sequents is the same as that of a
  \emph{hypersequent} (see, e.g.,~\cite{Avron:1996kx}), there is an
  important difference between the two frameworks. In virtually all
  hypersequent calculi the formula interpretation of a hypersequent is
  given by some form of disjunction. E.g., in the context of modal
  logics the standard formula interpretation of the hypersequent
  $\Gamma_1 \seq \Delta_1 \mid \Gamma_2 \seq \Delta_2 \mid
  \Gamma_3\seq \Delta_3$ would be given by
  $\Box (\Land \Gamma_1 \to \Lor \Delta_1) \lor \Box (\Land \Gamma_2
  \to \Lor \Delta_2) \lor \Box (\Land \Gamma_3 \to \Lor \Delta_3)$. In
  particular, every component of the hypersequent is interpreted
  uniformly under exactly one application of $\Box$. In contrast, the
  formula interpretation of the linear nested sequent
  $\Gamma_1 \seq \Delta_1 \lns \Gamma_2 \seq \Delta_2 \lns \Gamma_3
  \seq \Delta_3$ according to the interpretation $\iota_\Box$ from
  above is given by
  $\Land \Gamma_1 \to \Lor \Delta_1 \lor \Box \left( \Land \Gamma_2
    \to \Lor \Delta_2 \lor \Box \left(\Land \Gamma_3 \to \Lor
      \Delta_3\right)\right)$. Crucially, every component is
  interpreted under a number of modal operators which depends on its
  position in the linear nested sequent.  So the formula
  interpretations of the hypersequents $\;\seq A \mid A \seq \;$ and
  $\Box A \seq \; \mid \; \seq \; \mid \;\seq A$ corresponding to the
  linear nested sequents of Ex.~\ref{ex:lns-interpretation} would be
  given by the formulae $\Box(\top \to A) \lor \Box (A \to \bot)$ and
  $\Box( \Box A \to \bot) \lor \Box (\top \to \bot) \lor \Box(\top \to
  A)$ respectively, which are equivalent to the formulae
  $\Box A \lor \Box \neg A$ and
  $\Box \neg \Box A \lor \Box \bot \lor \Box A$ respectively. Clearly,
  these interpretations are rather different from the ones in the
  linear nested sequent framework.  Accordingly, virtually all
  hypersequent calculi contain a rule like the \emph{external exchange
    rule} which permits a reordering of the components. However, under
  the linear nested sequent formula interpretation and for the logics
  considered here this rule would not be sound. In line with this
  observation, linear nested sequent calculi have also been considered
  as hypersequent calculi without the external exchange rule under the
  name of \emph{non-commutative hypersequents}
  e.g.~in~\cite{Indrzejczak:2016} and in their tableaux version as
  \emph{path-hypertableaux} in~\cite{Ciabattoni:2000}.  A more
  detailed investigation on the connection between linear nested
  sequents and hypersequents is contained in~\cite{Lellmann:2015lns}.
\end{remark}

In this work we consider only modal logics based on classical
propositional logic, and we take the system $\LNS_\G$
(Fig.~\ref{fig:lns-G}) as our base calculus. The linear nested sequent
versions of the standard (internal) structural rules are given in
Fig.~\ref{fig:struct-rules}.

\begin{definition}
  For a system $\mathcal{C}$ of linear nested sequent rules, we define
  a \emph{derivation} to be a finite directed tree where each node is
  labelled with a linear nested sequent in such a way that the linear
  nested sequent associated to each node is obtained from the linear
  nested sequents associated to its immediate successors by an
  application of one of the rules from $\mathcal{C}$. In particular,
  each leaf of a derivation is labelled with the conclusion of an
  instance of a zero-premiss rule, i.e., one of the rules
  $\init, \bot_L, \top_R$. The label of the root of a derivation is
  also called the \emph{conclusion} of that derivation, and we say
  that a linear nested sequent $\mathcal{G}$ is \emph{derivable in the
    system $\mathcal{C}$}, in symbols
  $\entails_{\mathcal{C}} \mathcal{G}$, if there is a derivation in
  $\mathcal{C}$ with conclusion $\mathcal{G}$. The \emph{depth} of a
  derivation is the length of the longest branch in the underlying
  directed tree plus one. In the following we
  will denote by $\LNS_{\mathcal{L}}$ a linear nested sequent system
  for a logic $\mathcal{L}$ obtained by adding a certain set of
  rules for the modal operators to the system $\LNS_\G$. By
  $\LNS_{\mathcal{L}}\Con\W$ we denote the extension of the system
  $\LNS_{\mathcal{L}}$ with the structural rules of contraction and
  weakening from Fig.~\ref{fig:struct-rules}, where we abbreviate
  $\ConL,\ConR$ to $\Con$ and $\W_L,\W_R$ to $\W$.
\end{definition}

Observe that $\LNS_\G$ is the linear nested version of the well known
system $\Gtcp$ from~\cite{Troelstra:2000fj} plus explicit rules for
negation.  The reason for considering the structural rules explicitly
is that, while in the logical systems considered in
Section~\ref{sec:non-normal-modal} contraction and weakening are
admissible (see Lemmas~\ref{weakening} and~\ref{contraction}), some of
the systems in Sections~\ref{sec:simply-depend-modal}
and~\ref{sec:dosens-principle} are based on sequent calculi which
include explicit contraction and weakening.  As a side remark, it is
worth noticing that the approach presented here could be easily
adapted to having $\LKF$ \cite{liang09tcs} as the base logical system,
since such a decision would not alter the proof theory developed for
the modal connectives.  Choosing $\LNS_\G$ has the advantage that the
initial sequents are atomic, weakening, contraction and cut are
admissible and all propositional rules are invertible.
\begin{figure}[t]
  \hrule\smallskip
  \[
  \infer[\mathsf{init}]{\mathcal{S}\con{\Gamma, p \seq p, \Delta}}{}
  \quad
   \infer[\bot_L]{\mathcal{S}\con{\Gamma,\bot \seq \Delta}}{}
  \quad
  \infer[\top_R]{\mathcal{S}\con{\Gamma \seq \top, \Delta}}{}
  \quad
  \infer[\neg_L]{\mathcal{S}\con{\Gamma, \neg A \seq  \Delta}}{\mathcal{S}\con{\Gamma \seq A, \Delta}}
  \quad
    \infer[\neg_R]{\mathcal{S}\con{\Gamma \seq \neg A,  \Delta}}{\mathcal{S}\con{\Gamma,A \seq  \Delta}}
  \]
  \[
  \infer[\lor_L]{\mathcal{S} \con{\Gamma,  A \lor B \seq
      \Delta}}{\mathcal{S} \con{\Gamma,A \seq \Delta} &
    \mathcal{S}\con{\Gamma,B \seq  \Delta}}
  \qquad
  \infer[\land_L]{\mathcal{S}\con{\Gamma,A\land B \seq \Delta}}{\mathcal{S}
    \con{\Gamma, A, B \seq \Delta}}
  \qquad
  \infer[\iimp_L]{\mathcal{S}\con{\Gamma,A \iimp B \seq\Delta}
  }{\mathcal{S}\con{\Gamma,B \seq \Delta} & \mathcal{S}\con{\Gamma
      \seq A,\Delta}}
  \]
  \[
  \infer[\lor_R]{\mathcal{S}\con{\Gamma \seq A\lor B, \Delta}}{\mathcal{S}
    \con{\Gamma \seq A, B, \Delta}}
   \qquad
  \infer[\land_R]{\mathcal{S} \con{\Gamma \seq A \land B,
      \Delta}}{\mathcal{S} \con{\Gamma \seq A,\Delta} &
    \mathcal{S}\con{\Gamma \seq B, \Delta}}
  \qquad\infer[\iimp_R]{\mathcal{S}\con{\Gamma \seq A \iimp
      B,\Delta}}{\mathcal{S}\con{\Gamma,A \seq B,\Delta}}
  \]
 \hrule
  \caption{System $\LNS_{\G}$ for
    classical propositional logic. In the $\init$ rule, $p$ is atomic.
  }
  \label{fig:lns-G}
\end{figure}
\begin{figure}
  \hrule
  \[
    \infer[\ConL]{\mathcal{S} \con{\Gamma, A \seq \Delta}
    }
    {\mathcal{S} \con{\Gamma, A, A \seq \Delta}
    }
    \qquad
    \infer[\ConR]{\mathcal{S} \con{\Gamma \seq A, \Delta}
    }
    {\mathcal{S} \con{\Gamma \seq A, A, \Delta}
    }
    \qquad
    \infer[\W_L]{\mathcal{S} \con{\Gamma, A \seq
        \Delta}
    }
    {\mathcal{S}\con{\Gamma \seq \Delta}
    }
    \qquad
    \infer[\W_R]{\mathcal{S} \con{\Gamma \seq A, \Delta}
    }
    {\mathcal{S} \con{\Gamma \seq \Delta}
    }
  \]
  \hrule
  \caption{The structural rules of contraction and weakening.}
  \label{fig:struct-rules}
\end{figure}

Fig.~\ref{fig:lns-K} presents the modal rules for the linear nested
sequent calculus $\LNS_{\K}$ for $\K$, essentially a linear version of
the standard nested sequent calculus
from~\cite{Brunnler:2009kx,Poggiolesi:2009vn}. Thus, the calculus
$\LNS_\K$ contains the rules of $\LNS_\G$ together with the rules of
Fig.~\ref{fig:lns-K}.

Conceptually, the main point is that the sequent rule $\mathsf{k}$ is
split into the two rules $\Box_L$ and $\Box_R$, which permit to
simulate the sequent rule treating one formula at a time. While this
is one of the main features of nested sequent calculi and deep
inference in general~\cite{Guglielmi:2001}, being able to separate the
left/right behaviour of the modal connectives is the key to modularity
for nested and linear nested sequent
calculi~\cite{Strasburger:2013,Lellmann:2015lns}.
\begin{figure}[t]
  \hrule
  \smallskip
  \[
     \infer[\Box_L]{\mathcal{S}\con{\Gamma,\Box A \seq \Delta\lnes{\Sigma
          \seq \Pi}}}{\mathcal{S} \con{\Gamma \seq
        \Delta\lnes{\Sigma,A \seq \Pi}}}
           \qquad
           \infer[\Box_R]{\mathcal{G} \lnes{\Gamma\seq \Delta,\Box A}}{\mathcal{G}
      \lnes {\Gamma\seq\Delta\lnes{\;\seq A}}}
    \]
  \hrule
  \caption{The modal rules of the linear nested sequent calculus $\LNS_{\K}$ for $\K$.}
  \label{fig:lns-K}
\end{figure} 

Completeness of $\LNS_\K$ w.r.t.\ modal logic $\K$ is shown by
simulating a sequent derivation bottom-up in the last two components
of the linear nested sequents, marking applications of modal rules by
the nesting $\lns$ and simulating the $\mathsf{k}$-rule by a block of
$\Box_L$ and $\Box_R$ rules~\cite{Lellmann:2015lns}. Hence, an
application of $\mathsf{k}$ on a branch with history captured by the
$\LNS$ $\mathcal{G}$ is simulated by:
\[
\begin{array}{ccc}
\vcenter{
\infer*[\mathcal{G}]{}{
\infer[\mathtt{k}]
{\Gamma',\Box \Gamma\seq\Box A,\Delta
}
{\Gamma\seq A}
  }
  }
& \quad\rightsquigarrow \quad
& 
\vcenter{\infer[\Box_R]{\Gscr\lns\Gamma',\Box\Gamma\seq \Box A,\Delta}
{\infer=[\Box_ L]{\Gscr\lns\Gamma',\Box\Gamma\seq \Delta\ \lns\seq A}
{\deduce{\Gscr\lns\Gamma'\seq \Delta\lns\Gamma\seq A}{}}}}
\end{array}
\]
where the double line indicates multiple rule applications.  Observe
that this method relies on the view of linear nested sequents as
histories in proof search, where intuitively the modal rules mark a
transition or jump to a new state in a corresponding Kripke model. It
also simulates the propositional sequent rules in the \emph{rightmost}
component of the linear nested sequents.  However, while the principal
formulae of the sequent rule can now be handled separately, the modal
rules in the $\LNS$ system do not need to occur in a block
corresponding to one application of the sequent rule anymore. In fact,
one way of deriving the instance
$\Box (p \iimp q) \iimp (\Box p \iimp \Box q)$ of the normality axiom
for modal logic $\K$ is as follows.
\[
\infer=[\iimp_R]{\;\seq \Box (p \iimp q) \iimp (\Box p \iimp \Box q)
}
{\infer[\Box_R]{\Box (p \iimp q), \Box p \seq \Box q
  }
  {\infer[\Box_L]{\Box (p \iimp q),\Box p \seq \; \lns  \seq q}
  {\infer[\iimp_L]{\Box p \seq \; \lns p \iimp q \seq q
    }
    {\infer[\init]{\Box p \seq \; \lns q \seq q}{}
    &
    \infer[\Box_L]{\Box p \seq \; \lns \;\seq p,q
      }
      {\infer[\init]{\;\seq\;\lns p \seq p,q}{}
      }
    }
  }
}
}
\] 
Note that the propositional rule $\iimp_L$ is applied between two
modal rules. Hence there are many derivations in $\LNS_\K$ which are
not the result of simulating a derivation of the sequent calculus for
$\K$. Thus, while the linear nested sequent calculus $\LNS_\K$ has
conceptual advantages over the standard sequent calculus for $\K$, its
behaviour in terms of proof search is worse: there are many more
possible derivations with the same conclusion, when compared to the
sequent calculus. In Section~\ref{sec:labelled}, we will consider how
to restrict proof search to a smaller class of derivations, while
retaining the conceptual advantages of the framework.

%%% Local Variables: 
%%% mode: latex
%%% TeX-master: "TOCL"
%%% End: 

\section{Linear nested sequent systems and modality}\label{sec:examples}
%!TEX root = TOCL.tex

In~\cite{Lellmann:2015lns} the method of granulizing sequent rules
into linear nested sequent rules was applied to some basic modal
logics and to the multi-succedent calculus for intuitionistic logic.
In the following, we will considerably extend these results and show
that the linear nested sequent framework is a \emph{good} formalism for
a large class of modal systems. We first concentrate on non trivial
extensions of multimodal $\K$ (the so called simply dependent
multimodal logics) and then we show how to modularly extend the $\LNS$
approach also for handling non-normal modal logics.

%%% Local Variables: 
%%% mode: latex
%%% TeX-master: "TOCL"
%%% End: 

\subsection{Simply dependent multimodal logics}
\label{sec:simply-depend-modal}
%!TEX root = TOCL.tex

As a first example we consider multimodal logics with a simple
interaction between the modalities, called \emph{simply dependent
  multimodal logics}~\cite{Demri:2000}. The language for these logics
contains indexed modalities $\ibox{i}$ for indices $i$ from an index
set $N \subseteq \nat$ of natural numbers. In a Hilbert-style
presentation of these logics, the axioms are given by extensions of
the axioms of modal logic $\K$ for every modality $\ibox{i}$ together
with \emph{interaction axioms} of the form
$\ibox{i} A \iimp \ibox{j} A$. A simple example of such a logic is the
simply dependent bimodal logic $\K\T \oplus_\subseteq \Sf$
from~\cite{Demri:2000}, whose language contains the two modalities
$\ibox{1}$ and $\ibox{2}$. Its Hilbert-style axiomatisation is given
by the axioms and rules of classical propositional logic together with
the axioms and rules of modal logic $\K\T$ for the modality
$\ibox{1}$, i.e., axioms $\K$ and $\T$ and the rule $\mathsf{nec}$,
the $\Sf$ axioms for the modality $\ibox{2}$, i.e., axioms $\K,\T,\4$
and rule $\mathsf{nec}$, and the single interaction axiom
$\ibox{2} A \iimp \ibox{1} A$. Standard modal logics such as $\K$ or
extensions also can be seen as the trivial case of simply dependent
multimodal logics where the index set $N$ is a singleton. Other
examples include multimodal logics with a \emph{justified knowledge}
or ``any fool knows'' modality
from~\cite{Artemov:2006,McCarthy:1978}. Here and in the following we
will identify a logic with its set of theorems and write
$A \in \mathcal{L}$ if the formula $A$ is a theorem of logic
$\mathcal{L}$, i.e., derivable in the Hilbert-style system for
$\mathcal{L}$.

A general framework to describe simply dependent multimodal logics was
given in~\cite[Sec.~4]{Achilleos:2016}. There, such a logic is given
essentially by a triple $(N,\cless,F)$, where $N$ is a finite set of
natural numbers, $(N,\cless)$ is a partial order (i.e., transitive,
reflexive and antisymmetric), and $F$ is a mapping from $N$ to a set
$\mathfrak{L}$ of logics.

In the present work, we will take $\mathfrak{L}$ to be the set of
extensions of modal logic $\K$ with axioms from the set $\{\D,\T,\4\}$
(see Fig.~\ref{fig:modal-axioms}). The \emph{logic described by}
$(N,\cless,F)$ then has modalities $\ibox{i}$ for every $i \in N$,
with axioms for the modality $i$ given by the logic $F(i)$ and
interaction axioms $\ibox{j} A \iimp \ibox{i} A$ for every $i,j \in N$
with $i \cless j$. We write $\mathcal{L}_{(N,\cless,F)}$ for the logic
described by $(N,\cless,F)$.
\begin{figure}[t]
  \hrule
  \[
    \K\; \Box (A \iimp B) \iimp (\Box A \iimp \Box B)
    \qquad
    \vcenter{\infer[\mathsf{nec}]{\Box A}{A}}
    \qquad
    \D\; \neg \Box \bot
    \qquad
    \T\; \Box A \iimp A
    \qquad
    \4\; \Box A \iimp \Box \Box A
  \]
  \hrule
  \caption{
  Some modal axioms and rule $\mathsf{nec}$. Modal logic $\K$ contains
    the propositional tautologies, modus ponens, $\K$ and
    $\mathsf{nec}$.}\label{fig:modal-axioms}
\end{figure}

\begin{example}\label{ex:simply-dependent-logic}
The simply dependent bimodal
logic $\K\T\oplus_\subseteq\Sf$ is given by the description
$(N,\cless,F)$ with $N = \{ 1,2 \}$, and $F(1) = \K\T$, $F(2) =
\Sf$, where $\cless$ is given by $1\cless 1, 1 \cless 2, 2 \cless 2$.
\end{example}

The following definition extends the concept of frames to simply
dependent multimodal logic.  The notions of valuations, model and
truth in a world of the model are defined as usual (see,
e.g.,~\cite{Chellas:1980fk,Blackburn:2001fk}).
\begin{definition}
  Let $(N,\cless,F)$ be a description for a simply dependent
  multimodal logic. A {\em $(N,\cless,F)$-frame} is a tuple
  $(W,(R_i)_{i\in N})$ consisting of a set $W$ of \emph{worlds} and an
  \emph{accessibility relation} $R_i$ for every index $i \in N$, such
  that for all $i,j \in N$:
  \begin{itemize}
  \item If the logic $F(i)$ contains $\K\D$, then $R_i$ is serial.
  \item If the logic $F(i)$ contains $\K\T$, then $R_i$ is reflexive.
  \item If the logic $F(i)$ contains $\K\4$, then $R_i$ is transitive.
  \item If $i \cless j$, then $R_i \subseteq R_j$.
  \end{itemize}
\end{definition}

Since here we only consider simply dependent multimodal logics where
the different component logics are extensions of $\K$ with axioms from
$\{\D,\T,\4\}$, and since the interaction axioms are of a particularly
simple shape, standard results e.g.\ from Sahlqvist
theory~\cite[Thm.~4.42]{Blackburn:2001fk} immediately yield soundness
and completeness:

\begin{theorem}
  The modal logic given by the description $(N,\cless,F)$ is the logic
  of the class of $(N,\cless,F)$-frames, i.e., a formula is a
    theorem of the logic $\mathcal{L}_{(N,\cless,F)}$ iff it is valid in
    all $(N,\cless,F)$-frames.
\end{theorem}

By standard modal reasoning we immediately obtain the following lemma
stating upwards propagation of the modal axioms $\D$ and $\T$.

\begin{lemma}
  Let $(N,\cless,F)$ be a description of a simply dependent
  multimodal logic $\mathcal{L}$. Then for every $i \in N$:
  \begin{itemize}
  \item If $\K\D \subseteq F(i)$, then for every $j \in N$ with $i
    \cless j$, $\neg \ibox{j} \bot$ is also a theorem of $\mathcal{L}$.
  \item If $\K\T \subseteq F(i)$, then for every $j \in N$ with $i
    \cless j$,  $\ibox{j} A \iimp A$ is also a theorem of $\mathcal{L}$.
  \end{itemize}
\end{lemma}

\begin{proof}
  By closing the axioms under the interaction axioms $\ibox{j} A \iimp
  \ibox{i} A$ for $i \cless j$. Alternatively, this can be seen from
  the semantical characterisation.
\end{proof}

Hence we may assume, without loss of generality, that for any
description $(N,\cless,F)$ and any $i \in N$, if $\K\D\subseteq F(i)$
(or $\K\T\subseteq F(i)$), then for every $j \in N$ with $i\cless j$
we have $\K\D \subseteq F(j)$ (resp. $\K\T \subseteq F(j)$). As a more
economic notation we also write $\upset{i}$ for the \emph{upset} of
the index $i$, i.e., the set $\{j \in N : i \cless j\}$. Furthermore,
in light of the comments above we extend this notation to the sets
$\upset[\mathsf{Ax}]{i} \defs \{ j \in N : i \cless j, \,
\K\mathsf{Ax}\subseteq F(j)\}$ and
$\upset[\neg\mathsf{Ax}]{i} \defs \{ j \in N : i \cless j, \,
\K\mathsf{Ax}\not\subseteq F(j)\}$ where $\mathsf{Ax}$ is any of the
axioms $\D,\T,\4$. Thus e.g.\ the set $\upset[\neg\4]{i}$ is the set
of indices $j$ with $i\cless j$ such that $\K\4 \not\subseteq F(j)$,
i.e., the logic $F(j)$ does not derive the transitivity axiom $\4$.

The next step is to obtain cut-free sequent calculi for logics of this
family. In order to obtain cut-free completeness, i.e., completeness
without the cut rule, we also need the set of transitive logics to be
upwards closed. Formally:

\begin{definition}
  A description $(N,\cless,F)$ is \emph{transitive-closed} if for
  every $i,j\in N$ with $i\cless j$, if $\K\4 \subseteq F(i)$ then
  $\K\4\subseteq F(j)$.
\end{definition}

Using the method of \emph{cut elimination by saturation} for sequent
rules with restrictions on the context, as developed
in~\cite{Lellmann:2013fk,Lellmann:2013}, it is then reasonably
straightforward to construct cut-free sequent calculi for simply
dependent multimodal logics given by a transitive-closed
description. Since the actual construction of the sequent rules is not
central to this paper, we will omit the details.  The resulting modal
rules and rule sets are given in Fig.~\ref{fig:simply-dep-multimodal}.

\begin{definition}\label{def:seq-simplpy-dep}
  The restriction of the propositional calculus $\LNS_\G$ from
  Fig.~\ref{fig:lns-G} to sequents is denoted by $\G$. If
  $(N,\cless,F)$ is a description for a simply dependent multimodal
  logic, then $\G_{(N,\cless,F)}$ is the sequent calculus extending
  the propositional calculus $\G$ with the modal rules
  $\mathcal{R}_{(N,\cless,F)}$ according to
  Fig.~\ref{fig:simply-dep-multimodal}.
\end{definition}
\begin{figure}[t]
  \hrule
  \[
    \infer[\mathsf{k}_i]{\Omega,\{\ibox{j} \Gamma_j, \ibox{j} \Sigma_j
      : j\in \upset[\4]{i}\}, \{ \ibox{j} \Sigma_j :
      j\in\upset[\neg\4]{i}\} \seq \ibox{i} A, \Xi}{\{\ibox{j}
      \Gamma_j, \Sigma_j : j \in \upset[\4]{i}\},
      \{ \Sigma_j : j \in \upset[\neg\4]{i}\} \seq A}
  \]
  \[
    \infer[\mathsf{d}_i]{\Omega,\{\ibox{j} \Gamma_j, \ibox{j} \Sigma_j
      : j \in \upset[\4]{i}\}, \{ \ibox{j} \Sigma_j : j\in
      \upset[\neg\4]{i}\} \seq \Xi}{\{\ibox{j} \Gamma_j,
      \Sigma_j : j \in \upset[\4]{i}\}, \{ \Sigma_j : j \in
      \upset[\neg\4]{i}\} \seq \;} \qquad\quad
    \infer[\mathsf{t}_i]{\Omega, \{\ibox{j} \Sigma_j : j \in
      \upset{i}\} \seq \Xi }{\Omega, \{\Sigma_j : j
      \in \upset{i} \} \seq \Xi}
  \]
  \[
    \mathcal{R}_{(N,\cless,F)} \defs \{ \mathsf{k}_i : i \in N \} \cup
    \{ \mathsf{d}_i : i \in N, \K\D \subseteq F(i)\} \cup \{
    \mathsf{t}_i : i \in N, \K\T \subseteq F(i) \}
  \]
  \hrule
  \caption{The modal sequent rules for the simply dependent multimodal
    logic given by the transitive-closed description $(N, \cless, F)$.}
  \label{fig:simply-dep-multimodal}
\end{figure}

The intuition behind the rules perhaps is best obtained by considering
an example:

\begin{example}
  Continuing our Ex.~\ref{ex:simply-dependent-logic}, in the case of
  the logic $\K\T \oplus_\subseteq \Sf$ we have
  $\K\D \subseteq \K\T \subseteq F(i)$ for $i = 1,2$ and
  $\K\4 \subseteq F(2)$ but $\K\4 \not\subseteq F(1)$. Hence we have
  $\upset{1} = \{ 1,2 \}, \upset{2} = \{2\}$ and furthermore
  $\upset[\4]{1} = \{ 2 \}, \upset[\neg \4]{1}= \{ 1\}$ as well as
  $\upset[\4]{2} = \{ 2\}, \upset[\neg \4]{2} = \emptyset$. Thus the
  sequent calculus $\G_{\K\T \oplus_\subseteq \Sf}$ for this logic
  contains the following modal rules, obtained as specific instances
  of the rules given in Fig.~\ref{fig:simply-dep-multimodal}: 
\[
  \infer[\mathsf{k}_1]{\Omega,\ibox{2} \Gamma_2,\ibox{2} \Sigma_2,\ibox{1}
    \Sigma_1 \seq \ibox{1} A,\Xi}{\ibox{2} \Gamma_2,
    \Sigma_2,\Sigma_1 \seq A} \quad
  \infer[\mathsf{k}_2]{\Omega,\ibox{2} \Gamma_2, \ibox{2} \Sigma_2 \seq \ibox{2}
    A,\Xi}{\ibox{2} \Gamma_2, \Sigma_2 \seq A}
\]
\[
  \infer[\mathsf{d}_1]{\Omega, \ibox{2} \Gamma_2, \ibox{2}\Sigma_2,
    \ibox{1}\Sigma_1 \seq \Theta}{\ibox{2}\Gamma_2,
    \Sigma_2, \Sigma_1 \seq \;}
  \quad
  \infer[\mathsf{d}_2]{\Omega,\ibox{2} \Gamma_2, \ibox{2}\Sigma_2
    \seq \Theta}{\ibox{2} \Gamma_2, \Sigma_2
    \seq\;}
  \quad
  \infer[\mathsf{t}_1]{\Omega,\ibox{1} \Sigma_1 \seq
    \Theta}{\Omega,\Sigma_1\seq\Theta}
  \quad
  \infer[\mathsf{t}_2]{\Omega,\ibox{2} \Sigma_2 \seq
    \Theta}{\Omega,\Sigma_2\seq\Theta}
  \quad
\]
Note that this rule set could still be simplified in two ways. First
we observe that the rules $\mathsf{d}_1, \mathsf{d}_2$ are derivable
using rules $\mathsf{t}_1,\mathsf{t}_2$ and weakening, and hence could
be omitted from the rule set. For the sake of a uniform presentation
we decided to keep them. Further, the rules $\mathsf{t}_i$ could be
restricted to the more traditional version with only a single
principal formula.  We chose the current formulation because this is
the form of the rules which is obtained directly from the construction
and facilitates a more uniform cut elimination proof.
\end{example}

\begin{remark}\label{rem:modularity-simply-dep}
  The previous example also serves to illustrate the issues with
  \emph{modularity} in the sequent framework: suppose we wanted to
  obtain a cut-free sequent system for the logic
  $\K\T \oplus_{\subseteq} \K\T$ instead of the logic
  $\K\T \oplus_{\subseteq} \Sf$, i.e., we only drop the axiom
  $\ibox{2} A \to \ibox{2}\ibox{2} A$ from the Hilbert-style
  system. Then to obtain the calculus
  $\G_{\K\T \oplus_{\subseteq}\K\T}$ from the calculus
  $\G_{\K\T \oplus_{\subseteq}\Sf}$ above we would need to drop the
  context formulae $\ibox{2}\Gamma_2$ from \emph{each} of the rules
  $\mathsf{k}_1, \mathsf{k}_2, \mathsf{d}_1, \mathsf{d}_2$. This is no
  accident: in general, adding or deleting one axiom from the
  Hilbert-style presentation of a logic requires heavy modifications
  of the corresponding sequent calculi, which need to take \emph{all}
  rules of that calculus into account.
\end{remark}

While soundness and completeness of the calculi $\G_{(N,\cless,F)}\Con\W$
follow directly from the construction, for later reference and the reader not familiar
with the general construction we state them explicitly and
briefly sketch 
the proofs.

\begin{theorem}
  Let $(N,\cless,F)$ be a transitive-closed description of a simply
  dependent multimodal logic. Then the sequent calculus
  $\G_{(N,\cless,F)}\Con\W$ is sound with respect to this logic, i.e.,
  for every formula $A$ we have that $A\in \mathcal{L}_{(N,\cless,F)}$
  if $\entails_{\G_{(N,\cless,F)}\Con\W} \;\seq A$.
\end{theorem}

\begin{proof}
  We use the fact that the logic given by the description is also
  characterised by frames $(W, (R_i)_{i \in N})$ where for $i \in N$
  the accessibility relation $R_i$ satisfies the properties stipulated
  by the logic $F(i)$ (i.e., is serial if $\K\D \subseteq F(i)$,
  reflexive if $\K\T \subseteq F(i)$ and transitive if
  $\K\4 \subseteq F(i)$), and where for every $i,j \in N$ with
  $i \cless j$ we have $R_i \subseteq R_j$. Then it is easy to show
  that all the modal rules preserve validity by showing that if the
  negation of the conclusion is satisfiable in such a frame, then so
  is the premiss. Since the zero-premiss rules are valid, i.e., the
  negation of their formula interpretation is not satisfiable in any
  frame, from this we obtain the soundness statement by induction on
  the depth of the derivation. As an example we fix a description
  $(N,\cless,F)$ and consider the following application of the rule
  $\mathsf{d}_i$ for an index $i\in N$ such that $F(i)$ is serial.
  \[
    \infer[\mathsf{d}_i]{\Omega,\{\ibox{j} \Gamma_j, \ibox{j} \Sigma_j
      : j\in\upset[\4]{i}\}, \{ \ibox{j} \Sigma_j :
      j\in\upset[\neg\4]{i}\} \seq \Xi}{\{\ibox{j} \Gamma_j,
      \Sigma_j : j\in\upset[\4]{i}\}, \{ \Sigma_j : j \in
      \upset[\neg\4]{i}\} \seq \;}
  \]
  If the negation of the conclusion of this rule is satisfiable in a
  $(N,\cless,F)$-model $\mathfrak{M} = (W,(R_i)_{i\in N},\sigma)$,
  then we have a world $w \in W$ such that
  \begin{equation}
    \mathfrak{M},w \forces
    \Land \Omega \land \Land_{j \in \upset[\4]{i}}\left(\Land
      \ibox{j}\Gamma_j\land \Land \ibox{j}\Sigma_j\right) \land
    \Land_{j \in \upset[\neg\4]{i}}\Land \ibox{j}\Sigma_j
    \land \neg\Lor \Xi\;.\label{eq:soundness1}
  \end{equation}
  Since $F(i)$ is serial, there is a world $v \in W$ with $wR_i v$,
  and since $i\cless j$ implies $R_i \subseteq R_j$ for all $i,j \in
  N$, for this $v$ we also have $wR_jv$ for every $j$ with $i\cless
  j$. Hence using~\eqref{eq:soundness1} and transitivity of the
  relations $R_j$ for $j$ with $\K\4 \subseteq F(j)$ we obtain
  \[
    \mathfrak{M},v \forces \Land_{j\in \upset[\4]{i}}\left(\Land
      \ibox{j}\Gamma_j \land \Land
      \Sigma_j\right) \land \Land_{j \in \upset[\neg\4]{i}}\Land
    \Sigma_j\;.
  \]
  Hence the negation of the interpretation of the premiss of this rule
  application is satisfied in $v$. The reasoning for the remaining
  rules is similar.
  \end{proof}

\begin{theorem}
  Let $(N,\cless,F)$ be a transitive-closed description of a simply
  dependent multimodal logic. Then the sequent calculus
  $\G_{(N,\cless,F)}\Con\W$ is (cut-free) complete with respect
  to this logic, i.e., for every formula $A$ we have that
  $A\in \mathcal{L}_{(N,\cless,F)}$ only if
  $\entails_{\G_{(N,\cless,F)}\Con\W} \;\seq A$.
\end{theorem}

\begin{proof}
  As usual, completeness of the system without a cut rule is shown by
  first showing that every axiom and rule of the Hilbert-style system
  for the logic $\mathcal{L}_{(N,\cless,F)}$ can be simulated in the
  system $\G_{(N,\cless,F)}$ together with the \emph{multicut rule},
  i.e., the following rule, where $n,m \geq 1$ and $A^k$ is an
  abbreviation for the multiset $A,\dots,A$ containing exactly $k$
  copies of the formula $A$:
  \[
    \infer[\mcut]{\Gamma, \Sigma \seq \Delta, \Pi}{\Gamma \seq
      \Delta, A^n & A^m, \Sigma \seq \Pi}
  \]
  We call the formula $A$ in this application of the multicut rule the
  \emph{cut formula}.  Deriving all the axioms for
  $\mathcal{L}_{(N,\cless,F)}$ in the system $\G_{(N,\cless,F)}\Con\W$
  is straightforward: the axioms for the logics $F(i)$ are derived as
  in the monomodal case, and the interaction axioms for $i \cless j$
  are obtained by a single application of the rule $\mathsf{k}_i$. As
  usual, the rule of modus ponens is simulated by applications of the
  cut rule.

  In the second step we show in a rather standard way that the
  multicut rule can be eliminated from derivations in the system
  $\G_{(N,\cless,F)}\Con\W$ (this also follows directly from checking
  that the system $\G_{(N,\cless,F)}\Con\W$ satisfies the general
  criteria for cut elimination
  in~\cite{Lellmann:2013fk,Lellmann:2013}).  As usual, the proof is by
  double induction on the \emph{complexity} of the cut formula, i.e.,
  the number of symbols in the cut formula, and the sum of the depths
  of the derivations of the two premisses of the application of the
  multicut rule. Applications of the multicut rule are then pushed
  upwards into the derivations of the premisses of that application,
  until in both of the latter at least one occurrence of the cut
  formula is introduced by the last applied rule, at which point the
  complexity of the cut formula is reduced. Since the reasoning for
  the different cases is rather standard, here we only consider two
  exemplary cases. See, e.g.,~\cite[Sec.~4.1.9]{Troelstra:2000fj} for
  the reasoning in the propositional cases.
  
  As a first example, consider the multicut below, with applications
  of rules based on a description such that $i \cless j,k$ and
  $\ell \cless i,m,n$ with $\K\4 \subseteq F(j),F(m)$ but with $\K\4$
  not contained in the other logics.
  \[
  \infer[\mcut]{\Omega,\ibox{j}\Gamma_j, \ibox{j}\Sigma_j, \ibox{k}\Sigma_k,
    \Upsilon, \ibox{i}\Sigma_i, \ibox{m}\Gamma_m, \ibox{m}\Sigma_m,
    \ibox{i}\Sigma_i, \ibox{n}\Sigma_n \seq \Xi, \Pi
  }
  {\infer[\mathsf{k}_i]{\Omega, \ibox{j} \Gamma_j, \ibox{j} \Sigma_j, \ibox{k}
      \Sigma_k, \ibox{i} \Sigma_i \seq \ibox{i} A, \Xi, \ibox{i} A^{n-1}
    }
    {\ibox{j} \Gamma_j,  \Sigma_j, \Sigma_k, \Sigma_i
      \seq  A
    }
    &
  \infer[\mathsf{d}_\ell]{ \Upsilon, \ibox{i} A^{t-s}, \ibox{m} \Gamma_m, \ibox{m}
    \Sigma_m, \ibox{i} A^s, \ibox{i}
    \Sigma_i, \ibox{n} \Sigma_n \seq \Pi
    }
    {\ibox{m} \Gamma_m, \Sigma_m, A^s, \Sigma_i, \Sigma_n \seq \;
    }
  }
  \]
  As usual, the multicut is replaced with a multicut on the formula of
  lower complexity $A$ as follows.
  \[
  \infer[\mathsf{d}_\ell]{\Omega, \Upsilon, \ibox{j}\Gamma_j, \ibox{j}\Sigma_j, \ibox{k}\Sigma_k,
    \ibox{i}\Sigma_i, \ibox{m}\Gamma_m, \ibox{m}\Sigma_m,
    \ibox{i}\Sigma_i, \ibox{n}\Sigma_n \seq \Xi, \Pi
  }
  {\infer[\mcut]{\ibox{j}\Gamma_j, \Sigma_j, \Sigma_k,
    \Sigma_i, \ibox{m}\Gamma_m, \Sigma_m,
    \Sigma_i, \Sigma_n \seq \;
    }
    {\ibox{j} \Gamma_j, \Sigma_j, \Sigma_k, \Sigma_i
      \seq  A
    &\quad
    \ibox{m} \Gamma_m, \Sigma_m, A^s, \Sigma_i, \Sigma_n \seq \;
    }
  }
  \]
  Crucially, since the relation $\cless$ is transitive, we know that
  $\ell \cless j,k$ as well, which renders the application of the rule
  $\mathsf{d}_\ell$ at the bottom permissible.

  As a second example, consider the multicut with cut formula
  $\ibox{i} A$ below, based on a description $(N,\cless, F)$ with
  $k \cless i \cless j$, such that $\K\4 \subseteq F(i), F(j)$ and
  $\K\D \subseteq F(k)$.
  \[
  \infer[\mcut]{\Omega, \ibox{j} \Gamma_j, \ibox{j} \Sigma_j,
    \Upsilon, \ibox{i}
    \Gamma_i, \ibox{i} \Sigma_i, \ibox{k} \Gamma_k \seq \Xi, \Pi
  }
  {\infer[\mathsf{k}_i]{ \Omega, \ibox{j} \Gamma_j, \ibox{j}\Sigma_j
      \seq \ibox{i} A, \Xi, \ibox{i} A^{n-1}
    }
    {\ibox{j} \Gamma_j, \Sigma_j \seq A
    }
  &
  \infer[\mathsf{d}_k]{\Upsilon, \ibox{i}A^{\ell-s}, \ibox{i} \Gamma_i, \ibox{i}
    A^s, \ibox{i} \Sigma_i, \ibox{k} \Gamma_k \seq \Pi
    }
    {\ibox{i} A^{\ell-s}, \ibox{i} \Gamma_i, A^s, \Sigma_i,
      \Gamma_k \seq \;
    }
  }
  \]
  This multicut is replaced by two multicuts, an application
    of $\mathsf{d}_k$ and contractions as follows:
  \[
  \infer=[\Con]{\Omega, \Upsilon, \ibox{j} \Gamma_j, \ibox{j}
    \Sigma_j, \ibox{i} \Gamma_i, \ibox{i} \Sigma_i, \ibox{k} \Gamma_k
    \seq \Xi, \Pi
  }
  {
  \infer[\mathsf{d}_k]{\Omega, \Upsilon, \ibox{j} \Gamma_j, \ibox{j} \Sigma_j, \ibox{j}
    \Gamma_j, \ibox{j} \Sigma_j, \ibox{i} \Gamma_i, \ibox{i} \Sigma_i,
    \ibox{k}\Gamma_k \seq \Xi, \Pi
  }
  {\infer[\mcut]{\ibox{j} \Gamma_j, \Sigma_j, \ibox{j} \Gamma_j, \ibox{j} \Sigma_j, \ibox{i}
      \Gamma_i, \Sigma_i,\Gamma_k \seq \;
    }
    {\ibox{j} \Gamma_j, \Sigma_j \seq A
      &
      \infer[\mcut]{\ibox{j} \Gamma_j, \ibox{j} \Sigma_j, \ibox{i}
      \Gamma_i, A^s, \Sigma_i, \Gamma_k \seq \;
      }
      {\infer[\mathsf{k}_i]{\ibox{j} \Gamma_j, \ibox{j} \Sigma_j \seq
          \ibox{i} A
        }
        {\ibox{j} \Gamma_j, \Sigma_j \seq A
        }
      &\quad
      \ibox{i} A^{\ell-s}, \ibox{i} \Gamma_i, A^s, \Sigma_i,
      \Gamma_k \seq \;
      }
    }
  }
  }
  \]
  The upper multicut is then eliminated using the (inner) induction
  hypothesis on the sum of the depths of the derivations of its
  premisses, the lower one is eliminated using the (outer) induction
  hypothesis on the complexity of the cut formula.  Note that for this
  transformation to work it is crucial that the logic $F(j)$ is also
  transitive, i.e., that the description $(N,\cless,F)$ is transitive
  closed, since otherwise we would not be able to apply the rule
  $\mathsf{d}_k$ with boxed context formulae $\ibox{j} \Gamma_j$. A
  similar situation occurs if the application of $\mathsf{d}_k$ is
  replaced with an application of $\mathsf{k}_k$ with an additional
  principal formula on the right.

  The general cases of the above examples as well as the remaining
  cases are treated similarly.
\end{proof}

In order to convert the resulting sequent systems into $\LNS$ systems,
we need to modify the linear nested setting to account for all the
different non-invertible right rules. For this, given a description
$(N,\cless,F)$ we introduce nesting operators $\lns[i]$ for every
$i \in N$, and change the interpretation so that they are interpreted
by the corresponding modality:
\begin{align*}
  \iota(\Gamma \seq \Delta) & \defs \Land \Gamma \iimp \Lor \Delta \\
  \iota(\Gamma \seq \Delta \lns[i] \mathcal{H}) & \defs \Land \Gamma
  \iimp \Lor \Delta \lor \ibox{i} \,\iota(\mathcal{H})
\end{align*}
The modal sequent rules of $\G_{(N,\cless,F)}$ are then decomposed
into the modal linear nested sequent rules shown in
Fig.~\ref{fig:lns-rules-simply-dep-mult}. The propositional rules are
those of $\LNS_\G$ (Fig.~\ref{fig:lns-G}). We call the resulting
calculus $\LNS_{(N,\cless,F)}$.  The intuition behind the rules is
that an application of the sequent rule $\mathsf{k}_i$ is decomposed
into an application of the rule ${\ibox{i}}_R$ followed by
applications of ${\ibox{ij}}_L$ to unpack the principal formulae of
the sequent rule, and applications of the rule $\4_{ij}$ to move the
boxed context formulae into the next component.
\begin{figure}[t]
  \hrule
  \[
  \infer[{\ibox{ij}}_L]{\mathcal{S}\con{ \Gamma, \ibox{i} A \seq \Delta
    \lns[j]\Sigma \seq \Pi}}{\mathcal{S}\con{\Gamma \seq \Delta
    \lns[j]\Sigma,A \seq \Pi}}
  \qquad
  \infer[{\ibox{i}}_R]{\mathcal{G}\lns[k]  \Gamma \seq \Delta, \ibox{i}
    A}{\mathcal{G}\lns[k]  \Gamma \seq \Delta\, \lns[i] \;\seq A}
  \]
  \[
  \infer[\mathsf{d}_{ij}]{\mathcal{G}\lns[k] \Gamma, \ibox{i} A \seq
    \Delta}{\mathcal{G} \lns[k] \Gamma \seq \Delta \lns[j] A \seq \;}
  \qquad
  \infer[\mathsf{t}_i]{\mathcal{S} \con{ \Gamma, \ibox{i} A \seq
      \Delta}}{\mathcal{S} \con{ \Gamma, A \seq \Delta}}
  \qquad
  \infer[{\4_{ij}}]{\mathcal{S} \con{ \Gamma, \ibox{i} A \seq \Delta
    \lns[j]\Sigma \seq \Pi}}{\mathcal{S} \con{\Gamma \seq \Delta
    \lns[j]\Sigma,\ibox{i}A \seq \Pi}}
  \]
  \begin{align*}
  \mathcal{R}_{(N,\cless, F)} \defs &\; \{ {\ibox{i}}_R : i \in N\} \cup \{
  {\ibox{ij}}_L : i,j \in N, i \in \upset{j}\} \cup \{ \mathsf{d}_{ij} :
  i,j \in N, i \in \upset[\D]{j}\}\\
  & 
   \cup \{ \mathsf{t}_i : i \in N, \K\T \subseteq F(i)\} 
  \cup \{
  {\4_{ij}} : i,j \in N, i\in \upset[\4]{j}\}
  \end{align*}
  \hrule
    \caption{The linear nested sequent rules for the simply dependent
    multimodal logic given by the description $(N,\cless,F)$.}
  \label{fig:lns-rules-simply-dep-mult}
\end{figure}

\begin{example}
  The linear nested sequent calculus for the logic $\K\T
  \oplus_\subseteq \Sf$ contains the $\LNS$ rules $\ibox{11}_L,
  \ibox{21}_L, \ibox{22}_L, \ibox{1}_R, \ibox{2}_R, \mathsf{d}_{11},
  \mathsf{d}_{21}, \mathsf{d}_{22}, \mathsf{t}_1,
  \mathsf{t}_2$, ${\4_{21}}$, and $\4_{22}$.
\end{example}

\begin{remark}
  The previous example illustrates the added modularity of the linear
  nested sequent approach: if, as in
  Rem.~\ref{rem:modularity-simply-dep} we wanted to obtain a linear
  nested sequent calculus for the logic $\K\T\oplus_\subseteq \K\T$
  from the calculus for $\K\T\oplus_\subseteq \Sf$ above, we would
  only need to delete the rules $\4_{21}$ and $\4_{22}$ from the rule
  set, keeping all other rules the same. This is in stark contrast to
  the modification of almost all modal rules required in the ordinary
  sequent setting. Note however, that we have modularity, and indeed
  completeness, only for transitive-closed descriptions. I.e., we
  would not be able to obtain a calculus for the logic
  $\Sf \oplus_\subseteq \K\T$, since it is not given by a 
  transitive-closed description.
\end{remark}

\begin{theorem}
  If $(N,\cless,F)$ is a transitive-closed description of a simply
  dependent multimodal logic, then $\LNS_{(N,\cless,F)}\Con\W$ is
  sound and complete for the logic given by $(N,\cless,F)$, i.e., for
  every formula $A$ we have that $A \in \mathcal{L}_{(N,\cless,F)}$ if
  and only if $\entails_{\LNS_{(N,\cless,F)}\Con\W} \;\seq A$.
\end{theorem}

\begin{proof}
  For soundness, i.e., the ``if'' statement, we show that whenever the
  negation of the interpretation of the conclusion of a rule from
  $\LNS_{(N,\cless,F)}\Con\W$ is satisfiable in a
  $(N,\cless,F)$-frame, then so is the negation of the interpretation
  of at least one of its premisses. This makes essential use of the
  fact that in such frames we have $R_i \subseteq R_j$ whenever
  $i \cless j$.  For completeness, i.e., the ``only if'' statement, we
  again simulate the sequent rules in the last components, i.e., we
  translate a sequent derivation in $\G_{(N,\cless, F)}\Con\W$
  bottom-up into a linear nested sequent derivation in
  $\LNS_{(N,\cless,F)}$, simulating propositional sequent rules by
  their linear nested sequent counterparts, and modal sequent rules by
  a number of applications of the corresponding linear nested sequent
  rules. E.g., an application of the modal sequent rule $\mathsf{d}_i$
  with history (i.e., trace to the conclusion of the sequent
  derivation) captured by the linear nested sequent $\mathcal{G}$ is
  simulated as follows (assuming that $k \in \upset[\neg \4]{i}$):
  \[
  \infer*[\mathcal{G}]{}{
    \infer[\mathsf{d}_i]{\Omega,\{\ibox{j} \Gamma_j, \ibox{j} \Sigma_j
      : j \in \upset[\4]{i}\}, \{ \ibox{j} \Sigma_j : j\in
      \upset[\neg\4]{i}\}, \ibox{k} A \seq \Xi
    }
    {\{\ibox{j} \Gamma_j, 
    \Sigma_j : j \in \upset[\4]{i}\}, \{ \Sigma_j : j \in
    \upset[\neg\4]{i}\}, A \seq \;
    }
  }
  \]
\[
  \qquad\leadsto\qquad
  \vcenter{
  \infer[\mathsf{d}_{ki}]{\mathcal{G} \lns \Omega,\{\ibox{j} \Gamma_j, \ibox{j} \Sigma_j
      : j \in \upset[\4]{i}\}, \{ \ibox{j} \Sigma_j : j\in
      \upset[\neg\4]{i}\}, \ibox{k} A \seq \Xi
  }
  {\infer=[\ibox{ji}_L]{\mathcal{G} \lns \Omega,\{\ibox{j} \Gamma_j, \ibox{j} \Sigma_j
      : j \in \upset[\4]{i}\}, \{ \ibox{j} \Sigma_j : j\in
      \upset[\neg\4]{i}\} \seq \Xi \lns[i] A \seq \;
    }
    {\infer=[\ibox{ji}_L]{\mathcal{G} \lns \Omega,\{\ibox{j} \Gamma_j, \ibox{j} \Sigma_j
      : j \in \upset[\4]{i}\}  \seq \Xi \lns[i] \{ \Sigma_j : j\in
      \upset[\neg\4]{i}\}, A \seq \;
        }
        {\infer=[\4_{ji}]{\mathcal{G} \lns \Omega,\{\ibox{j} \Gamma_j
      : j \in \upset[\4]{i}\},  \seq \Xi \lns[i] \{ \Sigma_j :
      j \in \upset[\4]{i} \}, \{ \Sigma_j : j\in
      \upset[\neg\4]{i}\}, A \seq \;
          }
          {\mathcal{G} \lns \Omega,  \seq \Xi \lns[i] \{\ibox{j} \Gamma_j, \Sigma_j
      : j \in \upset[\4]{i}\},\{ \Sigma_j :
      j \in \upset[\4]{i} \}, \{ \Sigma_j : j\in
      \upset[\neg\4]{i}\}, A \seq \;
          }
        }
    }
  }
  }
\]
The remaining modal rules are simulated in a similar way.
\end{proof}

Note that the proof of completeness via simulation of the sequent
calculus in the last component actually shows a slightly stronger
statement, i.e., completeness for a variant of the calculus where the
rules are restricted so they only manipulate the last components. More
precisely:

\begin{definition}\label{end-active}\label{def:end-active}
  An application of a linear nested sequent rule is \emph{end-active}
  if the rightmost components of the premisses are active and the only
  active components (in premiss and conclusion) are the two rightmost
  ones. The \emph{end-active variant} of a $\LNS$ calculus is the
  calculus with the rules restricted to end-active applications.
\end{definition}

\begin{example}
  The application of the rule $\land_L$ below left is end-active, the one
  below right is not, since the rightmost component is not active.
  \[
    \infer[\land_L]{\mathcal{G} \lns \Gamma, A\land B \seq \Delta
    }
    {\mathcal{G} \lns \Gamma, A, B \seq \Delta
    }
    \qquad\qquad
    \infer[\land_L]{\mathcal{G} \lns \Gamma, A\land B \seq \Delta \lns
      \Sigma \seq \Pi
    }
    {\mathcal{G} \lns \Gamma, A, B \seq \Delta\lns \Sigma \seq \Pi
    }
  \]
  Applications of the modal rules in the $\LNS$ calculi for non-normal
  modal logics considered in this paper (see next section) are always
  end-active. An application of the modal rule
  \[
    \infer[\Box_L]{\mathcal{S}\con{\Gamma, \Box A \seq \Delta \lns \Sigma \seq
      \Pi}}{\mathcal{S} \con{\Gamma \seq \Delta \lns \Sigma, A \seq \Pi}}
  \]
  is end-active only if $\Sigma \seq \Pi$ is the rightmost component.
\end{example}

\begin{corollary}
  If $(N,\cless,F)$ is a transitive-closed description of a simply
  dependent multimodal logic, then the end-active variant of
  $\LNS_{(N,\cless,F)}\Con\W$ is sound and complete for the logic
  given by $(N,\cless,F)$, i.e., for every formula $A$ we have
  $A \in \mathcal{L}_{(N,\cless,F)}$ if and only if $\;\seq A$ is
  derivable in the end-active variant of $\LNS_{(N,\cless,F)}\Con\W$.
\end{corollary}

\begin{proof}
  Soundness, i.e., the ``if'' statement, follows immediately from
  soundness for the full calculus. For completeness, i.e., the ``only
  if'' statement, observe that the sequent rules are simulated in the
  last component, i.e., by end-active applications of the linear
  nested sequent rules.
\end{proof}

The fact that we can restrict the linear nested calculi to their
end-active variants will be exploited in Section~\ref{sec:focused} for
reducing the search space in proof search.

The example of simply dependent multimodal logics shows another
conceptual advantage of $\LNS$ calculi over standard sequent calculi:
for more involved sequent calculi such as the ones in
Fig.~\ref{fig:simply-dep-multimodal} the decomposition of the sequent
rules into their different components tends to make the corresponding
$\LNS$ calculi (Fig.~\ref{fig:lns-rules-simply-dep-mult}) a lot more
readable.  Of course the previous theorem also shows that the obvious
adaption of this calculus to the full nested sequent setting
of~\cite{Brunnler:2009kx,Poggiolesi:2009vn} is sound and cut-free
complete for the corresponding logic.

%%% Local Variables:
%%% mode: latex
%%% TeX-master: "TOCL"
%%% End:

\subsection{Non-normal modal logics}
\label{sec:non-normal-modal}
%!TEX root = TOCL.tex

The same ideas also yield $\LNS$ calculi for some \emph{non-normal}
modal logics, i.e., modal logics that are not extensions of modal
logic $\K$ (see~\cite{Chellas:1980fk} for an introduction).  The
calculi themselves are of independent interest since, to the best of
our knowledge, nested sequent calculi for the logics below have not
been considered before in the literature.  The most basic non-normal
logic, \emph{classical modal logic} $\mathsf{E}$, is given
Hilbert-style by extending the axioms and rules for classical
propositional logic with only the \emph{congruence rule}
$(\mathsf{E})$ for the $\Box$ connective
\[
\infer[(\mathsf{E})]{\Box A \iimp \Box B}{ A \iimp B & B
  \iimp A}
\]
which allows exchanging logically
equivalent formulae under the modality. Some of the better known
extensions of this logic are formulated by the addition of axioms from
the list below left.
\begin{center}
\begin{minipage}[c]{.4\textwidth}
  $\mathsf{M} \quad\Box (A\land B) \iimp (\Box A \land \Box B)$\medskip\\
               $\mathsf{C} \quad (\Box A \land \Box B) \iimp \Box (A \land B)$\medskip\\
               $\mathsf{N} \quad \Box \top$
\end{minipage}
\begin{minipage}[T]{.4\textwidth}
  \begin{tikzpicture}
    \node at (0,0) (E) {\textsf{E}};
    \node at (0,1.5) (M) {\textsf{M}};
    \node at (2,.5) (N) {\textsf{EN}};
    \node at (-1,.5) (C) {\textsf{EC}};
    \node at (2,2) (MN) {\textsf{MN}};
    \node at (-1,2) (MC) {\textsf{MC}};
    \node at (1,1) (NC) {\textsf{ECN}};
    \node at (1,2.5) (MNC) {$\mathsf{MNC} = \K$};
    \draw (E) -- (M) -- (MN) -- (N) -- (E) -- (C) -- (MC) -- (MNC) --
    (NC) -- (C);
    \draw (N) -- (NC);
    \draw (M) -- (MC);
    \draw (MN) -- (MNC);
  \end{tikzpicture}
\end{minipage}
\end{center}
Together, these extensions form what could be termed the
\emph{classical cube} (above right). Note that the extension of
$\mathsf{E}$ with all the axioms $\mathsf{M},\mathsf{C},\mathsf{N}$ is
normal modal logic~$\K$.  Fig.~\ref{fig:non-normal-seq} shows the
modal rules of the standard cut-free sequent calculi for these logics,
where in addition weakening is embedded in the conclusion.  Extensions
of $\mathsf{E}$ are written by concatenating the names of the axioms,
and in presence of the monotonicity axiom $\mathsf{M}$, sometimes the
initial $\mathsf{E}$ is dropped. E.g., the logic
$\mathsf{EMC} = \mathsf{MC}$ is the extension of $\mathsf{E}$ with
axioms $\mathsf{M}$ and $\mathsf{C}$. Its sequent calculus
$\G_{\mathsf{MC}}$ is given by the standard propositional rules of
$\G$ (see Def.~\ref{def:seq-simplpy-dep}) together with the rule
$(\mathsf{E})$ as well as the rules $(\mathsf{M}n)$ for $n \geq 1$.
For all of the logics apart from $\mathsf{EN}$ and $\mathsf{ECN}$
these calculi were given in~\cite{Lavendhomme:2000}, the one for
$\mathsf{EN}$ is considered in~\cite{Indrzejczak:2011}, the remaining
calculus is an easy extension.  It is not too difficult to show
admissibility of weakening and contraction in these calculi. However,
since the calculi originally were considered in \emph{op.~cit.} for
sequents based on sets instead of multisets, and in preparation for
later results, we mostly consider their extensions with the structural
rules. E.g., we consider $\G_{\mathsf{EM}}\Con\W$ instead of
$\G_{\mathsf{EM}}$.
\begin{figure}[t]
  \hrule\medskip
  \[
  \infer[(\mathsf{E})]{\Gamma,\Box A \seq \Box B,\Delta}{A \seq B & B \seq A}
  \qquad
  \infer[(\mathsf{M})]{\Gamma,\Box A \seq \Box B,\Delta}{A \seq B}
  \qquad
  \infer[(\mathsf{N})]{\Gamma \seq \Box A,\Delta}{ \;\seq A}
  \]
  \[
  \infer[(\mathsf{E}n)]{\Gamma,\Box A_1,\dots,\Box A_n \seq \Box
    B,\Delta}{A_1,\dots,A_n \seq B & B \seq A_1 & \cdots & B \seq A_n}
  \qquad
  \infer[(\mathsf{M}n)]{\Gamma,\Box A_1,\dots,\Box A_n \seq \Box
    B,\Delta}{A_1,\dots,A_n \seq B }
  \]
  \[
  \begin{array}{l@{\quad}l@{\qquad}l@{\quad}l@{\qquad}l@{\quad}l@{\qquad}l@{\quad}l}
    \G_{\mathsf{E}} & \{\, (\mathsf{E})\,\}
    &\G_{\mathsf{EN}} & \{ \, (\mathsf{E}), (\mathsf{N})\, \}
    &\G_{\mathsf{EC}} & \{\, (\mathsf{E}n) : n \geq 1\,\}
    &\G_{\mathsf{ECN}} & \{ \,(\mathsf{E}n) : n \geq 1 \,\} \cup \{
    (\mathsf{N}) \}\\
    \G_{\mathsf{M}} & \{\, % (\mathsf{E}),
                      (\mathsf{M})\,\}
    &\G_{\mathsf{MN}} & \{\, (\mathsf{M}), (\mathsf{N})\,\}
    &\G_{\mathsf{MC}} & \{\, (\mathsf{M}n) : n \geq 1\,\}
    &\G_{\mathsf{MCN}} & \{\, (\mathsf{M}n) : n \geq 0\,\}
  \end{array}
  \]
  \hrule
  \caption{Sequent rules and calculi for some non-normal modal logics}
  \label{fig:non-normal-seq}
\end{figure}

In order to construct linear nested calculi for these logics, again we
would like to decompose the sequent rules from
Fig.~\ref{fig:non-normal-seq} into their different
components. However, there are two complications compared to the case
of normal modal logics: we need a mechanism for capturing the fact
that e.g.\ in the rule $(\mathsf{M})$ exactly one boxed formula is
introduced on the left hand side; and we need a way of handling
multiple premisses of rules such as $(\mathsf{E})$ and
$(\mathsf{E}n)$. We solve the first problem by introducing an
auxiliary nesting operator $\lnse$ to capture a state where a sequent
rule has been {\em partly processed}, i.e., where the simulation of
the sequent rule is still unfinished. The intuition behind this is
that in this ``partly processed'' state, only other $\LNS$ rules
continuing or eventually finishing the simulation of the original
sequent rule can be applied. In contrast, the intuition for the
original nesting $\lns$ is that the simulation of a rule is
finished. We restrict the occurrence of $\lnse$ to the end of the
structures.

To solve the problem of multiple premisses, we make the nesting
operator $\lnse$ \emph{binary}, which permits the storage of more
information about the premisses. In particular, we can now store the
two ``directions'' of implications given, e.g., in the premisses of
rule $(\mathsf{E})$. Linear nested sequents for classical non normal
modal logics are then given by:
\begin{figure}
  \hrule
  \[
  \infer[\Box_R^\e]{\mathcal{G} \lns \Gamma \seq \Box B,
    \Delta}{\mathcal{G} \lns \Gamma \seq \Delta \lnse( \;\seq B; B
    \seq \;)}
  \quad
  \infer[\Box_L^\e]{\mathcal{G} \lns \Gamma,\Box A \seq
    \Delta\lnse(\Sigma \seq \Pi; \Omega \seq \Theta)}{\mathcal{G} \lns \Gamma \seq
    \Delta\lns \Sigma, A \seq \Pi & \mathcal{G} \lns \Gamma \seq
    \Delta\lns  \Omega \seq A, \Theta}
  \]
  \[
    \infer[\mathsf{N}]{\mathcal{G} \lns \Gamma \seq \Box B,
    \Delta}{\mathcal{G} \lns \Gamma \seq \Delta \lns \;\seq B}
  \qquad\qquad
  \infer[\mathsf{M}]{\mathcal{G} \lnse(\Sigma \seq \Pi; \Omega \seq
    \Theta)}{\mathcal{G} \lnse(\Sigma \seq \Pi; \Omega, \bot \seq
    \Theta)}
  \]
  \[
  \infer[\mathsf{C}]{\mathcal{G} \lns \Gamma, \Box A \seq \Delta
    \lnse(\Sigma \seq \Pi; \Omega \seq \Theta)}{\mathcal{G} \lns \Gamma \seq \Delta
    \lnse(\Sigma, A \seq \Pi; \Omega \seq \Theta) &
    \mathcal{G} \lns \Gamma \seq \Delta \lns \Omega \seq A, \Theta}
  \]
  \[
  \LNS_{\mathsf{E}\mathcal{A}}: \quad \{ \Box_R, \Box_L \} \cup \mathcal{A}
  \quad \text{for }\mathcal{A} \subseteq \{ \mathsf{N},\mathsf{M},\mathsf{C}\}
  \]
  \hrule
  \caption{Linear nested sequent calculi for non-normal modal logics}
  \label{fig:lns-classical-modal-logics}
\end{figure}
\[
\LNS_\e ::= \Gamma \seq \Delta \mid \Gamma \seq \Delta \lnse
(\Sigma \seq \Pi; \Omega \seq \Theta) \mid \Gamma \seq \Delta \lns \LNS_\e
\]
Fig.~\ref{fig:lns-classical-modal-logics} shows the modal rules for
these logics. For a logic $\mathsf{E} \mathcal{A}$ with
$\mathcal{A} \subseteq \{ \mathsf{N}, \M, \C\}$ the calculus
$\LNS_{\mathsf{E}\mathcal{A}}$ then contains the corresponding modal
rules along with the propositional rules of $\LNS_\G$
(Fig.~\ref{fig:lns-G}) with the restriction that they are not applied
inside the nesting $\lnse$ .  To keep the presentation simple, we
slightly abuse notation and write e.g.~$\mathsf{M}$ both for the axiom
and the corresponding rule.

\begin{remark}
  At first sight the linear nested sequent rules of
  Fig.~\ref{fig:lns-classical-modal-logics} might look more
  complicated than the ordinary sequent rules of
  Fig.~\ref{fig:non-normal-seq}. Moreover, it could be argued that the
  latter systems could be considered to be modular, since, e.g., to obtain a
  calculus for the logic $\M\C$ from the logic $\mathsf{E}\C$ it would
  be sufficient to simply \emph{add} the rules
  $\{ (\M_n) : n \geq 1 \}$ to the system for $\mathsf{E}\C$, so it
  might not be clear immediately what we have gained by moving to the
  nested sequent setting. However, it is worth noting that, for every
  extension of $\mathsf{E}\C$, the sequent systems of
  Fig.~\ref{fig:non-normal-seq} consist of an infinite number of
  rules, where in particular every rule has a different number of
  principal formulae. In contrast, the linear nested sequent systems
  of Fig.~\ref{fig:lns-classical-modal-logics} each consist of a
  finite number of rules, where each rule has at most one principal
  formula. Moreover, each of the axioms
  $\mathsf{N}, \mathsf{M}, \mathsf{C}$ corresponds to exactly one
  additional rule in the linear nested sequent setting. Both of these
  properties are philosophically highly relevant, e.g., from the point
  of view of \emph{proof-theoretic semantics}. For more details on
  this, see, e.g., the discussion in~\cite[Sec.~1.2 and~1.3]{Wansing:2002fk}.
\end{remark}

\begin{theorem}[Completeness]
  The linear nested sequent calculi of
  Fig.~\ref{fig:lns-classical-modal-logics} are complete w.r.t.\ the
  corresponding logics, i.e., if $A \in \mathsf{E}\mathcal{A}$, then
  $\entails_{\LNS_{\mathsf{E}\mathcal{A}}\Con\W} \;\seq A$.
\end{theorem}

\begin{proof}
  Again the proof is via simulation of the sequent calculi. An
  application of the rule $(\mathsf{E}n)$ is simulated by the
  derivation
\[
  \infer[\Box_R^\e]{\mathcal{G} \lns \Gamma,\Box A_1,\dots, \Box A_n
    \seq \Box B,\Delta
  }
  {\infer[\mathsf{C}]{\mathcal{G} \lns \Gamma, \Box A_1,\dots, \Box
      A_n \seq  \Delta\lnse( \;\seq B ; B \seq \;)
    }
    {\infer*[]{\mathcal{G} \lns  \Gamma,\Box A_2,\dots, \Box A_n \seq
        \Delta \lnse( A_1\seq B ; B \seq \;)
      }
      {\infer[\Box_L^\e]{\mathcal{G} \lns \Gamma,\Box A_n \seq \Delta
          \lnse(A_1,\dots, A_{n-1} \seq B ; B \seq \;)
        }
        {\mathcal{G} \lns \Gamma \seq \Delta
          \lns A_1,\dots, A_n \seq B
        &\quad
        \mathcal{G} \lns  \Gamma\seq \Delta
          \lns B \seq A_n
        }
      }
    &
    \mathcal{G} \lns  \Gamma,\Box A_2,\dots, \Box A_n\seq \Delta \lns
    B \seq A_1
    }
  }
\]
The case of $n=1$ gives the simulation of the rule $(\mathsf{E})$. The
sequent rule $\mathsf{N}$ is simulated directly by the $\LNS$ rule
$\mathsf{N}$. In the monotone case the simulations are essentially the
same, but after creating the new nesting using the $\Box_R^\e$ rule we
first apply the rule $\M$ to make all the premisses for the
``backwards direction'' trivially derivable. The sequent rule
$(\mathsf{M}n)$ then is simulated by
\[
\infer[\Box_R^\e]{\mathcal{G} \lns \Gamma, \Box A_1,\dots, \Box A_n
  \seq \Box B, \Delta
  }
  {\infer[\mathsf{M}]{\mathcal{G} \lns \Gamma, \Box A_1,\dots, \Box A_n
    \seq \Delta \lnse (\;\seq B; B \seq \;)
    }
    {\infer[\mathsf{C}]{\mathcal{G} \lns \Gamma, \Box A_1,\dots, \Box A_n
        \seq \Delta \lnse (\;\seq B; B, \bot \seq \;)
      }
      {\infer*[]{\mathcal{G} \lns \Gamma, \Box A_2, \dots, \Box A_n \seq
          \Delta \lnse(A_1 \seq B; B,\bot \seq \;)
        }
        {\infer[\Box_L^\e]{\mathcal{G} \lns \Gamma, \Box A_n \seq \Delta
            \lnse (A_1,\dots, A_{n-1} \seq B; B,\bot \seq \;)
          }
          {\mathcal{G} \lns \Gamma \seq \Delta \lns A_1, \dots, A_n \seq
            B
            &
            \infer[\bot_L]{\mathcal{G} \lns \Gamma \seq \Delta \lns B,\bot \seq A_n}{}
          }
        }
      &
      \infer[\bot_L]{\mathcal{G} \lns \Gamma, \Box A_2,\dots, \Box A_n
        \seq \Delta \lns B, \bot \seq \;}{}
      }
    }
  }
\]
Again, the case of $n=1$ gives the simulation of the rule $(\mathsf{M})$.
\end{proof}

As in the case of simply dependent multimodal logics, the proof of
completeness via simulation of the sequent rules in the last component
also shows completeness of the end-active variants of the calculi.

\begin{corollary}
  The end-active variants of the linear nested sequent calculi of
  Fig.~\ref{fig:lns-classical-modal-logics} are complete w.r.t.\ the
  corresponding logics, i.e., if $A \in \mathsf{E}\mathcal{A}$, then
  $\;\seq A$ is derivable in the end-active variant of
  $\LNS_{\mathsf{E}\mathcal{A}}\Con\W$.
\end{corollary}

For showing soundness of such calculi we need a different method,
though. This is due to the fact that, unlike for normal modal logics,
there is no clear formula interpretation for linear nested sequents
for non-normal modal logics.  However, since the propositional rules
cannot be applied inside the auxiliary nesting $\lnse$\;, the modal
rules can only occur in blocks which can be seen as a macro-rule
corresponding to a modal sequent rule.  In addition, we will show by a
permutation-of-rules argument that it is possible to restrict the
propositional rules to end-active applications (see
Def.~\ref{def:end-active}). Soundness of the full calculus then
follows from soundness of the end-active variant, which is shown by
translating derivations back into derivations in the corresponding
sequent calculus.

Similar to the argument for \emph{levelled derivations}
in~\cite[p.~241, Prop.~2]{Masini:1992}, the following Lemmata show
that the propositional rules can be restricted to be end-active. The
first step is to show invertibility of the general forms of the
propositional rules. Since all our calculi include the contraction
rule we show this in a slightly more general form.

\begin{definition}
  If $\Gamma_1 \seq \Delta_1 \lns \dots\ \lns \Gamma_n \seq \Delta_n$
  is a linear nested sequent, then the \emph{level} of the occurrences
  of formulae in $\Gamma_i, \Delta_i$ is $i$.
\end{definition}

In all the results stated in this section, we will assume that
$\mathcal{A} \subseteq \{ \mathsf{N},\M,\C\}$.

\begin{lemma}[Admissibility of Weakening]\label{weakening}
  The weakening rules $\W_L, \W_R$ are depth-preserving admissible in
  $\LNS_{\mathsf{E}\mathcal{A}}$ and
  $\LNS_{\mathsf{E}\mathcal{A}}\Con$, i.e., if there is a derivation
  $\mathcal{D}$ of $\mathcal{S}\con{\Gamma \seq \Delta}$ with depth at
  most $n$, then there are derivations $\mathcal{D}_1$ and
  $\mathcal{D}_2$ of $\mathcal{S} \con{\Gamma, A \seq \Delta}$ and
  $\mathcal{S} \con{\Gamma \seq A, \Delta}$ respectively with depth at
  most $n$ in the same system. Moreover, if the level of the active
  components of every rule application in $\mathcal{D}$ is at least
  $k$, then the same holds for $\mathcal{D}_1$ and $\mathcal{D}_2$.
\end{lemma}

\begin{proof}
  As usual by induction on the depth of the derivation: applications
  of weakening are permuted upwards over every rule until they are
  absorbed by the initial sequents. Since this does not change the
  structure of the derivation and in particular does not introduce any
  new rule applications, the depth of the derivation and the minimal
  level of the active components of the rule applications is
  preserved.
\end{proof}

\begin{lemma}[Multi-invertibility of the propositional rules]\label{lem:m-invertibility}
  The non-end-active versions of the propositional rules are
  \emph{m-invertible} in $\LNS_{\mathsf{E}\mathcal{A}}\Con$, i.e., for
  every $n \geq 1$ we have:
  \begin{enumerate}
  \item If
    $\entails_{\LNS_{\mathsf{E}\mathcal{A}}\Con} \mathcal{S} \con{\Gamma,
      (\neg A)^n \seq \Delta}$, then also
    $\entails_{\LNS_{\mathsf{E}\mathcal{A}}\Con} \mathcal{S} \con{\Gamma
      \seq A^{n+k}, \Delta}$ for some $k \geq 0$.
  \item If
    $\entails_{\LNS_{\mathsf{E}\mathcal{A}}\Con} \mathcal{S} \con{\Gamma
      \seq (\neg A)^n, \Delta}$, then also
    $\entails_{\LNS_{\mathsf{E}\mathcal{A}}\Con} \mathcal{S} \con{\Gamma,
      A^{n+k} \seq \Delta}$ for some $k \geq 0$.
  \item If
    $\entails_{\LNS_{\mathsf{E}\mathcal{A}}\Con} \mathcal{S} \con{\Gamma,
      (A\iimp B)^n \seq \Delta}$, then also
    $\entails_{\LNS_{\mathsf{E}\mathcal{A}}\Con} \mathcal{S} \con{\Gamma,
      B^{n+k} \seq \Delta}$ and \\
    $\entails_{\LNS_{\mathsf{E}\mathcal{A}}\Con} \mathcal{S} \con{\Gamma
      \seq A^{n+\ell}, \Delta}$ for some $k,\ell \geq 0$.
  \item If
    $\entails_{\LNS_{\mathsf{E}\mathcal{A}}\Con} \mathcal{S} \con{\Gamma
      \seq (A\iimp B)^n, \Delta}$, then also
    $\entails_{\LNS_{\mathsf{E}\mathcal{A}}\Con} \mathcal{S} \con{\Gamma,
      A^{n+k} \seq B^{n+\ell}, \Delta}$ for some $k,\ell \geq 0$.
  \item If
    $\entails_{\LNS_{\mathsf{E}\mathcal{A}}\Con} \mathcal{S} \con{\Gamma,
      (A\lor B)^n \seq \Delta}$, then also
    $\entails_{\LNS_{\mathsf{E}\mathcal{A}}\Con} \mathcal{S} \con{\Gamma,
      A^{n+k} \seq \Delta}$ and \\
    $\entails_{\LNS_{\mathsf{E}\mathcal{A}}\Con} \mathcal{S}
    \con{\Gamma,B^{n+\ell} \seq \Delta}$ for some $k,\ell \geq 0$.
  \item If
    $\entails_{\LNS_{\mathsf{E}\mathcal{A}}\Con} \mathcal{S} \con{\Gamma
      \seq (A\lor B)^n, \Delta}$, then also
    $\entails_{\LNS_{\mathsf{E}\mathcal{A}}\Con} \mathcal{S} \con{\Gamma
      \seq A^{n+k}, B^{n+\ell}, \Delta}$ for some $k,\ell \geq 0$.
  \item If
    $\entails_{\LNS_{\mathsf{E}\mathcal{A}}\Con} \mathcal{S} \con{\Gamma,
      (A\land B)^n \seq \Delta}$, then also
    $\entails_{\LNS_{\mathsf{E}\mathcal{A}}\Con} \mathcal{S} \con{\Gamma,
      A^{n+k}, B^{n+\ell} \seq \Delta}$ for some $k,\ell \geq 0$.
  \item If
    $\entails_{\LNS_{\mathsf{E}\mathcal{A}}\Con} \mathcal{S} \con{\Gamma
      \seq (A\land B)^n, \Delta}$, then also
    $\entails_{\LNS_{\mathsf{E}\mathcal{A}}\Con} \mathcal{S} \con{\Gamma
      \seq A^{n+k}, \Delta}$ and \\
    $\entails_{\LNS_{\mathsf{E}\mathcal{A}}\Con} \mathcal{S} \con{\Gamma
      \seq B^{n+\ell}, \Delta}$ for some $k,\ell \geq 0$.
  \end{enumerate}
  Moreover, both the depth of the derivation and the minimal level of
  the active components of rule applications are preserved.
\end{lemma}

\begin{proof}
  By induction on the depth of the derivation, distinguishing cases
  according to the last applied rule. E.g., for the rule $\iimp_R$ we
  have: if $\mathcal{S}\con{\Gamma \seq (A \iimp B)^n, \Delta}$ is an
  initial sequent or the conclusion of one of the rules $\bot_L$ or
  $\top_R$, then so is the linear nested sequent
  $\mathcal{S}\con{\Gamma, A^{n+k} \seq B^{n+\ell}, \Delta}$ for any
  $k,\ell \geq 0$. If the last applied rule was not a contraction
  rule, we apply the induction hypothesis to its premiss(es), followed
  by the same rule. E.g., if the last applied rule was the rule $\C$,
  and the component containing $(A \iimp B)^n$ is the penultimate one,
  we have a derivation ending in
  \[
    \infer[\C]{\mathcal{G} \lns \Gamma', \Box C \seq (A \iimp B)^n
      \Delta \lnse (\Sigma \seq \Pi; \Omega \seq \Theta)
    }
    {\infer*{\mathcal{G} \lns \Gamma' \seq (A \iimp B)^n, \Delta \lnse (\Sigma,
      C \seq \Pi; \Omega \seq \Theta)}{} \quad& \infer*{\mathcal{G} \lns \Gamma' \seq (A
      \iimp B)^n, \Delta \lns \Omega \seq C, \Theta}{}
    }
  \]
  Using the induction hypothesis, for some $i,j,k,\ell$ we obtain
  derivations of
  $\mathcal{G} \lns \Gamma', A^{n+i} \seq B^{n+j}, \Delta \lnse
  (\Sigma, C \seq \Pi; \Omega \seq \Theta) $ and
  ${\cal G} \lns \Gamma', A^{n+k} \seq B^{n+\ell}, \Delta \lns \Omega
  \seq C, \Theta$ and admissibility of weakening
  (Lemma~\ref{weakening}) followed by an application of $\C$ yields
  the desired
  $\mathcal{G} \lns \Gamma', \Box C, A^{n+\max\{i,k\}} \seq
  B^{n+\max\{j,\ell\}}, \Delta \lnse (\Sigma \seq \Pi; \Omega \seq
  \Theta)$. Finally, if the last applied rule was the contraction
  rule, we simply apply the induction hypothesis to its premiss. E.g.,
  if the contracted formula is $A \iimp B$ and we have a derivation
  ending in
  \[
    \infer[\ConR]{\mathcal{S}\con{\Gamma \seq (A \iimp B)^n, \Delta}
    }
    {\mathcal{S}\con{\Gamma \seq (A \iimp B)^{n+1}, \Delta}
    }
  \]
  we use the induction hypothesis to obtain $\mathcal{S}\con{\Gamma, A^{n+1+k}
    \seq B^{n+1+\ell}, \Delta}$ for some $k,\ell \geq 0$. 
\end{proof}

Of course, setting $n = 1$ in the statement of the previous lemma and
(possibly) applying a number of contractions to the result recovers
standard invertibility of the propositional rules, albeit not the
depth-preserving version.

Using this we first obtain soundness of the full calculus with
contraction with respect to the end-active variant.

\begin{lemma}\label{lem:end-activeness-LNSEA}
  If a linear nested sequent $\Gamma \seq \Delta$ is derivable in
  $\LNS_{\mathsf{E}\mathcal{A}}\Con$, then it is derivable in the
  end-active variant of $\LNS_{\mathsf{E}\mathcal{A}}\Con$.
\end{lemma}

\begin{proof}
  Due to the nature of the modal rules it is clear that in a
  derivation only applications of the propositional rules and
  contraction can violate the end-activeness condition. We then
  successively transform a derivation of $\Gamma \seq \Delta$ into an
  end-active derivation as follows. Take the bottom-most block of
  modal rules such that there is an application of a propositional
  rule or contraction above it with level of the active component
  smaller than the maximal level of the active components in the modal
  block. Since the modal rules only apply to formulae in the last
  component, all such applications of propositional rules introduce a
  propositional connective which in the conclusion of the modal block
  is not under a modality. Using multi-invertibility of the
  propositional connectives (Lemma~\ref{lem:m-invertibility}) we
  replace every such formula in the conclusion of the modal block by
  its constituents, possibly with multiplicity more than one. E.g, if
  the conclusion of the modal block has the form
  \[
    \infer[\Box_R^\e]{\mathcal{G} \lns \Gamma \seq A \iimp B, \Delta \lns
      \mathcal{H} \lns \Sigma \seq \Box C, \Pi
    }
    {\infer*{\mathcal{G} \lns \Gamma \seq A \iimp B, \Delta \lns
      \mathcal{H} \lns \Sigma \seq \Pi \lnse (\;\seq C; C \seq \;)}{}
    }
  \]
  with the formula $A\iimp B$ introduced above the modal block, using
  m-invertibility we obtain
  \[
    \infer[]{\mathcal{G} \lns \Gamma, A^{1+k} \seq B^{1+\ell}, \Delta
      \lns \mathcal{H} \lns \Sigma \seq \Box C, \Pi
    }
    {\infer*[\mathcal{D}]{\mathcal{G} \lns \Gamma, A^{1+k} \seq B^{1+\ell}, \Delta
      \lns \mathcal{H} \lns \Sigma \seq \Pi \lnse (\;\seq C; C \seq \;)
      }
      {
      }
    }
  \]
  Then we delete every application of the contraction rule with active
  component of level smaller than the maximal level of the active
  components in the modal block from the derivation, possibly using
  Lemma~\ref{weakening} to ensure that the contexts in two-premiss
  rules are the same. From the proof of
  Lemma~\ref{lem:m-invertibility} it can be seen that afterwards the
  minimal level of the active components in rule applications in the
  derivation up to the conclusion of the modal block is at least the
  maximal level of the active components in the modal block
  itself. Finally, we use end-active applications of contraction to
  remove unwanted duplicates followed by end-active applications of
  the propositional rules to reintroduce the propositional connectives
  in the right place, i.e., when the component containing the
  constituent formulae is the last one. Since the conclusion of the
  original derivation contained only a single component, this is
  always possible.
\end{proof}

From this we obtain soundness of the full calculus by first
translating derivations into derivations in the end-active variant,
then into derivations in the corresponding sequent calculus:

\begin{theorem}[Soundness]
  If a sequent $\Gamma \seq \Delta$ is derivable in
  $\LNS_{\mathsf{E}\mathcal{A}}\Con$ for
  $\mathcal{A}\subseteq \{\mathsf{N},\mathsf{M},\mathsf{C}\}$, then it
  is derivable in the corresponding sequent calculus
  $\G_{\mathsf{E}\mathcal{A}}\Con$. Hence if
  $\entails_{\LNS_{\mathsf{E}\mathcal{A}}\Con\W}\;\seq A$, then
  $A \in \mathsf{E}\mathcal{A}$.
\end{theorem}

\begin{proof}
  From the previous lemma we obtain that if a sequent
  $\Gamma \seq \Delta$ is derivable in
  $\LNS_{\mathsf{E}\mathcal{A}}\Con$, then it is derivable in the
  end-active variant of $\LNS_{\mathsf{E}\mathcal{A}}\Con$.  A
  derivation of the latter form then is translated into a
  $\G_{\mathsf{E}\mathcal{A}}\Con$ derivation, discarding everything
  apart from the last component of the linear nested sequents, and
  translating blocks of modal rules into the corresponding modal
  sequent rules. E.g., a block consisting of an application of
  $\Box_L^\e$ followed by $n$ applications of $\mathsf{C}$ and an
  application of $\Box_R^\e$ is translated into an application of the
  rule $(\mathsf{E}n)$. In the monotone case we use the fact that the
  rule $\M$ permutes down over the rule $\C$, i.e., a modal block
  \[
  \infer[\Box_R^\e]{\mathcal{G} \lns \Gamma, \Box A_1,\dots, \Box A_n
    \seq \Box B, \Delta
  }
  {\infer[\C]{\mathcal{G} \lns \Gamma, \Box A_1,\dots, \Box A_n
    \seq \Delta \lnse (\;\seq B; B \seq \;)
    }
    {\infer*{\mathcal{G} \lns \Gamma, \Box A_2,\dots, \Box A_n
    \seq \Delta \lnse (A_1\seq B; B \seq \;)
      }
      {\infer[\M]{\mathcal{G} \lns \Gamma, \Box A_{k+1},\dots, \Box A_n
    \seq \Delta \lnse (A_1, \dots, A_k\seq B; B \seq \;)
        }
        {\infer*{\mathcal{G} \lns \Gamma, \Box A_{k+1},\dots, \Box A_n
    \seq \Delta \lnse (A_1, \dots, A_k\seq B; B,\bot \seq \;)
          }
          {\infer[\Box_L^\e]{\mathcal{G} \lns \Gamma, A_n
              \seq \Delta \lnse (A_1, \dots, A_{n-1}\seq B; B,\bot
          \seq \;)
            }
            {\deduce
            {\mathcal{G} \lns \Gamma
              \seq \Delta \lns A_1, \dots, A_n\seq B}{}
            &
            \infer[\bottom_L]{\mathcal{G} \lns \Gamma, A_n
              \seq \Delta \lns B, \bot \seq A_n
            }{}
          }
        }
      }}
    &\hspace{-2em}
    \mathcal{G} \lns \Gamma, \Box A_2,\dots, \Box A_n
    \seq \Delta \lns B\seq A_1
    }
  }
  \]
  is first turned into the following block by permuting the rule $\M$
  downwards and closing the derivations of the superfluous premisses
  using the $\bot_L$ rule:
  \[
  \infer[\Box_R^\e]{\mathcal{G} \lns \Gamma, \Box A_1,\dots, \Box A_n
    \seq \Box B, \Delta
  }
  {\infer[\M]{\mathcal{G} \lns \Gamma, \Box A_1,\dots, \Box A_n
    \seq \Delta \lnse (\;\seq B; B \seq \;)
    }
  {\infer[\C]{\mathcal{G} \lns \Gamma, \Box A_1,\dots, \Box A_n
    \seq \Delta \lnse (\;\seq B; B,\bot \seq \;)
    }
    {\infer*{\mathcal{G} \lns \Gamma, \Box A_2,\dots, \Box A_n
    \seq \Delta \lnse (A_1\seq B; B,\bot \seq \;)
      }
      {\infer[\Box_L^\e]{\mathcal{G} \lns \Gamma, A_n
          \seq \Delta \lnse (A_1, \dots, A_{n-1}\seq B; B,\bot
          \seq \;)
        }
        {\mathcal{G} \lns \Gamma
          \seq \Delta \lns A_1, \dots, A_n\seq B
          &
          \infer[\bot_L]{\mathcal{G} \lns \Gamma, A_n
            \seq \Delta \lns B, \bot \seq A_n}{}
        }
      }
    &\hspace{-2.5em}
    \infer[\bot_L]{\mathcal{G} \lns \Gamma, \Box A_2,\dots, \Box A_n
    \seq \Delta \lns B,\bottom\seq A_1}{}
    }
    }
  }
  \]
  The resulting modal block then is translated into an application of
  the rule $(\mathsf{M}n)$ with premiss $A_1,\dots, A_n \seq B$ and
  conclusion $\Gamma, \Box A_1,\dots, \Box A_n \seq \Box B,
  \Delta$. The propositional rules only work on the last component,
  never inside the nesting $\lnse$ and are translated easily by the
  corresponding sequent rules.

  Soundness of the system $\LNS_{\mathsf{E}\mathcal{A}}\Con\W$ then
  follows from this using admissibility of weakening
  (Lem.~\ref{weakening}) and soundness of the sequent system
  $\G_{\mathsf{E}\mathcal{A}}\Con$.  \qed
\end{proof}

Note that due to the following lemma for the logics of the non-normal
cube we could have avoided the complications arising from including
the contraction rules in the calculi. However, in view of the calculi
in later sections and the fact that the original sequent systems
include contraction explicitly or implicitly in the structure of
sequents as defined by sets instead of multisets we chose the given
more general method for proving soundness.

\begin{lemma}[Admissibility of Contraction]\label{contraction}
  Contraction is admissible in the calculus
  $\LNS_{\mathsf{E}\mathcal{A}}$, that is, if there is a derivation
  $\mathcal{D}$ of $\mathcal{S}\con{\Gamma, A, A \seq \Delta}$
  (resp. $\mathcal{S}\con{\Gamma \seq \Delta, A, A}$) in
  $\LNS_{\mathsf{E}\mathcal{A}}$, then there is a derivation
  $\mathcal{D}'$ of $\mathcal{S} \con{\Gamma, A \seq \Delta}$ (resp.
  $\mathcal{S} \con{\Gamma\seq \Delta, A}$) in
  $\LNS_{\mathsf{E}\mathcal{A}}$.
\end{lemma}

\begin{proof}
  The proof is for both statements simultaneously by double induction
  on the complexity of the contracted formula and the depth of the
  derivation. In case the main connective of the contracted formula is
  a propositional connective, we use invertibility of the
  propositional rules (Lemma~\ref{lem:m-invertibility}) followed by
  the (outer) induction hypothesis on the complexity of the contracted
  formula.  The cases where $A$ is a modal formula and not principal
  in the last applied rule are dealt with in the standard way by
  appealing to the (inner) induction hypothesis on the depth of the
  derivation. If the contracted formula is a modal formula and
  principal in the last applied rule, we distinguish cases according
  to the last applied rule.  Suppose, e.g., that
  $\mathcal{S}\con{\Gamma, \Box A, \Box A \seq \Delta}$ has a
  derivation of the shape
\[
\infer[\mathsf{C}]{\mathcal{G}\lns\Gamma,\Box A,\Box A\seq\Delta \lnse(\Sigma \seq \Pi; \Omega \seq \Theta)}
{\deduce{\mathcal{G} \lns \Gamma,\Box A \seq \Delta
    \lnse(\Sigma, A \seq \Pi; \Omega \seq \Theta)}{\pi_1} &
    \deduce{\mathcal{G} \lns \Gamma,\Box A \seq \Delta \lns \Omega \seq A, \Theta}{\pi_2}}
\]
Note that the $\Box A$ in the penultimate component of the conclusion
of $\pi_2$ will be necessarily weakened, since no logical rules can
act on it.  In the derivation $\pi_1$, either $\Box A$ is never
active, in which case it can be weakened, or it is active via one of
the rules $\mathsf{C}$ or $\Box_L^\e$. Let's consider the first case:
\[
\infer[\mathsf{C}]{\mathcal{G} \lns \Gamma,\Box A \seq \Delta
    \lnse(\Sigma, A \seq \Pi; \Omega \seq \Theta)}
{\deduce{\mathcal{G} \lns \Gamma \seq \Delta
    \lnse(\Sigma, A, A \seq \Pi; \Omega \seq \Theta)}{\pi_1'} &
    \deduce{\mathcal{G} \lns \Gamma \seq \Delta \lns \Omega \seq A,  \Theta}{\pi_1''}}
\]
Observe that the modal block in $\pi_1 '$ will eventually end by
producing leaves of the form
$\mathcal{G}\lns \Gamma'\seq \Delta'\lns \Sigma',A,A\seq \Pi'$ and
$\Omega'\seq \Theta'$. By induction hypothesis, for every derivation
for a sequent of the first form there is a derivation of
$\mathcal{G}\lns \Gamma'\seq \Delta'\lns \Sigma',A\seq \Pi'$. Hence,
starting from such leaves and applying the same sequence of rules as
in the modal block of $\pi_1'$, we have a derivation $\pi$ of
$\mathcal{G} \lns \Gamma \seq \Delta \lnse(\Sigma, A \seq \Pi; \Omega
\seq \Theta)$. Thus
\[
\infer[\mathsf{C}]{\mathcal{G} \lns \Gamma,\Box A \seq \Delta
    \lnse(\Sigma \seq \Pi; \Omega \seq \Theta)}
{\deduce{\mathcal{G} \lns \Gamma \seq \Delta
    \lnse(\Sigma, A \seq \Pi; \Omega \seq \Theta)}{\pi} &
    \deduce{\mathcal{G} \lns \Gamma \seq \Delta \lns \Omega \seq A,  \Theta}{\pi_1''}}
\]
The other cases are similar and simpler.
\end{proof}

It is worth noting that modular calculi for the logics in the
non-normal cube were also given in~\cite{Gilbert:2015} using the
framework of labelled sequents. The calculi presented there are very
much semantically motivated and are based on a translation of
non-normal modal logics into normal modal logics. The complexity of
the resulting semantic conditions then is captured using \emph{systems
  of rules}~\cite{Negri:2016}.

%%% Local Variables: 
%%% mode: latex
%%% TeX-master: "TOCL"
%%% End: 

\section{Structural Variants and the Modal Tesseract}
\label{sec:dosens-principle}
%!TEX root = TOCL.tex

The systems for the non-normal logics introduced in the last section
make use of different \emph{logical} rules, but sometimes it is
preferable to change logics only by modifying the \emph{structural}
rules of the system, i.e., the rules governing the behaviour of the
\emph{structural connective} $\lns$.  In particular, for sequent
systems varying the structural rules instead of the logical rules
often results in higher modularity, since cut elimination proofs are
usually less affected by additional structural rules. This has also
been called \emph{Do\v{s}en's Principle} in~\cite{Wansing:2002fk}. We
will now apply this idea to obtain modular calculi for a number of
extensions of monotone modal logic $\M$ (see also~\cite{Hansen:2003}
for a semantic treatment not only of these logics). In order to do so,
we first simplify the calculus for monotone modal logic.  As the avid
reader might have noticed, there is quite a lot of redundancy in this
calculus. In particular, after applying the rule $\M$, the second
premiss of the following applications of $\C$ or $\Box_L^\e$ become
trivially derivable. Hence for the present purpose we might as well
omit these premisses and the corresponding component of the nesting
operator, replacing the binary operator $\lnse$ with the unary
operator $\lnsm$\,. Linear nested sequents for monotone modal logics
then are given by:
\[
\LNS_\m ::= \Gamma \seq \Delta \mid \Gamma \seq \Delta \lnsm \Sigma
\seq \Pi \mid \Gamma \seq \Delta \lns \LNS_\m
\]
The rules $\Box_R^\e$ and $\Box_L^\e$ in the monotone setting then are
simplified to the rules $\Box_R^\m$ and $\Box_L^\m$ of
Fig.~\ref{fig:struct-monotone-LNS}, which now only need to carry
information about one direction of the premisses. The additional rules
for the axioms $\C$ and $\M$ (shown in the same figure) now are given
in their structural variants, permitting to switch from the ``finished
rule'' marker $\lns$ to the ``unfinished rule'' marker $\lnsm$ and
back. Obviously, adding both rules $\mathsf{N}$ and $\C$ collapses
both nesting operators into one, and essentially brings us to the
linear nested sequent calculus for modal logic $\K$ from
Fig.~\ref{fig:lns-K}, as should be the case since $\K$ is precisely
the logic $\M\mathsf{N}\C$. Finally, observe that applying rule $\C$
allows propositional rules to be applied between modal phases.
\begin{figure}[t]
  \hrule
  \[
  \infer[\Box_R^\m]{\mathcal{G} \lns \Gamma \seq \Box B,
    \Delta}{\mathcal{G} \lns \Gamma \seq \Delta \lnsm \;\seq
    B}
  \qquad
  \infer[\Box_L^\m]{\mathcal{G} \lns \Gamma, \Box A \seq \Delta \lnsm
    \Sigma \seq \Pi}{\mathcal{G} \lns \Gamma \seq \Delta
    \lns \Sigma, A \seq \Pi}
  \qquad
  \infer[\mathsf{C}]{\mathcal{G} \lns \Gamma\seq \Delta}{\mathcal{G}
    \lnsm \Gamma\seq \Delta }
  \qquad
  \infer[\mathsf{N}]{\mathcal{G} \lnsm \Gamma\seq \Delta}{\mathcal{G}
    \lns \Gamma\seq \Delta }
  \]
  \[
  \LNS_{\M\mathcal{A}} \quad \{\, \Box_R^\m, \Box_L^\m \, \} \cup
  \mathcal{A} \quad \text{ for }\mathcal{A} \subseteq \{ \mathsf{C}, \mathsf{N} \}
  \]
  \hrule
  \caption{The structural variants of the linear nested systems for monotone modal logics}
  \label{fig:struct-monotone-LNS}
\end{figure}

The main benefit of capturing the axioms $\C$ and $\mathsf{N}$ by
structural rules instead of logical rules is that it is now possible
to give calculi for further extensions in a uniform way, independent
of normality or non-normality of the base logic. The further axioms we
are going to consider are (using the terminology
of~\cite{Hansen:2003}):
\[
  \PP\; \neg \Box \bot
  \qquad
  \D\; \neg (\Box A \land \Box \neg A)
  \qquad
  \mathsf{T}\; \Box A \to A
  \qquad
  \mathsf{4}\; \Box A \to \Box \Box A
  \qquad
  \mathsf{5}\; \Box A \lor \Box \neg \Box A
\]
Note that we included both the two axioms $\PP$ and $\D$ which are
usually taken to be two different formulations of the axiom for
seriality. This is due to the fact that in the non-normal setting the
two formulations are not equivalent: while $\PP$ is derivable from
$\D$, the opposite does not hold in logics not validating the axiom
$\C$.  The reason for why we here only consider extensions of monotone
modal logic with these axioms instead of extensions of classical modal
logic $\mathsf{E}$ is that obtaining cut-free sequent calculi for many
of these extensions seems to be problematic, see
e.g.~\cite{Indrzejczak:2011}.

\begin{definition}[Sequent calculi]
  The sequent rules for extensions of monotone modal logic with axioms
  from $\{\PP,\D,\T,\4,\5\}$ are given in
  Fig.~\ref{fig:sequent-rules-ext-mon}.  Let
  $\mathcal{A} \subseteq \{ \mathsf{N}, \PP, \D, \mathsf{T}, \4, \5
  \}$. The sequent system $\G_{\mathsf{M}\mathcal{A}}$ contains the
  standard propositional rules of $\G$ (see
  Def.~\ref{def:seq-simplpy-dep}) as well as the following modal
  rules:
  \begin{itemize}
  \item $\{ \M \} \cup \mathcal{A}$
  \item $\D\4$ if $\{ \D, \4 \} \subseteq \mathcal{A}$
  \item $\D\5$ if $\{ \D,\5 \} \subseteq \mathcal{A}$.
  \end{itemize}
  The sequent system $\G_{\mathsf{M}\C\mathcal{A}}$ contains the
  standard propositional rules together with the additional rules
  \begin{itemize}
  \item $\{\C\} \cup \mathcal{A}$
  \item $\C\D$ if $\PP \in \mathcal{A}$ or $\D \in \mathcal{A}$
  \item $\C\4$ if $\4 \in \mathcal{A}$
  \item $\C\D\4$ if $\{\PP,\4\} \subseteq \mathcal{A}$ or $\{\D,\4\} \subseteq \mathcal{A}$
  \item $\K\4$ if $\{\mathsf{N},\4\}\subset \mathcal{A}$
  \item $\K\4\5$ if $\{\mathsf{N},\4,\5\} \subseteq \mathcal{A}$
  \item $\K\D\4\5$ if $\{\mathsf{N},\PP,\4,\5\} \subseteq\mathcal{A}$
    or $\{\mathsf{N},\D,\4,\5\} \subseteq \mathcal{A}$
  \end{itemize}
\end{definition}

\begin{figure}[t]
  \hrule\smallskip
  \[
  \infer[\M]{\Box A \seq \Box B}{A \seq B}
  \qquad
  \infer[\mathsf{N}]{\;\seq\Box A}{\;\seq A}
  \qquad
  \infer[\PP]{\Box A \seq \;}{A \seq \;}
  \qquad
  \infer[\D]{\Box A, \Box B \seq \;}{A,B \seq \;}
  \qquad
  \infer[\mathsf{T}]{\Gamma, \Box A \seq \Delta}{\Gamma, A \seq
    \Delta}
  \]
  \[
  \qquad
  \infer[\mathsf{4}]{\Box A \seq \Box B}{\Box A \seq B}
  \qquad
  \infer[\mathsf{5}]{\;\seq \Box A, \Box B}{\;\seq A, \Box B}
  \qquad
  \infer[\D\4]{\Box A, \Box B \seq \;}{A, \Box B \seq \;}
  \qquad
  \infer[\D\5]{\Box A \seq \Box B}{A \seq \Box B}
  \]
  \[
  \begin{array}{c}
  \infer[\C]{\Box \Gamma \seq \Box A}{\Gamma \seq A}\\
  (|\Gamma| \geq 1)
  \end{array}
  \qquad
  \infer[\C\D]{\Box \Gamma \seq \;}{\Gamma \seq \;}
  \qquad
  \begin{array}{c}
  \infer[\C\4]{\Box \Gamma, \Box \Sigma \seq \Box B}{\Box \Gamma,
    \Sigma \seq B}\\
  (|\Gamma,\Sigma| \geq 1)
  \end{array}
  \qquad
  \infer[\C\D\4]{\Box \Gamma, \Box \Sigma \seq \;}{\Box \Gamma,
    \Sigma \seq \;}
  \]
  \[
  \infer[\K\4]{\Box \Gamma, \Box \Sigma \seq \Box A}{\Box \Gamma, \Sigma \seq A}
  \qquad
  \infer[\K\4\5]{\Box \Gamma, \Box \Sigma \seq \Box A, \Box
    \Delta}{\Box \Gamma, \Sigma \seq A, \Box \Delta}
  \qquad
  \infer[\K\D\4\5]{\Box \Gamma, \Box \Sigma \seq \Box
    \Delta}{\Box \Gamma, \Sigma \seq \Box \Delta}
  \]
  \hrule
  \caption{Sequent rules for extensions of monotonic logics. We
    slightly abuse notation and write the same letters for axioms and
    the corresponding rules.}
  \label{fig:sequent-rules-ext-mon}
\end{figure}

Note that the modal sequent rules of
Fig.~\ref{fig:sequent-rules-ext-mon} do not absorb weakening into the
conclusion. Indeed, the weakening and contraction rules are not
admissible in most of these calculi. While of course the modal rules
could be modified as to make the structural rules admissible, to stay
closer to the literature we consider the calculi with explicit
structural rules. The additional sequent rules stipulated in the above
definition are required for cut elimination. In particular, this means
that modularity fails almost completely for these systems. Decomposing
the rules yields the linear nested sequent rules and rule sets
$\LNS_{\M\mathcal{A}}$ given in Fig.~\ref{fig:lns-ext-mon}. Note in
particular that we do not need to include additional rules at all, and
hence the calculi are completely modular.
\begin{figure}[t]
  \hrule\medskip
  \[
  \infer[\PP]{\mathcal{G} \lns \Gamma \seq \Delta}{\mathcal{G} \lns
    \Gamma \seq \Delta \lnsm \;\seq \;}
  \qquad
  \infer[\D]{\mathcal{G} \lns \Gamma, \Box A \seq \Delta}{\mathcal{G} \lns
    \Gamma \seq \Delta \lnsm A\seq \;}
  \qquad
  \infer[\mathsf{T}]{\mathcal{G} \lns \Gamma, \Sigma \seq \Delta,
    \Pi}{\mathcal{G} \lns \Gamma \seq \Delta \lnsm \Sigma \seq \Pi}
  \]
  \[
  \infer[\mathsf{4}]{\mathcal{G} \lns \Gamma, \Box A \seq \Delta \lnsm \Sigma
    \seq \Pi}{\mathcal{G} \lns \Gamma \seq \Delta \lns \Sigma,\Box A
    \seq\Pi}
  \qquad
  \infer[\mathsf{5}]{\mathcal{G} \lns \Gamma \seq \Delta, \Box A \lnsm \Sigma
    \seq \Pi}{\mathcal{G} \lns \Gamma \seq \Delta \lns \Sigma
    \seq\Pi,\Box A}
  \qquad
  \]
  \[
  \LNS_{\M\mathcal{A}} \quad \{ \Box_R^\m, \Box_L^\m\} \cup
  \mathcal{A} \quad \text{ for }\mathcal{A} \subseteq \{ \mathsf{C},
  \mathsf{N}, \PP, \D, \4,\5\}
  \]
  \hrule
  \caption{Linear nested sequent rules for extensions of monotonic modal logics}
  \label{fig:lns-ext-mon}
\end{figure}
As for normal modal logics, extensions of $\M$ including the axiom
$\5$ are not as well behaved as those without it. We first consider
logics not including $\5$. Most of the following results can be found
in the literature.

\begin{proposition}\label{prop:cut-elim-mon}
  For
  $\mathcal{A} \subseteq \{\mathsf{N},\mathsf{C},\PP,\D,\mathsf{4}\}$
  the sequent calculus $\G_{\mathsf{M}\mathcal{A}}\Con\W$ is sound and
  complete for the logic $\mathsf{M}\mathcal{A}$, i.e., for every
  formula $A$ we have $A \in \mathsf{M}\mathcal{A}$ if and only if
  $\entails_{\G_{\mathsf{M}\mathcal{A}}\Con\W}\;\seq A$.
\end{proposition}

\begin{proof}
  For the extensions of normal modal logic $\K = \M\mathsf{N}\C$, see
  e.g.~\cite{Wansing:2002fk}.  For the logic $\M\C\T$ the result is
  shown in~\cite{Ohnishi:1957}, for the extensions of $\M$ with axioms
  from $\{\mathsf{N}, \C\}$ see~\cite{Lavendhomme:2000}. The result
  for the logics $\M\PP$ and $\M\C\PP = \M\C\D$ can be found
  in~\cite{Orlandelli:2014}. The majority of the results for the
  non-normal logics are due to~\cite{Indrzejczak:2005}, namely the
  calculi for all extensions of $\M$ with axioms from
  $\{\mathsf{N}, \D, \T, \4 \}$. The remaining calculi for the logics
  $\M\PP, \M\PP\4, \M\mathsf{N}\PP, \M\mathsf{N}\PP\4, \M\C\4,
  \M\C\PP\4 = \M\C\D\4$ and $\M\T\4$ can be % easily
  constructed using methods similar to the ones
  in~\cite{Lellmann:2013fk,Lellmann:2013}. The cut elimination proof
  for these calculi essentially is an extension of the cut elimination
  proof given in~\cite{Indrzejczak:2005}. Since it is not central to
  the topic of this paper we relegate it to
  Appendix~\ref{app:cut-elim-mon}.
\end{proof}

This gives modular linear nested sequent systems for every combination
of $\mathsf{C},\mathsf{N},\PP,\D, \mathsf{T}, \mathsf{4}$.

\begin{theorem}[Soundness and Completeness]\label{thm:non-normal-lns-compl}
  Let $\mathcal{A}$ be a subset of
  $\{\mathsf{C}, \mathsf{N}, \PP, \D, \mathsf{T}, \mathsf{4}\}$. Then
  the linear nested sequent calculus
  $\LNS_{\mathsf{M}\mathcal{A}}\Con\W$ is sound and complete for the
  logic $\mathsf{M}\mathcal{A}$, i.e., for every formula $A$ we have
  $A \in \mathsf{M}\mathcal{A}$ if and only if
  $\entails_{\LNS_{\mathsf{M}\mathcal{A}}\Con\W} \;\seq A$.
\end{theorem}

\begin{proof}
  We first show completeness by simulating sequent derivations in the
  last component. Here we only show how to simulate the modal rules.
  First, the rules $(\mathsf{M}n)$ for monotone logics including the
  axiom $\C$ are simulated by
  \[
  \vcenter{\infer[(\mathsf{M}n)]{\Gamma, \Box A_1,\dots,\Box A_n \seq
      \Box B, \Delta}{A_1, \dots, A_n
      \seq B}}
  \quad\leadsto\quad
  \vcenter{
    \infer[\Box_R^\m]{\Gamma, \Box A_1, \dots, \Box A_n \seq \Box B, \Delta
    }
    {\infer[\Box_L^\m]{\Gamma, \Box A_1,\dots, \Box A_n \seq \Delta \lnsm \;\seq B
      }
      {\infer[\C]{\Gamma, \Box A_2,\dots, \Box A_n \seq \Delta
          \lns \;A_1 \seq B
        }
        {\infer*[]{\Gamma, \Box A_2,\dots, \Box A_n \seq \Delta \lnsm
            A_1\seq B
          }
          {\infer[\Box_L^\m]{\Gamma, \Box A_n \seq \Delta \lnsm
            A_1,\dots, A_{n-1}\seq B
            }
            {\Gamma \seq \Delta \lns A_1,\dots, A_n \seq B
            }
          }
        }
      }
    }
  }
  \]
  The case for $n = 1$ gives the simulation of the rule $\M$. The
  necessitation rule $N$ is simulated by:
  \[
  \vcenter{
    \infer[\mathsf{N}]{\Gamma \seq \Box B, \Delta}{\;\seq B}
  }
  \quad\leadsto\quad
  \vcenter{
  \infer[\Box_R^\m]{\Gamma \seq \Box B, \Delta
  }
  {\infer[\mathsf{N}]{\Gamma \seq \Delta \lnsm \;\seq B
    }
    {\Gamma \seq \Delta \lns \;\seq B
    }
  }
  }
  \]
  The simulations of the rules from Fig.~\ref{fig:lns-ext-mon} then
  are (omitting the general context $\mathcal{G}$):
  \[
  \vcenter{\infer[\PP]{\Box A \seq \;}{A \seq \;}}
  \quad\leadsto\quad
  \vcenter{\infer[\PP]{\Box A \seq \;
  }
  {\infer[\Box_L^\m]{\Box A \seq \; \lnsm \;\seq \;
    }
    {\;\seq\; \lns A \seq \;
    }
  }}
  \qquad\qquad
  \vcenter{\infer[\D]{\Box A, \Box B \seq \;}{A, B \seq \;}}
  \quad\leadsto\quad
  \vcenter{\infer[\D]{\Box A, \Box B \seq \;
  }
  {\infer[\Box_L^\m]{\Box A \seq \; \lnsm\, B \seq \;
    }
    {\;\seq\; \lns A, B \seq \;
    }
  }}
  \]
  \[
  \vcenter{\infer[\mathsf{T}]{\Gamma, \Box A \seq \Delta}{\Gamma, A
      \seq \Delta}}
  \quad\leadsto\quad
  \vcenter{
    \infer[\mathsf{T}]{\Gamma, \Box A \seq \Delta
    }
    {\infer[\Box_L^\m]{\Box A \seq \; \lnsm\; \Gamma \seq \Delta
      }
      {\;\seq\; \lns \Gamma, A \seq \Delta
      }
    }
  }
  \qquad\qquad
  \vcenter{
  \infer[\4]{\Box A \seq \Box B}{\Box A \seq B}
  }
  \quad\leadsto\quad
  \vcenter{\infer[\Box_R^\m]{\Box A \seq \Box B
    }
    {\infer[\4]{\Box A \seq \; \lnsm \;\seq B
      }
      {\;\seq\; \lns \Box A \seq B
      }
    }}
  \]
  \[
  \vcenter{
    \infer[\D\4]{\Box A, \Box B \seq \;}{\Box A, B \seq \;}
  }
  \quad  \leadsto \quad
  \vcenter{
  \infer[\D]{\Gamma, \Box A, \Box B \seq \Delta
  }
  {\infer[\4]{\Gamma, \Box B \seq \Delta \lnsm A \seq \;
    }
    {\Gamma \seq \Delta \lns  A, \Box B \seq \;
    }
  }
  }
  \]
  In the presence of $\mathsf{C}$ we use the rule $\mathsf{C}$ to move
  additional formulae on the left hand side. E.g., for the system
  containing the axioms $\mathsf{C},\mathsf{D}$ and $\4$ we would
  have:
  \[
  \vcenter{
  \infer[\mathsf{CD4}]{\Box B, \Box A_1, \Box A_2 \seq\;
    \;}{\Box B, A_1, A_2 \seq \;}
  }
  \qquad\leadsto\qquad
  \vcenter{
  \infer[\D]{\Box B, \Box A_1, \Box A_2 \seq \;
  }
  {\infer[\4]{\Box B,  \Box A_2 \seq \;\lnsm\, A_1\seq\;
    }
    {\infer[\mathsf{C}]{   \Box A_2 \seq \;\lns \Box B,A_1\seq\;
      }
      {\infer[\Box_L^\m]{  \Box A_2 \seq \;\lnsm\, \Box B,A_1\seq\;
        }
        {\; \seq \;\lnsm\,  \Box B, A_1,A_2\seq\;
        }
      }
    }
  }
  }
  \]
  If the logic contains $\PP$ but not $\D$, the application of the
  rule $\D$ above is replaced by an application of $\PP$ followed by
  applications of $\Box_L^\m$ and $\C$.  The cases of the remaining
  rules $\C, \C\D, \C\4, \K\4$ are analogous.

  To show soundness we first observe that applications of the
  propositional rules in the last component can be permuted above
  blocks of applications between the rules $\mathsf{C}$ and
  $\Box_L^\m$, i.e., above blocks of rule applications where the last
  nesting is $\lnsm$ . This is due to the fact that only modal rules
  can be applied inside the nesting $\lnsm$ , no modal rule creates a
  new nesting after a $\lnsm$ nesting, and every modal rule keeps all
  the formulae occurring under the nesting $\lnsm$ in the conclusion
  at the same place. Hence we may assume that in a derivation in
  $\LNS_{\mathsf{M}\mathcal{A}}$ all the modal rules occur in a
  block. Then, analogously to Lemma~\ref{lem:end-activeness-LNSEA} of
  the previous section, and using the formulations of
  Lemma~\ref{weakening} and~\ref{lem:m-invertibility} for the
  present calculi (the proofs of which are completely analogous), we
  convert the derivation into a derivation where all the applications
  of rules are end-active. Such a derivation then is converted into a
  sequent derivation. In particular, every modal block then can be
  translated into one or more modal rules in the sequent system:
  whenever we have a block
  \[
  \infer{\mathcal{G} \lns \Gamma, \Sigma' \seq \Delta,
    \Pi'}{\mathcal{G} \lns \Gamma \seq \Delta \lns \Sigma \seq \Pi}
  \]
  consisting only of modal rules, then the sequent rule
  \[
  \infer[]{\Gamma, \Sigma' \seq \Delta, \Pi'}{\Sigma \seq \Pi}
  \]
  is derivable in the corresponding sequent system. The
  transformations are essentially the backwards directions of the
  transformations given above.
\end{proof}

Thus we obtain modular nested sequent calculi for all the logics in
what could be called the \emph{modal tesseract}
(Fig.~\ref{fig:modal-tesseract}), hence repairing the bridge between
non-normal and normal modal logics. Note that the modal tesseract
includes one side of the standard (normal) modal cube, see
e.g.~\cite{Brunnler:2009kx}.
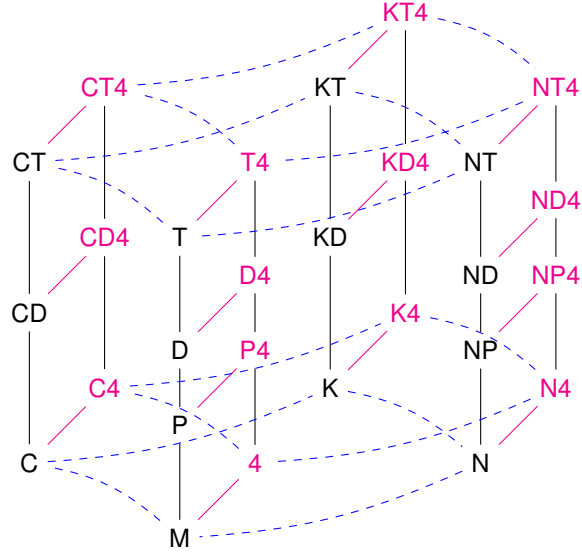
\begin{figure}[t]
  \hrule\bigskip
  \begin{center}
    \begin{tikzpicture}
      \node (M) at (0,0) {$\M$};
      \node (MD1) at (0,1.5) {$\PP$};
      \node (MD2) at (0,2.5) {$\D$};
      \node (MT) at (0,4) {$\T$};
      \node (M4) [color=magenta] at (1,1) {$\4$};
      \node (MD14) [color=magenta] at (1,2.5) {$\PP\4$};
      \node (MD24) [color=magenta] at (1,3.5) {$\D\4$};
      \node (MT4) [color=magenta] at (1,5) {$\T\4$};
      \draw [color=magenta] (M) -- (M4);
      \draw [color=magenta] (MD1) -- (MD14);
      \draw [color=magenta] (MD2) -- (MD24);
      \draw [color=magenta] (MT) -- (MT4);
      \draw (M) -- (MD1) -- (MD2) -- (MT);
      \draw (M4) -- (MD14) -- (MD24) -- (MT4);
      \node (C) at (-2,1) {$\C$};
      \node (CD) at (-2,3) {$\C\D$};
      \node (CT) at (-2,5) {$\C\T$};
      \node (C4) [color=magenta] at (-1,2) {$\C\4$};
      \node (CD4) [color=magenta] at (-1,4) {$\C\D\4$};
      \node (CT4) [color=magenta] at (-1,6) {$\C\T\4$};
      \draw [color=magenta] (C) -- (C4);
      \draw [color=magenta] (CD) -- (CD4);
      \draw [color=magenta] (CT) -- (CT4);
      \draw (C) -- (CD) -- (CT);
      \draw (C4) -- (CD4) -- (CT4);
      \node (N) at (4,1) {$\mathsf{N}$};
      \node (ND1) at (4,2.5) {$\mathsf{N}\PP$};
      \node (ND2) at (4,3.5) {$\mathsf{N}\D$};
      \node (NT) at (4,5) {$\mathsf{N}\T$};
      \node (N4) [color=magenta] at (5,2) {$\mathsf{N}\4$};
      \node (ND14) [color=magenta] at (5,3.5) {$\mathsf{N}\PP\4$};
      \node (ND24) [color=magenta] at (5,4.5) {$\mathsf{N}\D\4$};
      \node (NT4) [color=magenta] at (5,6) {$\mathsf{N}\T\4$};
      \draw [color=magenta] (N) -- (N4);
      \draw [color=magenta] (ND1) -- (ND14);
      \draw [color=magenta] (ND2) -- (ND24);
      \draw [color=magenta] (NT) -- (NT4);
      \draw (N) -- (ND1) -- (ND2) -- (NT);
      \draw (N4) -- (ND14) -- (ND24) -- (NT4);
      \node (K) at (2,2) {$\K$};
      \node (KD) at (2,4) {$\K\D$};
      \node (KT) at (2,6) {$\K\T$};
      \node (K4) [color=magenta] at (3,3) {$\K\4$};
      \node (KD4) [color=magenta] at (3,5) {$\K\D\4$};
      \node (KT4) [color=magenta] at (3,7) {$\K\T\4$};
      \draw [color=magenta] (K) -- (K4);
      \draw [color=magenta] (KD) -- (KD4);
      \draw [color=magenta] (KT) -- (KT4);
      \draw (K) -- (KD) -- (KT);
      \draw (K4) -- (KD4) -- (KT4);
      \draw [color=blue,dashed] (M) to [bend right=15] (C);
      \draw [color = blue,dashed] (C) to [bend right=10] (K);
      \draw [color=blue,dashed] (K) to [bend left=15] (N);
      \draw [color=blue,dashed] (N) to [bend left=10] (M);
      \draw [color=blue,dashed] (MT) to [bend right=15] (CT);
      \draw [color=blue,dashed] (CT) to [bend right=10] (KT);
      \draw [color=blue,dashed] (KT) to [bend left=15] (NT);
      \draw [color=blue,dashed] (NT) to [bend left=10] (MT);
      \draw [color=blue,dashed] (M4) to [bend right=15] (C4);
      \draw [color=blue,dashed] (C4) to [bend right=10] (K4);
      \draw [color=blue,dashed] (K4) to [bend left=15] (N4);
      \draw [color=blue,dashed] (N4) to [bend left=10] (M4);
      \draw [color=blue,dashed] (MT4) to [bend right=15] (CT4);
      \draw [color=blue,dashed] (CT4) to [bend right=10] (KT4);
      \draw [color=blue,dashed] (KT4) to [bend left=15] (NT4);
      \draw [color=blue,dashed] (NT4) to [bend left=10] (MT4);
    \end{tikzpicture}
  \end{center}
  \hrule
  \caption{The modal tesseract}
  \label{fig:modal-tesseract}
\end{figure}
  
For logics including the axiom $\5$ the situation is a bit more
complicated, since not all of these (in particular $\K\5$ and $\Sfi$)
have cut-free sequent calculi. However, while in this case we do not
obtain full modularity, we still obtain calculi for a number of
logics. The number of logics we need to consider in this case is
greatly reduced by the following simple observation.

\begin{lemma}
  The axiom $\mathsf{N}$ is derivable in any extension of $\mathsf{M}$
  including an axiom of the form $\Box A_1 \lor \dots \lor \Box A_n$. In
  particular, the axiom $\mathsf{N}$ is derivable in $\M\5$.
\end{lemma}

\begin{proof}
  Using the fact that $A_i \iimp \top$ is a tautology for every $A_i$
  and monotonicity of the box operator.
\end{proof}

Hence the lattice of extensions of $\M\5$ with axioms from
$\{\mathsf{N}, \C, \PP,\D, \T, \4 \}$ collapses to the 12 logics shown
in Fig.~\ref{fig:ext-m5} (the \emph{house of $\M\5$} -- not all
corners of it are safe, i.e., cut-free, though). In particular, the
extensions of $\M\C\5$ are the same as the extensions of normal modal
logic $\K\5$. Again, most of the following results are found in the
literature.
\begin{figure}[t]
  \hrule
  \begin{center}\smallskip
  \begin{tikzpicture}
    \node (N5) at (0,0) {$\M\5$};
    \node (aux) at (.85,0) {$=\M\mathsf{N}\5$};
    \node (ND15) at (0,1.25) {$\M\PP\5$};
    \node (ND25) at (0,2.25) {$\M\D\5$};
    \node (NT5) at (.75,4) {$\M\T\5$};
    \node (N45) at (1.5,.5) {$\M\4\5$};
    \node (ND145) at (1.5,1.75) {$\M\PP\4\5$};
    \node (ND245) at (1.5,2.75) {$\M\D\4\5$};

    \draw (N5) -- (ND15) -- (ND25) -- (NT5) -- (ND245) -- (ND145) --
    (N45) -- (N5);
    \draw (ND15) -- (ND145);
    \draw (ND25) -- (ND245);

    \node (K5) at (-3,.75)  {$\K\5$};
    \node (aux2) at (-3.85,.75) {$\mathsf{MC5} = $};
    \node (KD5) at (-3,3) {$\K\D\5$};
    \node (KT5) at (-2.25,4.75) {$\K\T\5$};
    \node (K45) at (-1.5,1.25) {$\K\4\5$};
    \node (KD45) at (-1.5,3.5) {$\K\D\4\5$};
    \draw (K5) -- (KD5) -- (KT5) -- (KD45) -- (K45) -- (K5);
    \draw (KD5) -- (KD45);

    \draw (NT5) -- (KT5);
    \draw (ND25) -- (KD5);
    \draw (ND245) -- (KD45);
    \draw (ND15) to [bend right=10] (KD5);
    \draw (ND145) to [bend right=10] (KD45);
    \draw (N5) -- (K5);
    \draw (N45) -- (K45);
  \end{tikzpicture}
  \end{center}
  \hrule
  \caption{The extensions of modal logic $\M\5$}
  \label{fig:ext-m5}
\end{figure}
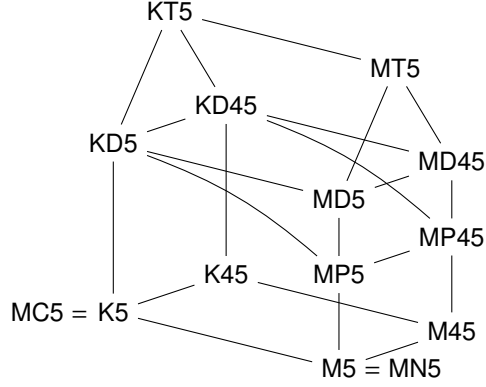

\begin{proposition}\label{prop:cut-elim-five}
  Let $\mathcal{L}$ be one of the logics
  \[
  \{ \M\5, \M\PP\5, \M\45,\M\PP\4\5, \M\D\4\5, \K\4\5, \K\D\4\5 \}
  \]
  Then $\G_\mathcal{L}\Con\W$ is sound and complete for $\mathcal{L}$,
  i.e., for every formula $A$ we have $A \in \mathcal{L}$ if and only
  if $\entails_{\LNS_{\mathcal{L}}} \;\seq A$.
\end{proposition}

\begin{proof}
  For the logics $\K\4\5$ and $\K\D\4\5$ this was shown
  in~\cite{Shvarts:1989}, but note that the calculi for $\K\4\5$ and
  $\K\D\4\5$ considered here are slight variations of the ones
  in~\emph{op.~cit.} In particular, there also the rule
  \[
  \infer[]{\Box \Gamma, \Box \Sigma \seq \Box \Delta}{\Box \Gamma,
    \Sigma \seq \Box \Delta}
  \]
  with nonempty $\Delta$ is included in the rule set for
  $\K\4\5$. This rule can be derived in the calculus considered here
  using the rule $\K\4\5$ together with a cut on the derivable sequent
  $\Box \Box A \seq \Box A$. Equivalence of the cut-free systems then
  follows from cut elimination for $\G_{\K\4\5}$. The latter follows
  from the general criteria of~\cite[Thm.~2.3.16]{Lellmann:2013}, but
  is also given explicitly in Appendix~\ref{app:cut-elim-mon}.
  
  In the non-normal case, for the logics $\M\5, \M\4\5$ and $\M\D\4\5$
  the result can be found in~\cite{Indrzejczak:2005}. The result for
  the remaining two logics $\M\PP\5$ and $\M\PP\4\5$ follows similarly
  to the results of Prop.~\ref{prop:cut-elim-mon} from the cut
  elimination proof in Appendix~\ref{app:cut-elim-mon}.
\end{proof}

For all the other cases there are counterexamples to cut
elimination. In particular, for the non-normal logic $\M\D\5$ this is
given e.g.\ by the formula $\Box p \iimp \Box \Diam p$, which is
derivable using the instance $\Box p \iimp \Diam p$ of axiom $\D$ and
the instance $\Diam p \iimp \Box \Diam p$ of axiom $\5$, but not
cut-free derivable in $\G_{\M\D\5}\Con\W$. Interestingly, this formula is
not a theorem of $\M\PP\5$, and hence not a counterexample to cut
elimination for $\G_{\M\PP\5}\Con\W$.\footnote{For the reader familiar with
  neighbourhood semantics~\cite{Chellas:1980fk,Hansen:2003}: The
  $\M\PP\5$-model $(\{a,b\},\eta,\sigma)$ with
  $\eta(a) = \eta(b) = \{ \{a\}, \{b\}, \{a,b\} \}$ and
  $\llbracket p \rrbracket = \{a \}$ witnesses satisfiability of the
  negation of this formula.}

\begin{theorem}
  Let $\mathcal{L}$ be one of the logics
  \[
  \{ \M\5, \M\PP\5, \M\45,\M\PP\4\5, \M\D\4\5, \K\4\5, \K\D\4\5 \}
  \]
  Then $\LNS_\mathcal{L}\Con\W$ is sound and complete for
  $\mathcal{L}$, i.e., for every formula $A$ we have
  $A \in \mathcal{L}$ if and only if
  $\entails_{\LNS_{\mathcal{L}}\Con\W} \;\seq A$.
\end{theorem}

\begin{proof}
  Analogous to the proof of Thm.~\ref{thm:non-normal-lns-compl}. The
  missing transformations from sequent rules into linear nested
  sequent derivations are:
  \[
  \vcenter{\infer[\5]{\;\seq \Box A, \Box B}{\;\seq A, \Box B}}
  \quad\leadsto\quad
  \vcenter{\infer[\Box_R^\m]{\; \seq\Box A, \Box B
    }
    {\infer[\5]{ \;\seq \Box B \lnsm \;\seq A
      }
      {\;\seq\; \lns \; \seq  A,\Box B
      }
    }}
  \qquad\qquad
  \vcenter{
    \infer[\D\5]{\Box A \seq \Box B}{A \seq \Box B}
  }
  \quad  \leadsto \quad
  \vcenter{
  \infer[\D]{\Gamma, \Box A \seq \Box B, \Delta
  }
  {\infer[\5]{\Gamma \seq \Box B, \Delta \lnsm A \seq \;
    }
    {\Gamma \seq \Delta \lns A \seq \Box B
    }
  }
  }
  \]
  The rules $\K\4\5$ and $\K\D\4\5$ are transformed similar to the
  case of $\C\D\4$.

  The soundness proof is analogous to the one for the cases not
  involving the axiom $\5$.
\end{proof}

%%% Local Variables: 
%%% mode: latex
%%% TeX-master: "TOCL"
%%% End: 

\section{Reconciling sequents with nested sequents}\label{sec:focused}
%!TEX root = TOCL.tex

We have presented a modular way of proposing several different modal
systems.  The beauty in this is that all systems share the same core,
where modal rules can be plugged in and/or mixed together.  For that,
we refined sequent rules, exposing their behaviour locally. The price
to pay for this modularity is, of course, efficiency, since there are
more rules which could have been applied to derive a given sequent. In
particular, the propositional rules could be applied in any component,
giving rise to a great number of derivations which should be
identified modulo bureaucracy. This alone could be taken care of by
simply restricting the calculi to their end-active variants, so that
the propositional rules are applied only in the last
component. However, doing so would still leave open the possibility of
mixing propositional and modal rules, e.g., applying (bottom-up) a
rule $\ibox{i}_R$ followed by a propositional rule in the last
component, and then a rule $\ibox{ji}_L$. This as well is a potential
source of inefficiency when compared to the sequent framework, where
we have blocks of propositional rules alternating with single modal
rules.

In this section, we will show how auxiliary nesting operators can be
used in order to guarantee a notion of \emph{normal form} for $\LNS$
derivations that mimic the respective sequent ones, hence reducing the
proof search space and optimizing proof search.  As noted in the last
sections, all the systems in this work can be restricted to their
end-active versions without losing completeness. However, as mentioned
above this is not enough to ensure that a $\LNS$ derivation
corresponds directly to a sequent derivation, since the propositional
rules could be applied between to applications of, e.g., the rule
$\Box_L$. Fortunately, rules for the propositional connectives permute
over box left rules, though. This allows modal rules to be restricted
so that they occur \emph{in a block}.

\begin{definition}\label{block}
  A $\LNS$ derivation is in \emph{block form} if, whenever a modal
  rule occurs directly above a propositional rule, then that modal
  rule creates a new component.
\end{definition}

In the following we use block form as the normal form of $\LNS$
derivations.  Considering first the simply dependent normal multimodal
logics of Section~\ref{sec:simply-depend-modal}, in
Fig.~\ref{fig:FLS-K} we present $\FLNS_{(N,\cless,F)}$, an end-active
version for $\LNS_{(N,\cless,F)}$
(Fig.~\ref{fig:lns-rules-simply-dep-mult}) where all derivations are
necessarily in block form: this is assured by an auxiliary nesting
operator $\lnsb[i]$ for each $i\in N$. This operator behaves much in
the same way as the ``unfinished rule marker'' in the systems for
non-normal modal logics. However, here we explicitly include the rule
$\cls$, which intuitively marks a sequent rule as finished.

This implies that, \emph{modulo the order of application of
  $\ibox{ij}_L$ and $\mathsf{d}_{ij}$ rules}, there is a 1-1
correspondence between derivations in the end-active variant of the
$\LNS$ system $\FLNS_{(N,\cless,F)}\Con\W$ and in the sequent system
$\G_{(N,\cless,F)}\Con\W$ (see Fig.~\ref{fig:simply-dep-multimodal}).
In this way, sequent rules can be seen as \emph{macro rules} for
linear nested rules.

Since every $\FLNS_{(N,\cless,F)}\Con\W$-derivation can be translated
into a $\LNS_{(N,\cless,F)}\Con\W$-derivation by replacing the nesting
$\lnsb[i]$ everywhere by $\lns[i]$ and omitting every application of
the rule $\cls$ we immediately obtain soundness of the system
$\FLNS_{(N,\cless,F)}\Con\W$. Completeness follows as mentioned above
from permuting propositional and structural rules below modal rules in
derivations in the end-active variant of $\LNS_{(N,\cless,F)}\Con\W$.

Observe that the \emph{normal} modal logics presented in this paper
form a particular case of simply dependent multimodal logics (with $N$
being a singleton). Hence all derivations in the end-active variant of
the correspondent $\FLNS_{(N,\cless,F)}\Con\W$ system will be in block
form.
\begin{figure}[t]
  \hrule
  \smallskip
  \[
  \infer[{\ibox{ij}}_L]{\mathcal{G}\lns[k] \Gamma, \ibox{i} A \seq \Delta
    \lnsb[j] \Sigma \seq \Pi}{ \mathcal{G}\lns[k]\Gamma \seq \Delta
    \lnsb[j] \Sigma,A \seq \Pi}
  \qquad
  \infer[{\ibox{i}}_R]{\mathcal{G}\lns[k]  \Gamma \seq \Delta, \ibox{i}
    A}{\mathcal{G}\lns[k]  \Gamma \seq \Delta\, \lnsb[i] \;\seq A}
     \qquad
  \infer[\cls]{\mathcal{G} \lnsb[k] \Gamma\seq \Delta}{\mathcal{G}
    \lns[k] \Gamma\seq \Delta }
  \]
  \[
  \infer[\mathsf{d}_{ij}]{\mathcal{G}\lns[k] \Gamma, \ibox{i} A \seq
    \Delta}{\mathcal{G} \lns[k] \Gamma \seq \Delta \lnsb[j] A \seq \;}
  \qquad
  \infer[\mathsf{t}_i]{\mathcal{G} \lns[k] \Gamma, \ibox{i} A \seq
    \Delta}{\mathcal{G} \lns[k]\Gamma, A \seq \Delta}
  \qquad
  \infer[{\4_{ij}}]{\mathcal{G} \lns[k] \Gamma, \ibox{i} A \seq \Delta
    \lnsb[j] \Sigma \seq \Pi}{\mathcal{G} \lns[k]\Gamma \seq \Delta
    \lnsb[j] \Sigma,\ibox{i}A \seq \Pi}
  \]
   \hrule
   \caption{Modal rules for $\FLNS_{(N,\cless,F)}$, where $k, i,j$ are
     as in Figure~\ref{fig:lns-rules-simply-dep-mult}.  The
     propositional rules are the same as in Fig.~\ref{fig:lns-G},
     restricted to the last component.}
  \label{fig:FLS-K}
\end{figure}

\begin{example}
The block form derivation for the normality axiom is  as follows
\[
\infer=[\iimp_R]{ \seq\Box(p \iimp q) \iimp (\Box p
    \iimp \Box q)
}
    {\infer[\Box_R]{ \Box(p \iimp q), \Box p
    \seq \Box q
      }
      {\infer=[\Box_L]{ \Box(p \iimp q), \Box p\seq \cdot
    \lnsb\;\cdot\seq q
        }
        {\infer[\cls]{\cdot \seq \cdot\lnsb\;p \iimp q, p\seq
    q}
          {\infer[\iimp_L]{\cdot\seq\cdot\lns p \iimp q, p \seq q
            }
              {\infer[\init]{\cdot\seq\cdot\lns p\seq p,q}{}
            &
              \infer[\init]{\cdot\seq\cdot\lns p,q\seq q}{}
            }
          }
        }
      }
      }
\]
\end{example}

All the systems presented for \emph{non-normal} modal logics in the
previous sections are end-active and in most systems the partial
nesting operator already forces that all valid derivations are in
block form. The exception are the systems containing the $\C$ rule
(Fig.~\ref{fig:struct-monotone-LNS}). In fact, this rule allows a
partial nesting to begin anywhere in the derivation, not only after an
application of a modal rule.

An alternative set of rules for these systems is obtained by adding
the nesting operator $\lnsb$ (so that the modal rules have two levels
of partial processing), together with the $\cls$ rule, that forces the
modal block to end. We illustrate this in Fig.~\ref{fig:FLNS} for the
system $\FLNS_{\M\C}$. Again, soundness and completeness follow as
above.
\begin{figure}[t]
  \hrule
  \[
  \infer[\Box_R^\m]{\mathcal{G} \lns \Gamma \seq \Box B,
    \Delta}{\mathcal{G} \lns \Gamma \seq \Delta \;\lnsm \seq
    B}
  \quad
  \infer[\Box_L^\m]{\mathcal{G} \lns \Gamma, \Box A \seq \Delta \;\lnsm
    \Sigma \seq \Pi}{\mathcal{G} \lns \Gamma \seq \Delta
  \;  \lnsr \Sigma, A \seq \Pi}
  \quad
  \infer[\mathsf{C}]{\mathcal{G} \;\lnsr \Gamma\seq \Delta}{\mathcal{G}
   \; \lnsm \Gamma\seq \Delta }
  \quad
  \infer[\mathsf{close}]{\mathcal{G}\; \lnsr \Gamma\seq \Delta}{\mathcal{G}
    \lns \Gamma\seq \Delta }
  \]
  \hrule
  \caption{System $\FLNS_{\M\C}$.}
  \label{fig:FLNS}
\end{figure}

\subsection{Block forms versus focused derivations}\label{sec:foc}
In~\cite{Andreoli:1992}, a notion of normal form for cut-free
derivations in linear logic was introduced. This normal form is given
by a focused proof system organised around two ``phases'' of proof
construction: the \emph{negative phase} for invertible inference rules
and the \emph{positive phase} for non-necessarily-invertible inference
rules.  Due to invertibility, when searching for a derivation it is
always safe to apply, reading bottom-up, a rule for a negative
formula, so these may be applied at any time. On the other hand, rules
for positive formulae may require a choice or restriction on the
application of rules.  Hence, in the focusing discipline, negative
formulae are decomposed eagerly until only positive formulae are left,
then one of them is non-deterministically chosen to be focused
on. Thus focused derivations alternate negative and positive phases.

Focused nested systems for modal logics were first considered
in~\cite{DBLP:conf/fossacs/ChaudhuriMS16}, where a focused variant for
all modal logics of the classical $\Sfi$ cube were proposed. This
approach was extended to the intuitionistic case
in~\cite{DBLP:conf/rta/ChaudhuriMS16}.

Since block form derivations entail a notion of normal forms in $\LNS$
by alternating modal and propositional blocks, it is natural to ask if
there is any relationship between focusing in nested systems and modal
blocks in linear nested systems.
 
In (plain) nested systems, the rule $\Box_R$ is invertible (negative),
since it basically implements the semantical description for the box
operator, described by a \emph{forall} connective (which has a
negative behaviour). On the other hand, $\Box_L$ is a non-invertible
rule (positive), since its application should be preceded by an
instance of a $\Box_R$ rule.
 
However, on passing to linear nesting systems, the dualities of
polarities for modal connectives is lost. In fact, $\Box_L$ and
$\Box_R$ rules are both non-invertible in $\LNS_{\K}$: while the left
rule can be applied only after a right rule, for the right rule a
boxed formula has to be chosen in order to be processed.  Hence there
seems to be no natural way of polarising the modal connectives
presented in this paper.

Let's take a closer look at the sequent rule $\mathsf{k}$
 \[
\infer[\mathsf{k}]{\Gamma',\Box \Gamma\seq\Box A,\Delta}{\Gamma\seq A} 
\]
and its interpretations in nested and linear nested systems.

In the \emph{nested} system proposed
in~\cite{DBLP:conf/fossacs/ChaudhuriMS16} all the existing right boxes
can be processed in parallel (and this is invertible) and then the
left boxes can be transferred one by one to all the nestings. A
derivation then proceeds by running all the possible traces in
parallel, and finish whenever one or more of them succeed. Although
considering the box left a positive connective leads to a complete
proof system, it has an inherited negative behaviour that is ignored
when adopting such polarisation. Also, in the sequent rule
$\mathsf{k}$, the box right should be chosen, which gives it a
positive behaviour, also not taken into account in the focused system
proposed in~\cite{DBLP:conf/fossacs/ChaudhuriMS16}.

In contrast, this positive/negative behaviour of box left and right
rules is present in $\LNS_{\K}$. In fact, while $\Box_R$ is not
invertible, proposing a focused version of this rule would render the
resulting system incomplete. And, as mentioned before, the $\Box_L$
rule has the restriction that it can be applied only after a $\Box_R$
rule is applied, hence it has a positive behaviour. But once the new
component is created by the $\Box_R$ rule, moving the left boxed
formulae can be done in any order and this action is invertible, hence
negative.

Thus, although modal blocks do not correspond to focusing, it produces
a normal form that mimics the sequential behaviour and preserves the
inherent positive/negative flavor of the box modality. Focusing, on
the other side, produces normal forms that do not correspond to
sequent derivations, hence the proof space is much bigger in (focused)
nested systems than in (block form) linear nested systems.

%%% Local Variables: 
%%% mode: latex
%%% TeX-master: "TOCL"
%%% End: 

\section{Labelled line sequent systems and bipoles}\label{sec:labelled}
%!TEX root = TOCL.tex

A logical framework is a meta-language used for the specification of
deductive systems.  Embedding systems into frameworks allows for
determining/analysing/proving meta-level properties of the object
level specified systems. And, since logical frameworks often come with
automated procedures, the meta-level machinery can often be used for
proving properties of the embedded systems automatically.

Restricting our attention to logical systems, since a specific logic
gives rise to specific sets of rules in different calculi, it is very
important to: choose a suitable, general logical framework, able to
specify a representative class of systems/calculi; and determine
whether there is a general and adequate methodology for embedding
deductive systems into the chosen logical framework, so that
object-level properties can be uniformly proven.

In this section, we will show that linear logic
($\LL$)~\cite{Girard:1987uq} is a general and adequate framework for
specifying a paramount subset of linear nested systems presented in
this paper. This is a fundamental results which opens the possibility
of exploring meta-level properties for such logical systems by
extending similar results obtained for sequent
systems~\cite{miller02tableaux,miller04lc,pimentel05lpar,DBLP:journals/tcs/MillerP13,nigam14jlc}.

One of the main advantages of the $\LNS$ calculi over the standard
sequent calculi is that the modal operators have separate left and
right rules, and that the number of principal formulae in the modal
rules is bounded. While the better control on moving formulae on
nested sequents facilitates the suggestion of a general method for
embedding $\LNS$ systems, the locality of the rules makes the quest of
proving \emph{adequacy} of the encodings harder. In fact, determining
adequate embedding maps on linear nested sequents that smoothly extend
existing ones on sequents is a non trivial task, as we will show next.

We start by reformulating the $\LNS$ structure in the language of
\emph{labelled sequents} (see,
e.g.,~\cite{Vigano:2000,Negri:2005p878,Negri:2011}), using an
extension of the correspondence between nested sequents and labelled
tree sequents in~\cite{DBLP:conf/aiml/GoreR12}. We then show how to
use such labelled systems in order to generate bipole clauses in
linear logic which adequately correspond to $\LNS$ modal rules.

\subsection{Labelled systems}
Let $\sv$ a countable infinite set of \emph{state variables} (denoted
by $x, y, z, \ldots$), disjoint from the set of propositional
variables. A \emph{labelled formula} has the form $x : A$ where
$x\in\sv $ and $A$ is a formula.  If $\Gamma = \{A_1,\ldots, A_n\}$ is
a multiset of formulae, then $x :\Gamma$ denotes the multiset
$\{x : A_1,\ldots , x : A_n\}$ of labelled formulae.  A (possibly
empty) set of \emph{relation terms} (\ie\ terms of the form $xRy$,
where $x,y\in\sv$) is called a \emph{relation set}. For a relation set
$\Rscr$, the \emph{frame} $Fr(\Rscr)$ defined by $\Rscr$ is given by
$(|\Rscr|,\Rscr)$ where
$|\Rscr| = \{x\; |\; xRy \in \Rscr \mbox{ or } yRx \in \Rscr \mbox{
  for some }y\in\sv\}$.  We say that a relation set $\Rscr$ is
\emph{treelike} if the frame defined by $\Rscr$ is a tree or $\Rscr$
is empty. A treelike relation set $\Rscr$ is called \emph{linelike} if
each node in $\Rscr$ has at most one child.

\begin{definition}\label{def:lls} 
  A \emph{labelled line sequent} $\LLS$ is a labelled sequent
  $\Rscr,X\seq Y$ where
\begin{enumerate}
\item $\Rscr$ is linelike;
\item if $\Rscr=\emptyset$ then $X$ has the form $x_0\lab
  A_1,\ldots,x_0\lab A_n$ and $Y$ has the form $x_0 \lab B_1,\ldots, x_0
  \lab B_m$ for some  $x_0\in\sv$;
\item if $\Rscr\not =\emptyset$ then every state variable $x$ that
  occurs in either $X$ or $Y$ also occurs in $\Rscr$.
\end{enumerate}
  A \emph{labelled line sequent calculus} is a labelled sequent
  calculus whose initial sequents and inference rules are constructed
  from $\LLS$.
\end{definition}

Observe that, in $\LLS$, if $xRy\in\Rscr$ then $uRy\notin\Rscr$
and $xRv\notin\Rscr$ for any $u,v\in\sv$ such that $u\neq x$ and $v\neq y$. 

Since linear nested sequents form a particular case of nested
sequents, the algorithm given in~\cite{DBLP:conf/aiml/GoreR12} can be
used for generating $\LLS$ from $\LNS$, and vice versa.  However, one
has to keep the linearity property invariant through inference
rules. For example, the following labelled sequent rule
\[
\infer[\Box_R']{\Rscr,X,\seq Y,x\lab \Box A}{\Rscr, xRy,
  X\seq Y, y\lab A}
\]
where $y$ is fresh, is not adequate w.r.t.\ the system $\LNS_{\K}$,
since there may exist $z\in|\Rscr|$ such that $xRz\in\Rscr$. That is,
for labelled sequents in general, freshness alone is not enough for
guaranteeing unicity of $x$ in $\Rscr$.  And it does not seem to be
trivial to assure this unicity by using logical rules without side
conditions.  To avoid this problem, we slightly modify the framework
by restricting $\Rscr$ to singletons, that is, $\Rscr=\{xRy\}$ will
record only the two last components, in this case labelled by $x$ and
$y$, and by adding a base case $\Rscr=\{x_0Rx_1\}$ for $x_0,x_1$
different state variables when there are no nested components. The
rule for introducing $\Box_R$ then is
\[
\infer[\Box_R]{zRx,X,\seq Y,x\lab \Box A}{xRy,
  X\seq Y, y\lab A}
\]
with $y$ fresh. Note that this solution corresponds to recording the
history of the proof search up to the last two steps similar to what
is outlined in~\cite{Pfenning:2015TalkLFMTP}, hence 
we are adopting an end-active  version of $\LLS$.  

\begin{definition}
  An \emph{end-active $\LLS$} is a singleton relation set
  $\mathcal{R}$ together with a sequent $X \seq Y$ of labelled
  formulae, written $\mathcal{R},X \seq Y$.  The rules of an
  \emph{end-active $\LLS$ calculus} are constructed from end-active
  labelled line sequents such that the active formulae in a premiss
  $xRy,X \seq Y$ are labelled with $y$ and the labels of all active
  formulae in the conclusion are in its relation set.
\end{definition}

From now on, we will use the end-active version of the propositional
rules (see Fig.~\ref{fig:lls-eG}).
\begin{figure}[t]
  \hrule
  \[
  \infer[\mathsf{init}]{zRx,X, x\lab p \seq x\lab p, Y}{}
  \quad
   \infer[\bottom_L]{zRx,X, x\lab \bottom \seq  Y}{}
  \quad
    \infer[\top_R]{zRx,X \seq x\lab \top, Y}{} 
    \]
    \[
      \infer[\neg_L]{zRx,X, x\lab \neg A, \seq  Y}{zRx,X \seq  Y, x\lab A}
  \quad
  \infer[\neg_R]{zRx,X \seq  Y,  x\lab\neg A}{zRx,X, x\lab A, \seq  Y} \quad
    \infer[\lor_L]{zRx,X,x\lab A\lor B \seq  Y}{zRx,X, x\lab A \seq  Y &
  zRx,X, x\lab B \seq  Y}  
  \]
  \[
   \infer[\lor_R]{zRx,X \seq  Y,  x\lab A\lor B}{zRx,X \seq Y, x\lab A, x\lab B}
\quad
  \infer[\land_L]{zRx,X, x\lab A\land B \seq  Y}{zRx,X, x\lab A, x\lab B \seq Y}
  \quad
  \infer[\land_R]{zRx,X \seq  Y,x\lab A\land B}{zRx,X \seq  Y,x\lab A &
  zRx,X \seq  Y,x\lab B}  
  \]
  \[
   \infer[\iimp_L]{zRx,X, x\lab A\iimp B \seq Y}
   {zRx,X \seq Y, x\lab A\quad zRx,X, x\lab  B \seq Y}
   \quad
    \infer[\iimp_R]{zRx,X \seq Y, x\lab A\iimp B}
   {zRx,X,x\lab A \seq Y, x\lab B}
    \]
      \hrule
\caption{The end-active version of $\LLS_{\G}$. In rule $\init$, $p$ is atomic.}\label{fig:lls-eG}
\end{figure}

We will now show how to automatically generate $\LLS$ from
$\LNS$. This is possible since the key property of end-active $\LNS$
calculi is that rules can only move formulae ``forward'', that is,
either an active formula produces other formulae in the same component
or in the next one.

\begin{definition}
  For a state variable $x$, define the mapping $\mathbb{TL}_{x}$ from
  $\LNS$ to end-active $\LLS$ as follows
\[
\begin{array}{lcl}
\mathbb{TL}_{x_1}(\Gamma_1 \seq \Delta_1)&=&x_0Rx_1,x_1\lab \Gamma_1 \seq x_1\lab \Delta_1\\
\mathbb{TL}_{x_n}(\Gamma_1 \seq \Delta_1\lns\ldots\lns \Gamma_n \seq \Delta_n) &=& x_{n-1}Rx_n, x_1\lab \Gamma_1,\ldots,x_n\lab \Gamma_n \seq x_1\lab \Delta_1,\ldots,
x_n\lab \Delta_n\quad n>0
\end{array}
\]
with all state variables pairwise distinct.
\end{definition}

We can use $\mathbb{TL}_{x}$ in order to construct a $\LLS$ inference
rule from an inference rule of an end-active $\LNS$ calculus.  The
procedure, that can also be automatised, is the same as the one
presented in~\cite{DBLP:conf/aiml/GoreR12}, as we shall illustrate
here.

\begin{example} 
  Consider the following application of the rule $\Box_R$ of
  Fig.~\ref{fig:lns-K}:
  \[
    \infer[\Box_R]{
  \Gamma_1 \seq \Delta_1 \lns \dots \lns \Gamma_{n-1} \seq \Delta_{n-1}
  \lns{\Gamma_{n}\seq \Delta_{n},\Box A}}
{
  \Gamma_1 \seq \Delta_1 \lns \dots \lns \Gamma_{n-1} \seq \Delta_{n-1}
      \lns \Gamma_{n}\seq\Delta_{n} \lns\;\seq A}
  \]
  Applying $\mathbb{TL}_x$ to the conclusion we obtain
  $x_{n-1}Rx_{n}, X \seq Y, x_{n}\lab \Box A$, where
  $X=x_1\lab \Gamma_1, \dots, x_{n}\lab \Gamma_{n} $ and
  $Y = x_1\lab \Delta_1, \dots, x_{n}\lab \Delta_{n}$.  Applying
  $\mathbb{TL}_x$ to the premise we obtain
  $x_nRx_{n+1}, X \seq Y, x_{n+1}\lab A$.  We thus obtain an
  application of the $\LLS$ rule
  \[
    \infer[\mathbb{TL}_x(\Box_R)]{x_{n-1}Rx_n,X\seq Y,x_n\lab \Box
      A}{x_nRx_{n+1},X\seq Y,x_{n+1}\lab A}
  \]
  Fig.~\ref{fig:LLS-K} presents the end-active labelled line sequent
  calculus $\LLS_\K$ for $\K$.
\end{example}

The following result follows readily by transforming derivations
bottom-up.

\begin{theorem}\label{trans}
  $\Gamma\seq\Delta$ is provable in a certain end-active $\LNS$
  calculus if and only if $\mathbb{TL}_{x_1}(\Gamma\seq\Delta)$ is
  provable in the corresponding end-active $\LLS$ calculus.
\end{theorem}

Note that, in an end-active $\LLS$, state variables might occur in the
sequent and not in the relation set. Such formulae will remain
inactive towards the leaves of the derivation and absorbed by the
initial sequents in systems where weakening is admissible.
\begin{figure}[t]
  \hrule
  \smallskip
  \[
    \infer[\mathbb{TL}_x(\Box_L)]{xRy,X,x\lab \Box A \seq Y}{xRy,X, y\lab A \seq Y}
     \qquad
       \infer[\mathbb{TL}_x(\Box_R)]{zRx,X\seq Y,x\lab \Box A}{xRy,X\seq Y,y\lab A}
 \]
  \hrule
  \caption{The modal rules of $\LLS_\K$.
The variable $y$ in rule $\Box_R$ is fresh.}
  \label{fig:LLS-K}
\end{figure}

The concepts of $\LLS$ and $\mathbb{TL}_{x}$ can be extended in order
to handle the extensions $\LNS_\m$ and $\LNS_\e$ by adding the
relations $R_\m\subseteq\sv \times\sv$ and
$R_\e\subseteq\sv\times(\sv\times\sv)$, respectively, and defining

\vspace{.2cm}
\resizebox{.98\textwidth}{!}{
$
\begin{array}{lcl}
\mathbb{TL}_{x_n}^{\m}(\Gamma_1 \seq \Delta_1\lns\ldots \lnsm \Gamma_n\seq \Delta_n) &=& x_{n-1}R_{\m}x_n, x_1\lab \Gamma_1,\ldots,x_n\lab \Gamma_n \seq x_1\lab \Delta_1,\ldots,
x_n\lab \Delta_n\\ \\
\mathbb{TL}_{x_n}^{\e}(\Gamma_1 \seq \Delta_1\lns\ldots \lnse( \Sigma\seq \Pi; \Omega
    \seq \Theta)) &=& x_{n-1}R_{\e}(x_n,y_n), x_1\lab \Gamma_1,\ldots,x_n\lab \Sigma, y_n\lab \Omega \seq x_1\lab \Delta_1,\ldots,
x_n\lab \Pi,y_n\lab \Theta
\end{array}
$}

\vspace{.2cm}
\noindent The corresponding $\LLS$ rules for these systems are
depicted in Figs.~\ref{fig:non-normal-LLS},~\ref{fig:monotone-LLS}
and~\ref{fig:monotone-extensions-LLS}.  Observe that this is a
generalisation of the algorithm in~\cite{DBLP:conf/aiml/GoreR12}, with
the careful remark that, in the case of non-normal systems, the
algorithm generates premisses that are weakened w.r.t.  the ones
presented in Fig.~\ref{fig:non-normal-LLS}.  Thus, Theorem~\ref{trans}
is also valid for all the $\LLS_{\m}$ and $\LLS_{\e}$ systems
presented in this work. Finally, it is worth noticing that the
definition of the mapping $\mathbb{TL}_x$ for auxiliary nesting
operators is the same as the respective final nesting operators.
\begin{figure}[t]
  \hrule
  \[
  \infer[\mathbb{TL}_{x}^{\e}(\Box_L^\e)]{x_{n-1}R_\e (x_{n},y_n),X, x_{n-1}\lab \Box A\seq Y}{x_{n-1}R x_{n},X,x_n\lab A\seq Y\quad x_{n-1}R y_{n},X\seq Y,y_n\lab A}
\]
\[ \infer[\mathbb{TL}_{x}^{\e}(\Box_R^\e)]{x_{n-1}Rx_n,X\seq Y, x_n\lab \Box B}{x_nR_\e (x_{n+1},y_{n+1}),X,y_{n+1}\lab  B\seq Y, x_{n+1}\lab  B}
  \quad
  \infer[\mathbb{TL}_{x}^{\e}(\mathsf{N})]{x_{n-1}Rx_n,X\seq Y,x_n\lab\Box B}{x_{n}R x_{n+1},X\seq Y,x_{n+1}\lab B}
 \]
 \[
 \infer[\mathbb{TL}_{x}^{\e}(\mathsf{M})]{x_{n-1}R_\e (x_n,y_n),X\seq Y}{x_{n-1}R_\e (x_n,y_n),X,y_n\lab\bottom \seq Y}
\quad
  \infer[\mathbb{TL}_{x}^{\e}(\mathsf{C})]{x_{n-1}R_\e (x_{n},y_n),X,x_{n-1}\lab\Box A \seq Y}{x_{n-1}R_\e(x_{n},y_n),X,
  x_{n}\lab A \seq Y \quad
 x_{n-1}Ry_{n},X
   \seq Y,y_{n}\lab A}
  \]
  \hrule
  \caption{System $\LLS_{\e}$ for non-normal labelled systems}
  \label{fig:non-normal-LLS}
\end{figure}

\begin{figure}[t]
  \hrule
  \[
  \infer[\mathbb{TL}_{x}^{\m}(\Box_R^\m)]{x_{n-1}Rx_n,X\seq Y, x_n\lab \Box B}{x_nR_\m x_{n+1},X\seq Y, x_{n+1}\lab  B}
  \qquad
  \infer[\mathbb{TL}_{x}^{\m}(\Box_L^\m)]{x_{n-1}R_\m x_{n},X, x_{n-1}\lab \Box B\seq Y}{x_{n-1}R x_{n},X,x_n\lab B\seq Y}
\]
\[
  \infer[\mathbb{TL}_{x}^{\m}(\mathsf{C})]{x_{n-1}Rx_n,X\seq Y}{x_{n-1}R_\m x_n,X\seq Y}
 \qquad
 \infer[\mathbb{TL}_{x}^{\m}(\mathsf{N})]{x_{n-1}R_\m x_n,X\seq Y}{x_{n-1}R x_n,X\seq Y}
 \]
   \hrule
  \caption{$\LLS_{\m}$  for monotone labelled systems}
  \label{fig:monotone-LLS}
\end{figure}

\begin{figure}[t]
  \hrule
  \[
  \infer[\mathbb{TL}_{x}^{\m}(\mathsf{P})]{x_{n-1}Rx_n,X\seq Y}{x_nR_\m x_{n+1},X\seq Y}
  \qquad
  \infer[\mathbb{TL}_{x}^{\m}(\mathsf{T})]{x_{n-1}R_\m x_{n}, X,x_{n}\lab \Sigma\seq Y,x_{n}\lab \Pi}{x_{n}R x_{n+1}, X,x_{n+1}\lab \Sigma\seq Y,x_{n+1}\lab \Pi}
\]
   \hrule
  \caption{Labelled systems for extensions of monotonic modal logics}
  \label{fig:monotone-extensions-LLS}
\end{figure}

%%% Local Variables: 
%%% mode: latex
%%% TeX-master: "TOCL"
%%% End: 

\subsection{Bipoles}\label{sec:ps}
%!TEX root = TOCL.tex
In this section we exploit the above mentioned fact that $\LNS$
systems often have separate left and right introduction rules for
modalities in order to present a systematic way of representing
labelled line nested rules as \emph{bipole clauses}.  For that, we
will use (focused) linear logic ($\LLF$), not only because it extends
the works in, e.g.,~\cite{DBLP:journals/tcs/MillerP13,nigam14jlc}, but
also since this is the basis for using the rich linear logic
meta-level theory in order to reason about the specified systems. It
is worth noticing, though, that our approach is general enough for
specifying inference rules in other frameworks, like $\LKF$
(\cite{DBLP:conf/lpar/MillerV15,DBLP:conf/aiml/MarinMV16}).
  
The set of \emph{formulae} of $\LLF$ is given by the following grammar:
\[
 F ::= p \mid p^\perp \mid \one \mid \zero \mid\top \mid \bot  \mid F_1 \tensor F_2 \mid F_1 \lpar 
  F_2 \mid F_1 \with F_2 
  \mid F_1 \oplus F_2
  \mid \exists x. F \mid \forall x. F
  \mid \quest F \mid \bang F
\]
The connectives $\bot, \top, \with, \lpar, \forall, \quest$ are taken
to be \emph{negative}, the connectives
$\one,\zero,\tensor,\oplus,\exists,\bang$ are considered to be
\emph{positive}. The notions of negative and positive polarities are
extended to formulae in the natural way by considering the outermost
connective. Formulae are taken to be in \emph{negation normal form}
using the standard classical linear logic dualities, e.g.,
$(A \tensor B)^\perp \equiv A^\perp \lpar B^\perp$. Sequents in
(one-sided) linear logic are multisets of linear logic
formulae. \emph{Focused Linear Logic} $\LLF$ then adds a focusing
mechanism to this structure (see Sec.~\ref{sec:foc}). We refer the
reader to~\cite{Girard:1987uq} for the rules of unfocused linear logic
and to~\cite{Andreoli:1992,DBLP:journals/tcs/MillerP13} for the
focused versions.

%%% Local Variables: 
%%% mode: latex
%%% TeX-master: "TOCL"
%%% End: 

%!TEX root = TOCL.tex

\subsubsection{Specifying sequents} We briefly recapitulate the basic
concepts of the specification of sequent-style calculi in $\LLF$
from~\cite{DBLP:journals/tcs/MillerP13}. Let $\obj$ be the type of
object-level formulae and let $\lft{\cdot}$ and $\rght{\cdot}$ be two
meta-level predicates on these, i.e., both of type
$\obj\rightarrow o$, where $o$ is a primitive type denoting formulas.
Object-level sequents of the form
${B_1,\ldots,B_n}\seq{C_1,\ldots,C_m}$ (where $n,m\ge0$) are specified
as the multiset
$\lft{B_1},\ldots,\lft{B_n},\rght{C_1},\ldots,\rght{C_m}$ within the
LLF proof system.  The $\lft{\cdot}$ and $\rght{\cdot}$ predicates
identify which object-level formulas appear on which side of the
sequent -- brackets down for left (useful mnemonic: $\lfloor$ for
``left'') and brackets up for right. Finally, a binary relation $R$ is
specified by a meta-level atomic formula of the form $R(\cdot,\cdot)$.

\subsubsection{Specifying inference sequent rules}
Inference rules are specified by a re-writing clause that replaces the
active formulae in the conclusion by the active formulae in the
premises. The linear logic connectives indicate how these object level
formulae are connected: contexts are copied ($\with$) or split
($\otimes$), in different inference rules ($\oplus$) or in the same
sequent ($\lpar$).  For example, the specification of the rules of
$\LLS_\K$ (Fig. \ref{fig:LLS-K}) is
 \[
\begin{array}{l@{\quad}ll}
 (\Box_R) & \rght{x\lab \Box A}^\perp \otimes R(z,x)^\perp&\otimes \ \forall y.(\rght{y\lab A}\lpar R(x,y)) \\
 (\square_L)&  \lft{x\lab  \square A}^{\bot}
 \otimes R(x,y)^\perp& \tensor \ \lft{y\lab  A}\lpar R(x,y)\\
\end{array}
 \]
where all the variables are bounded by an outermost existential quantifier.

The correspondence between focusing on a formula and an induced
big-step inference rule is particularly interesting when the focused
formula is a \emph{bipole}. 

\begin{definition}\label{monopoles}
  A \emph{monopole} formula is a linear logic formula that is built up
  from atoms and occurrences of the negative connectives, with the
  restriction that $\quest$ has atomic scope.  A \emph{bipole} is a
  positive formula built from monopoles and negated atoms using only
  positive connectives, with the additional restriction that $\bang$
  can only be applied to a monopole.
\end{definition}

Roughly speaking, bipoles are positive formulae in which no positive
connective can be in the scope of a negative one. Focusing on such a
formula will produce a single positive and a single negative
phase. This two-phase decomposition enables the adequate capturing of
the application of an object-level inference rule by the meta-level
logic.  For example, focusing on the bipole clause $ (\Box_R) $ will
produce the derivation
\[
\infer=[{[\exists,\otimes]}]{\Down{\Psi}{\Delta}{\exists A,x,z.\rght{x\lab \Box A}^\perp\otimes R(z,x)^\perp\otimes\forall y.(\rght{y\lab A}\lpar R(x,y))}}
{\quad\pi_1\quad&\quad\pi_2\quad&\quad\infer=[{[R\Uparrow,\forall,\lpar]}]{\Down{\Psi}{\Delta'}{\forall y.(\rght{y\lab A}\lpar R(x,y))}}
{\deduce{\Up{\Psi}{\Delta',\rght{y\lab A},R(x,y))}}{}}\quad}
\]
where $\Delta = \rght{x\lab \Box A}\cup R(z,x)\cup\Delta'$, and
$\pi_1$ and $\pi_2$ are, respectively,
\[
\infer[I_1]{\Down{\Psi}{\rght{x\lab \Box A}}{\rght{x\lab \Box A}^\perp}}{}\qquad\quad
\infer[{[\exists,I_1]}]{\Down{\Psi}{R(z,x)}{R(z,x)^\perp}}{}
\]
This one-step focused derivation will: (a) consume
$\rght{x\lab \Box A}$ and $R(z,x)$; (b) create a fresh label $y$; and
(c) add $\rght{y\lab A}$ and $R(x,y)$ to the context. Observe that
this matches \emph{exactly} the application of the object-level rule
$\mathbb{TL}_x(\Box_R)$.

When specifying a system (logical, computational, etc) into a
meta-level framework, it is desirable and often mandatory that the
specification is \emph{faithful}, that is, one step of computation on
the object level should correspond to one step of logical reasoning in
the meta level.  This is what is called
\emph{adequacy}~\cite{nigam10jar}.

\begin{definition}
  A specification of an object sequent system is \emph{adequate} if
  provability is preserved for (open) derivations, such as inference
  rules themselves.
\end{definition}

Clearly not every sequent rule can be (adequately) specified in
$\LLF$. As an example, the rule $\mathbb{TL}_{x}^{\m}(\mathsf{T})$
(Fig.~\ref{fig:monotone-extensions-LLS}) cannot be properly specified
in our setting, since it lacks a principal formula.  All other $\LLS$
rules derived from $\LNS$ systems presented in this paper can be
adequately specified.

As an example, Fig.~\ref{fig:spec-EC} shows adequate specifications in
LLF of the labelled systems for the logic $\mathsf{EC}$.  These
specifications can be used for automatic proof search as illustrated
by the following theorem which is shown readily using the methods
in~\cite{DBLP:journals/tcs/MillerP13}.
\begin{figure}[t]
  \hrule
  \[
    \begin{array}{l@{\quad}l}
    (\Box_R^\e) & \rght{x:\Box B}^\perp\otimes R(w,x)^\perp\otimes \forall y \forall
                  z.(\rght{y:B}\lpar\lft{z:B}\lpar R_\e(x,(y,z)))\\
                  \\
    (\Box_L^\e) & \lft{x: \Box A}^{\bot} \otimes R_\e(x,
(y,z))^{\bot}\tensor (\lft{y:  A}\lpar R(x,y))
    \otimes (\rght{z:  A}\lpar R(x,z))
  \\
  \\
    (\mathsf{C}) & \lft{x: \Box A}^{\bot} \tensor  R_\e(x,
(y,z))^{\bot}\otimes (\lft{y:
       A}\lpar R_{\e}(x,(y,z)))
    \otimes(\rght{z:
     A}\lpar R(x,z))
  \end{array}
  \]
\hrule
  \caption{The LLF specification of the modal rules of 
    $\LLS_{\mathsf{EC}}$ for the logic $\mathsf{EC}$.}
   \label{fig:spec-EC}
\end{figure}

\begin{theorem}
  Let $L$ be a $\LLS$ system.  A sequent $\Rscr,\Gamma\seq\Delta$ is
  provable in $L$ if and only if there is a finite $L_0\subseteq L$
  with $\mathcal{L}_0$ the theory given by the clauses of an adequate
  specification of the inference rules of $L_0$ such that
  $\Up{\Lscr_0}{\Rscr}{\lft{\Gamma},\rght{\Delta}}$ is provable in
  LLF.
\end{theorem}

It turns out that the encoding of $\LNS$ into $\LL$ enabled the
proposal of a general theorem prover.  The system (called POULE for
\emph{PrOver for seqUent and Labelled systEms} -- available at
\url{http://subsell.logic.at/nestLL/}) has an $\LLF$ interpreter that
takes specified $\LLS$ rules ($\LLF$ clauses -- the theory) and
sequents and outputs a proof of the sequent, if it is provable.

The prover is \emph{parametric} in the theory, hence it profits from
the modularity of the specified systems.  Indeed, since the core is a
prover for focused linear logic, POULE can be transformed into a
specific prover for each specified logic by parametrically adding the
encoded rules of the object system as theories in $\LL$. Hence, for
example, if the prover should validade theorems in $\K$, one only has
to add the encoding of the system (as a theory) to the prover.

It is well known that \emph{generality} often implies
\emph{inefficiency}, and POULE is no exception to that.  Hence we have
also done a direct implementation of $\LNS$ systems in Prolog,
parametric on the modal axioms, that can be found in
\url{https://logic.at/staff/lellmann/lnsprover/}. We have no intention
of comparing such implementations, since they are different in nature:
a direct prover built from axioms is more adequate for proving
\emph{theorems}, while the meta-level prover based in $\LL$ is
suitable for proving \emph{properties}.

%%% Local Variables:
%%% mode: latex
%%% TeX-master: "TOCL"
%%% End:

\section{Concluding remarks and future work}\label{sec:conc}
%!TEX root = TOCL.tex

Following~\cite{Masini:1992}, in~\cite{Lellmann:2015lns} linear nested
sequents were considered as an alternative presentation for modal
proof systems. Since locality often entails modularity, this enabled
modular presentations for different modal systems just by adding the
local rules related to the new modalities to already defined linear
nested sequent systems.  In~\cite{DBLP:conf/lpar/LellmannP15} we
continued the programme of representing modal proof systems in $\LNS$,
including suitable extensions of $\K$, a simply dependent bimodal
logic and some standard non-normal modal logics.

In this paper, we have generalised the works \emph{op. cit.},
presenting local systems for a family of simply dependent multimodal
logics as well as a large class of non-normal modal logics. All the
proposed systems were proven sound and complete w.r.t. the respective
sequent systems and, as a side effect, we proved that each $\LNS$
system presented in this work could be restricted to its end-active
version. This enabled a notion of normal forms for $\LNS$ derivations,
narrowing the proof search space and hence allowing the proposal of
more efficient local proof systems.  The possibility of restricting
systems to their end-active versions also entails an automatic
procedure for obtaining labelled sequent versions of $\LNS$
systems. Finally, we showed that the inference rules of such labelled
systems can be seen as bipoles, and hence are amenable to adequate
embeddings into linear logic, which enabled the implementation of a
general theorem prover, parametric in the modal axioms considered.

There are at least four future research directions that could be taken
from this work.

First, following the works
in~\cite{DBLP:journals/tcs/MillerP13,nigam14jlc}, it should be
possible to use some of the meta-theory of linear logic to draw
conclusions about the object-level $\LNS$ systems.  For example, the
problem of providing general procedures for guaranteeing cut
admissibility for nested systems is still little understood.  It turns
out that the cut rule has an inherent duality: the cut formula is both
a conclusion of a statement and an hypothesis of another. In sequent
systems, this duality is often an invariant, being preserved
throughout the cut elimination process. Developing general methods for
detecting such invariants enables the use of meta-level frameworks to
uniformly reason about object-level
properties. In~\cite{DBLP:journals/tcs/MillerP13}, bipoles and
focusing were enough for both: specifying sequent systems and
providing sufficient meta-level conditions for cut elimination. In
this work, we showed that the bipole-focusing approach can be extended
for specifying nested and labelled systems. However, the meta-level
characterisation of cut elimination invariants for these formalisms is
still an open problem.

Second, it would be interesting to see to what extend the labels in
$\LLS$ reflect the semantic models behind the studied logics. In the
labelled sequent framework, Kripke's relational
semantics\cite{kripke63} is explicitly added to sequents, so that
labels correspond to
\emph{worlds}~\cite{Vigano:2000,Negri:2005p878}. In $\LNS$ instead,
labels correspond to the \emph{depth} of the nesting.  Gor\'{e} and
Ramanayake in~\cite{DBLP:conf/aiml/GoreR12} presented a direct
translation between proofs in labelled and nested systems for some
normal modal logics, while Fitting in~\cite{Fitting:2014} showed how
to relate nestings with Kripke structures for intuitionistic logic.
We believe it is possible to extend these ideas for relating not only
(general) frames with (classical) labelled multi-modal logics, but
also neighbourhood semantics labelled systems~\cite{negri2017} with
(linear) nested systems for non-normal modal logics.

A next natural topic for investigation would be to study this proof
system/semantics problem in the substructural setting. The proof
theoretical problem of modalities over linear logic has been first
addressed in~\cite{DBLP:journals/igpl/GuerriniMM98}, where different
behaviours of the exponential connectives in linear logic were
studied. We have recently~\cite{DBLP:conf/lpar/LellmannOP17} extended
Guerrini's work in order to show how normal multi-modalities can be
added to linear logic, having, as a side effect, a notion of modality
that smoothly extends that of
subexponentials~\cite{danos93kgc,DBLP:journals/tcs/NigamOP17}.  Among
the many problems that are still open in this subject, two are of
particular interest: is it possible to extend substructural logics
with non-normal modalities in a general way, similar to what is done
in~\cite{Porello:2015}? If so, which is the semantical meaning of the
resulting logics?  We plan to investigate these problems under the
linear logic view.

Last but not least, concerning the construction of the $\LNS$ systems
themselves, a natural next step would be the investigation of general
methods for obtaining such systems from cut-free sequent systems, or
even directly from Hilbert-style axioms.

\paragraph{Acknowledgments}
We thank the anonymous reviewers for their valuable comments on an
earlier draft of this paper.

%%% Local Variables: 
%%% mode: latex
%%% TeX-master: "TOCL"
%%% End: 

\bibliographystyle{ACM-Reference-Format-Journals}
\bibliography{biblioB}

\newpage
%!TEX root = TOCL.tex

\appendix

\section{Proofs of Cut Elimination}
\label{app:cut-elim-mon}

Since the calculi include the contraction rules we follow the standard
method of eliminating the \emph{multicut} rule
\[
\infer[\mcut]{\Gamma,\Sigma \seq \Delta,\Pi}{\Gamma \seq \Delta, A^n &
A^m, \Sigma \seq \Pi}
\]
(with $n,m \geq 1$) instead of the standard cut rule. As usual, since
the latter is a specific instance of the multicut rule, this implies
cut elimination. For the sake of exposition we deviate slightly from
standard terminology in the following way.

\begin{definition}
  The \emph{main formula} of an application of a propositional rule or
  the modal rule $\T$ from Fig.~\ref{fig:sequent-rules-ext-mon} is the
  formula occurring in the conclusion with a greater multiplicity than
  in any of the premisses.  In particular, the \emph{main formulae} of
  an application of $\init$ or $\bot_L$ are all formulae occurring in
  the conclusion.  The \emph{main formulae} of an application of a
  modal rule from Fig.~\ref{fig:sequent-rules-ext-mon} apart from $\T$
  are all the formulae occurring in the conclusion. In an application
  of a structural rule, i.e., Weakening or Contraction, there are no
  main formulae.
\end{definition}

Hence, e.g., the formula $\Box A$ in would be a main formula in the
applications of rules $\D$ and $\D\4$ below left and centre, but not
in the application of rule $\T$ below right.
\[
  \infer[\D]{\Box A, \Box B \seq \;}{A, B \seq \;}
  \qquad
  \infer[\D\4]{\Box A,\Box B \seq\; }{\Box A, B \seq \;}
  \qquad
  \infer[\T]{\Box A, \Box B \seq \;}{\Box A, B \seq \;}
\]

In the statement of the cut elimination theorem we write
$\G_\logic\Con\W\mcut$ for the calculus $\G_\logic\Con\W$ with the
multicut rule $\mcut$.

\begin{theorem}
  Let $\logic$ be the logic $\M\mathcal{A}$ with
  $\mathcal{A} \subseteq \{ \mathsf{N}, \C, \PP, \D, \4 \}$ or one of
  the logics
  $\{ \M\5, \M\PP\5, \M\45,\M\PP\4\5, \M\D\4\5, \K\4\5, \K\D\4\5
  \}$. Then every derivation in $\G_\logic \Con\W\mcut$ can be
  converted into a derivation in $\G_\logic\Con\W$ with the same
  endsequent.
\end{theorem}

\begin{proof}
  The proof of elimination of multicut is reasonably standard by
  induction on the tuples $(c,d)$ in the lexicographic ordering $\lo$,
  where $c$ is the \emph{complexity} of the application of multicut,
  i.e., the number of symbols in the cut formula, and $d$ is its
  \emph{depth}, i.e., the sum of the depths of the derivations of the
  two premisses of the application of multicut.

  So take a topmost application
\[
  \infer[\mcut]{\Gamma, \Sigma \seq \Delta, \Pi
  }
  {\infer[R_1]{\Gamma \seq \Delta, A^n}{\infer*{}{\mathcal{D}_1}}
    &\qquad
    \infer[R_2]{A^m,\Sigma \seq \Pi}{\infer*{}{\mathcal{D}_2}}
  }
\]
of multicut in a derivation in $\G_\logic \Con\W\mcut$. Assume that this
application of multicut is of complexity $c$ and depth $d$, that
$\mathcal{D}_1$ and $\mathcal{D}_2$ are the derivations of the two
premisses of this application, and that $R_1$ and $R_2$ are the two
last applied rules in $\mathcal{D}_1$ and $\mathcal{D}_2$
respectively. Furthermore, assume that we have shown the
statement for applications of multicut with complexity $c'$ and depth
$d'$ such that $(c',d') \lo (c,d)$. 

If $d = 0$, then both $R_1$ and $R_2$ are one of the rules $\init$ or
$\bot_L$ and the conclusion of the multicut is obtained directly by
one of these rules.

So suppose $d > 0$. We distinguish cases according to whether an
occurrence of the cut formula was a main formula in the last applied
rule in $\mathcal{D}_1$ and $\mathcal{D}_2$ respectively.
\begin{enumerate}
\item No occurrence of the cut formula is a main formula in $R_1$. In
  this case $R_1$ is a structural rule, the rule $\T$, or a propositional
  rule apart from $\init, \bot_L$. This case is handled as usual by pushing the
  multicut into the premiss(es) of $R_1$ and applying the induction
  hypothesis on the depth of the application of multicut. E.g., if
  $R_1$ is $\lor_L$, the derivation $\mathcal{D}_1$ ends in
  \[\infer[\lor_L]{\Gamma',B \lor C \seq \Delta, A^n}{\Gamma',B \seq
    \Delta, A^n & \quad \Gamma', C \seq \Delta, A^n}
  \]
  From this we obtain a new derivation ending in
  \[
    \infer[\lor_L]{\Gamma', B \lor C,\Sigma,\Sigma \seq \Delta, \Pi
     }
     {\infer[\mcut]{\Gamma',B, \Sigma \seq \Delta, \Pi
       }
       {\Gamma',B \seq \Delta, A^n
         &\quad
         \infer*[]{A^m,\Sigma\seq\Pi}{\mathcal{D}_2}
       }
     &\qquad
     \infer[\mcut]{\Gamma',C, \Sigma \seq \Delta, \Pi
       }
       {\Gamma',C \seq \Delta,A^n
         &\quad
        \infer*[]{A^m,\Sigma\seq\Pi}{\mathcal{D}_2}
       }
     }
  \]
  Now the two applications of multicut have complexity $c$ and depth
  less than $d$ and we are done using the induction hypothesis.

\item At least one occurrence of the cut formula is a main formula in
  $R_1$, but none of its occurrences is a main formula in $R_2$. This
  case is analogous to the previous case, but pushing the multicut
  into the premiss(es) of $R_2$ instead of $R_1$.

\item Some occurrences of the cut formula are main formulae both in
  $R_1$ and $R_2$. In this case we have $c > 1$, since for $c = 1$ the
  cut formula $A$ is a propositional variable or $\bot$, and since
  some of its occurrences are main formulae both in $R_1$ and $R_2$ we
  would have $d = 0$. So the rules $R_1,R_2$ must be propositional
  rules apart from $\init,\bot_L$ or modal rules. As usual we
  distinguish cases according to the last applied rules $R_1,R_2$, and
  first apply \emph{cross-cuts}, i.e., multicuts on the premiss(es) of
  $R_1$ and the conclusion of $R_2$ and vice versa to eliminate
  occurrences of the cut formula from the premisses of the two
  rules. These multicuts then have smaller depth and are eliminated
  using the induction hypothesis. Then we reduce the complexity of the
  multicut. Since the propositional cases are standard we only treat
  an exemplary case.
  \begin{enumerate}
  \item $R_1 = \lor_R$ and $R_2 = \lor_L$. Then the derivations
    $\mathcal{D}_1$ and $\mathcal{D}_2$ end in
    \[
      \infer[\lor_R]{\Gamma \seq \Delta,B \lor C^n
      }
      {\Gamma \seq \Delta, B \lor C^{n-1},B,C
    }
    \qquad
    \infer[\lor_L]{B \lor C^m,\Sigma \seq \Pi
    }
    {B, B \lor C^{m-1},\Sigma \seq \Pi
      &\quad
      C, B \lor C^{m-1},\Sigma \seq \Pi
    }
    \]
    From this we obtain derivations ending in
    \[
      \infer[\mcut]{\Gamma,\Sigma \seq \Delta,\Pi, B,C
      }
      {\Gamma \seq \Delta, B \lor C^{n-1}, B, C
        &\quad
       \infer[\lor_L]{B \lor C^m, \Sigma \seq \Pi
        }
        {B, B \lor C^{m-1}, \Sigma \seq \Pi
          &
          C, B \lor C^{m-1}, \Sigma \seq \Pi
        }
      }
    \]
    and
    \[
      \infer[\mcut]{D, \Gamma, \Sigma \seq \Delta, \Pi
      }
      {\infer[\lor_R]{\Gamma \seq \Delta, B \lor C^n
        }
        {\Gamma \seq \Delta, D \lor C^{n-1},B,C
        }
      &
      D, B\lor C^{m-1},\Sigma \seq \Pi  
      }
    \]
    with $D$ either of the formulae $C,D$. These two applications of
    multicut have complexity $c$ and depth less than $d$ and hence are
    eliminated using the induction hypothesis. From the resulting
    derivations finally we obtain a derivation ending in
    \[
      \infer=[\Con]{\Gamma,\Sigma \seq \Delta, \Pi
      }
      {\infer[\mcut]{\Gamma, \Sigma,\Gamma, \Sigma,\Gamma, \Sigma \seq
          \Delta, \Pi,\Delta, \Pi,\Delta, \Pi
        }
        {\infer[\mcut]{\Gamma, \Sigma, \Gamma, \Sigma \seq \Delta,
            \Pi, \Delta, \Pi, C
          }
          {\Gamma, \Sigma \seq \Delta, \Pi, B, C
          &
          B, \Gamma, \Sigma \seq \Delta, \Pi
          }
        &
        C, \Gamma, \Sigma \seq \Delta, \Pi
        }
      }
    \]
    Here the newly introduced applications of $\mcut$ have depth
    possibly greater than $d$ but complexity less than $c$, and hence
    also are eliminated using the induction hypothesis.
  \item $R_1 = \M$: In this case the logic is $\M\mathcal{A}$ for
    $\mathcal{A} \subseteq \{ \mathsf{N}, \PP, \D, \4 \}$ or one of
    $\{ \M\5, \M\PP\5, \M\4\5, \M\PP\4\5, \M\D\4\5 \}$. So $R_2$ is
    one of $\M, \PP, \D, \T, \4, \D4, \D5$. For the sake of brevity,
    in the following we only show the reductions of the multicuts,
    denoted by $\leadsto$.  As usual, the newly introduced multicuts
    have the same complexity but lower depth, or of lower complexity
    than the original ones.
    \begin{enumerate}
    \item $R_2 = (\M)$:
      \[
        \vcenter{
        \infer[\mcut]{\Box A \seq \Box C
        }
        {\infer[\M]{\Box A \seq \Box B}{A \seq B}
          &
          \infer[\M]{\Box B \seq \Box C}{B \seq C}
        }
      }
      \quad \leadsto \quad
      \vcenter{
      \infer[\M]{\Box A \seq \Box C
      }
      {\infer[\mcut]{A \seq C
        }
        {A \seq B
          &
          B \seq C
        }
      }
      }
      \]\label{case:mon}\label{case:MvsM}
    \item $R_2 = \PP$: similar to the previous case.\label{case:MvsP}
    \item $R_2 = \D$: The case where $m = 1$ is as above. If $m > 1$
      we only need to add some structural rules:
      \[
        \vcenter{
        \infer[\mcut]{\Box A \seq \;
        }
        {\infer[\M]{\Box A \seq \Box B}{A \seq B}
        &
        \infer[\D]{\Box B, \Box B\seq \; }{B, B \seq \;}
      }
      }
      \qquad
      \leadsto
      \qquad
      \vcenter{
        \infer[\Con]{\Box A \seq \;
      }
      {\infer[\D]{\Box A, \Box A \seq \;
        }
        {\infer[\W]{A, A \seq\;
          }
          {\infer[\mcut]{A \seq \;
            }
            {A \seq B
              &
              B, B \seq \;
            }
          }
        }
      }
    }
      \]\label{case:MvsD}
    \item $R_2 = \T$:
      \[
        \vcenter{
        \infer[\mcut]{\Gamma, \Box A \seq\Delta
        }
        {\infer[\M]{\Box A \seq \Box B
          }
          {A \seq B
          }
          &
          \infer[\T]{\Gamma, \Box B^m \seq \Delta
          }
          {\Gamma, \Box B^{m-1}, B \seq \Delta
          }
        }
      }
      \leadsto
      \vcenter{
        \infer[\T]{\Gamma, \Box A \seq \Delta
        }
        {\infer[\mcut]{\Gamma, A \seq \Delta
          }
          {A \seq B
            &
            \infer[\mcut]{\Gamma, B \seq \Delta
            }
            {\infer[\M]{\Box A \seq \Box B
              }
              {A \seq B
              }
              &
              \Gamma, \Box B^{m-1}, B \seq \Delta
            }
          }
        }
      }
      \]\label{case:MvsT}
    \item $R_2 = \4$:
      \[
        \vcenter{
        \infer[\mcut]{\Box A \seq
        }
        {\infer[\M]{\Box A \seq \Box B
          }
          {A \seq B
          }
          &
          \infer[\4]{\Box B \seq \Box C
          }
          {\Box B \seq C
          }
        }
        }
        \quad\leadsto\quad
        \vcenter{
        \infer[\4]{\Box A \seq \Box C
        }
        {\infer[\mcut]{\Box A \seq C
          }
          {\infer[\M]{\Box A \seq \Box B
            }
            {A \seq B
            }
            &
            \Box B \seq C
          }
        }
        }
      \]\label{case:Mvs4}
    \item $R_2 = \D\4$: We consider the case that $m = 2$. The other
      cases are as above.
      \[
        \vcenter{
        \infer[\mcut]{\Box A \seq
        }
        {\infer[\M]{\Box A \seq \Box B
          }
          {A \seq B
          }
          &
          \infer[\D\4]{\Box B,\Box B \seq \;
          }
          {\Box B, B \seq \;
          }
        }
        }
        \quad\leadsto\quad
        \vcenter{
        \infer[\Con]{\Box A \seq \;
        }
        {\infer[\D\4]{\Box A, \Box A \seq \;
          }
          {\infer[\mcut]{\Box A, A \seq \;
            }
            {A \seq B
              &
              \infer[\mcut]{\Box A, B \seq \;
              }
              {\infer[\M]{\Box A \seq \Box B
                }
                {A \seq B
                }
                &
                \Box B, B \seq \;
              }
            }
          }
        }
        }
      \]\label{case:MvsD4}
    \item $R_2 = \D\5$: as for $\M$\label{case:MvsD5}
       \end{enumerate}
  \item $R_1 = \mathsf{N}$:\label{case:N}
    \begin{enumerate}
    \item $R_2 = \M$: Similar to case~\ref{case:MvsP}
    \item $R_2 = \PP$: As in the previous case.
    \item $R_2 = \D$: The case where $m=2$ is similar to the previous
      case. For $m=1$ we have:
      \[
        \vcenter{
        \infer[\mcut]{\Box B\seq\;
        }
        {\infer[\mathsf{N}]{\;\seq \Box A
          }
          {\;\seq A
          }
          &
          \infer[\D]{\Box A, \Box B \seq \;
          }
          {A, B \seq \;
          }
        }
        }
        \quad\leadsto\quad
        \vcenter{
        \infer[\Con]{\Box B \seq \;
        }
        {\infer[\D]{\Box B, \Box B \seq \;
          }
          {\infer[\W]{B, B \seq \;
            }
            {\infer[\mcut]{B\seq \;
              }
              {\;\seq A
                &
                A, B \seq \;
              }
            }
          }
        }
        }
      \]
    \item $R_2 = \T$: Similar to case~\ref{case:MvsT}.\label{case:NvsT}
    \item $R_2 = \4$: Like case~\ref{case:Mvs4}.\label{case:Nvs4}
    \item $R_2 = \D\4$: The case with $m = 2$ is like
      case~\ref{case:MvsD4}. If $m = 1$ again we need some structural
      rules:
      \[
        \vcenter{
        \infer[\mcut]{\Box B \seq \;
        }
        {\infer[\mathsf{N}]{\;\seq \Box A
          }
          {\;\seq A
          }
          &
          \infer[\D\4]{\Box A, \Box B \seq \;
          }
          {\Box A, B \seq \;
          }
        }
      }
      \quad\leadsto\quad
      \vcenter{
      \infer[\Con]{\Box B \seq \;
      }
      {\infer[\D\4]{\Box B, \Box B \seq \;
        }
        {\infer[\W]{B, \Box B \seq \;
          }
          {\infer[\mcut]{B \seq \;
            }
            {\infer[\mathsf{N}]{\;\seq \Box A
              }
              {\;\seq A
              }
              &
              \Box A, B \seq \;
            }
          }
        }
      }
      }
      \]\label{case:NvsD4}
    \item $R_2 = \D\5$: As for the previous case.
    \item $R_2 = \C$: We have the following (substituting $\mathsf{N}$
      for $\C$ in the last step if $\Gamma$ is empty):
      \[
        \vcenter{
        \infer[\mcut]{\Box \Gamma \seq \Box B
        }
        {\infer[\mathsf{N}]{\;\seq \Box A
          }
          {\;\seq A
          }
          &
          \infer[\C]{\Box A^m, \Box \Gamma \seq \Box B
          }
          {A^m, \Gamma \seq B
          }
        }
        }
        \quad\leadsto\quad
        \vcenter{
        \infer[\C]{\Box \Gamma \seq \Box B
        }
        {\infer[\mcut]{\Gamma \seq B
          }
          {\;\seq A
            &
            A^m, \Gamma \seq B
          }
        }
        }
      \]
    \item $R_2 = \C\D$: As for the previous case.
    \item $R_2 = \C\4$: Similar to case~\ref{case:NvsK45}.
    \item $R_2 = \C\D\4$: Similar to case~\ref{case:NvsK45}.
    \item $R_2 = \K\4$: Similar to case~\ref{case:NvsK45}.
    \item $R_2 = \K\4\5$:
      \[
        \infer[\mcut]{\Box \Gamma, \Box \Sigma \seq \Box B, \Box \Delta
        }
        {\infer[\mathsf{N}]{\;\seq \Box A
          }
          {\;\seq A
          }
          &
          \infer[\K\4\5]{\Box A^m, \Box \Gamma, \Box \Sigma \seq \Box B, \Box \Delta
          }
          {\Box A^{m-k},  \Box \Gamma, A^k, \Sigma \seq B, \Box \Delta
          }
        }\hspace{15em}
      \]
      \[
        \hspace{5em}\quad\leadsto\quad
        \vcenter{
        \infer[\K\4\5]{\Box \Gamma, \Box \Sigma \seq \Box B, \Box \Delta
        }
        {\infer[\mcut]{\Box \Gamma, \Sigma \seq  B, \Box \Delta
          }
          {\;\seq A
            &
            \infer[\mcut]{\Box \Gamma, A^k, \Sigma \seq  B, \Box \Delta
            }
            {\infer[\mathsf{N}]{\;\seq \Box A
              }
              {\;\seq A
              }
              &
              \Box A^{m-k},  \Box \Gamma, A^k, \Sigma \seq B, \Box \Delta
            }
          }
        }
        }
      \]\label{case:NvsK45}
    \item $R_2 = \K\D\4\5$: Similar to case~\ref{case:NvsK45}.
    \end{enumerate}

  \item $R_1 = \4$: Since the rule $\4$ is a special case of each of
    $\C\4$, $\K\4$, and $\K\4\5$, here we only treat the cases not
    involving $\C$, i.e., where $R_2$ is one of
    $\M,\PP,\D,\T,\4,\D\4$. The case of $\D\5$ does not occur with the
    considered logics.
    \begin{enumerate}
    \item $R_2 = (\M)$: similar to case~\ref{case:mon}
    \item $R_2 = \PP$:
      \[
        \vcenter{
        \infer[\mcut]{\Box A \seq \;
        }
        {\infer[\4]{\Box A \seq \Box B
          }
          {\Box A \seq B
          }
          &
          \infer[]{\Box B \seq \;
          }
          {B \seq \;
          }
        }
        }
        \quad\leadsto\quad
        \vcenter{
        \infer[\mcut]{\Box A \seq \;
        }
        {\Box A \seq B
          &
          B \seq \;
        }
        }
      \]\label{case:4vsP}
    \item $R_2 = \D$: The case where $m = 2$ is similar to the
      previous case. If $m = 1$ we have the following reduction, using
      the fact that whenever $\4,\D \in \mathcal{A}$, then
      $\G_{\M\mathcal{A}}$ contains the rule $\D\4$:
      \[
        \vcenter{
        \infer[\mcut]{\Box A, \Box C \seq \;
        }
        {\infer[\4]{\Box A \seq \Box B}{\Box A \seq B}
        &
        \infer[\D]{\Box B, \Box C\seq \; }{B, C \seq \;}
      }
      }
      \qquad
      \leadsto
      \qquad
      \vcenter{
      \infer[\D\4]{\Box A, \Box C \seq \;
      }
      {\infer[\mcut]{\Box A, C \seq \;
        }
        {\Box A \seq B
          &
          C, B \seq \;
        }
      }
      }
      \]\label{case:4vsD}
    \item $R_2 = \T$:
      \[
        \vcenter{
        \infer[\mcut]{\Gamma,\Box A \seq \Delta
        }
        {\infer[\4]{\Box A \seq \Box B
          }
          {\Box A \seq B
          }
          &
          \infer[\T]{\Gamma, \Box B^m \seq \Delta
          }
          {\Gamma, \Box B^{m-1}, B \seq \Delta
          }
        }
      }\hspace{15em}
    \]
    \[
      \hspace{5em}\quad\leadsto\quad
      \vcenter{
        \infer[\Con]{\Gamma, \Box A \seq \Delta
        }
        {\infer[\mcut]{\Gamma, \Box A, \Box A \seq \Delta
          }
          {\Box A \seq B
            &
            \infer[\mcut]{\Gamma, B \seq \Delta
            }
            {\infer[\4]{\Box A \seq \Box B
              }
              {\Box A \seq B
              }
              &
              \Gamma, \Box B^{m-1}, B \seq \Delta
            }
          }
        }
      }
      \]\label{case:4vsT}
    \item $R_2 = \4$:
      \[
        \vcenter{
        \infer[\mcut]{\Box A \seq
        }
        {\infer[\4]{\Box A \seq \Box B
          }
          {A \seq B
          }
          &
          \infer[\4]{\Box B \seq \Box C
          }
          {\Box B \seq C
          }
        }
        }
        \quad\leadsto\quad
        \vcenter{
        \infer[\4]{\Box A \seq \Box C
        }
        {\infer[\mcut]{\Box A \seq C
          }
          {\infer[\4]{\Box A \seq \Box B
            }
            {A \seq B
            }
            &
            \Box B \seq C
          }
        }
        }
      \]\label{case:4vs4}
    \item $R_2 = \D\4$: The case where $m = 1$ is similar to the
      previous case respectively case~\ref{case:4vsD4}. If $m = 2$ we
      have (similarly to case~\ref{case:4vsT}):
      \[
        \vcenter{
        \infer[\mcut]{\Box A \seq
        }
        {\infer[\4]{\Box A \seq \Box B
          }
          {\Box A \seq B
          }
          &
          \infer[\D\4]{\Box B, \Box B \seq \;
          }
          {B, \Box B \seq \;
          }
        }
      }
      \quad\leadsto\quad
      \vcenter{
        \infer[\Con]{ \Box A \seq \;
        }
        {\infer[\mcut]{ \Box A, \Box A \seq \;
          }
          {\Box A \seq B
            &
            \infer[\mcut]{\Box A, B \seq \;
            }
            {\infer[\4]{\Box A \seq \Box B
              }
              {\Box A \seq B
              }
              &
              B, \Box B \seq \;
            }
          }
        }
      }
      \]\label{case:4vsD4}
        \end{enumerate}
  \item $R_1 = \5$: Again, since the rule $\5$ is a special case of
    rule $\K\4\5$, here we only consider the cases not including $\C$,
    i.e., where the logic is not $\K\4\5$ or $\K\D\4\5$. The remaining
    cases are treated in case~\ref{case:left-k45}. In this case $R_2$
    then is one of $\M,\PP, \D, \4, \D\4,\D\5$.
    \begin{enumerate}
    \item $R_2 = \M$: Similar to case~\ref{case:MvsD4}.
    \item $R_2 = \PP$: Similar to case~\ref{case:NvsD4}.\label{case:5vsP}
    \item $R_2 = \D$: Where $m = 1$ this is similar to
      case~\ref{case:4vsD}, using that in this case the calculus also
      includes the rules $\D\5$ and $\4$. Where $n = 1$ and $m = 2$, this is
      similar to case~\ref{case:NvsD4}. 
    \item $R_2 = \4$: For $n = 2$ we have:
      \[
        \vcenter{
        \infer[\mcut]{\;\seq \Box B
        }
        {\infer[\5]{\;\seq \Box A, \Box A
          }
          {\;\seq A, \Box A
          }
          &
          \infer[\4]{\Box A \seq \Box B
          }
          {\Box A \seq B
          }
        }
      }
      \quad\leadsto\quad
      \vcenter{
      \infer[\Con]{\;\seq \Box B
      }
      {\infer[\5]{\;\seq \Box B, \Box B
        }
        {\infer[\W]{\;\seq B, \Box B
          }
          {\infer[\mcut]{\;\seq B
            }
            {\infer[\5]{\; \seq \Box A, \Box A
              }
              {\;\seq A, \Box A
              }
              &
              \Box A \seq B
            }
          }
        }
      }
      }
      \]
      For $n = 1$ we have:
      \[
        \vcenter{
        \infer[\mcut]{\;\seq \Box A, \Box C
        }
        {\infer[\5]{\;\seq \Box A, \Box B
          }
          {\;\seq A, \Box B
          }
          &
          \infer[\4]{\Box B \seq \Box C
          }
          {\Box B \seq C
          }
        }
        }
        \quad\leadsto\quad
        \vcenter{
        \infer[\5]{\;\seq \Box A, \Box C
        }
        {\infer[\mcut]{\;\seq A, \Box C
          }
          {\;\seq A, \Box B
            &
            \infer[\4]{\Box B \seq \Box C
            }
            {\Box B \seq C
            }
          }
        }
        }
      \]\label{case:5vs4}
    \item $R_2 = \D\4$: The case of $n = 1$ and $m = 2$ is similar to
      case~\ref{case:5vsP}, that for $n = 2$ and $m = 1$ to
      case~\ref{case:NvsD4}. For $n=m=2$ we have:
      \[
        \vcenter{
        \infer[\mcut]{\;\seq \;
        }
        {\infer[\5]{\;\seq \Box A, \Box A
          }
          {\;\seq A, \Box A
          }
          &
          \infer[\D\4]{\Box A, \Box A \seq \;
          }
          {A, \Box A \seq \;
          }
        }
        }\hspace{15em}
      \]
      \[
        \hspace{5em}\quad\leadsto\quad
        \vcenter{
        \infer[\mcut]{\;\seq\;
        }
        {\infer[\mcut]{\;\seq A
          }
          {\;\seq A, \Box A
            &
            \infer[\D\4]{\Box A, \Box A \seq \;
            }
            {A, \Box A \seq \;
            }
          }
          &
          \infer[\mcut]{A \seq \;
          }
          {\infer[\5]{\;\seq \Box A, \Box A
            }
            {\;\seq A, A
            }
            &
            A, \Box A \seq \;
          }
        }
        }
      \]
      The case of $n = m = 1$ is similar to cases~\ref{case:4vsD}
      and~\ref{case:4vs4}, using the fact that in this case the
      calculus also includes the rules $\4$ and $\D\5$.\label{case:5vsD4}
    \item $R_2 = \D\5$: Similar to the previous case.\label{case:5vsD5}
    \end{enumerate}
  \item $R_1 = \D\5$: Again, we only consider the cases not including
    $\C$, for the remaining case see case~\ref{case:KD45}. The only
    relevant cases then are that $R_2$ is one of $\M, \PP, \D, \4, \D\4,
    \D\5$.
    \begin{enumerate}
    \item $R_2 = \M$: Similar to case~\ref{case:Mvs4}:
      \[
        \vcenter{
        \infer[\mcut]{\Box A \seq \Box C
        }
        {\infer[\D\5]{\Box A \seq \Box B
          }
          {A \seq \Box B
          }
          &
          \infer[\M]{\Box B \seq \Box C
          }
          {B \seq C
          }
        }
        }
        \quad\leadsto\quad
        \vcenter{
        \infer[\D\5]{\Box A \seq \Box C
        }
        {\infer[\mcut]{A \seq \Box C
          }
          {A \seq \Box B
            &
            \infer[\M]{\Box B \seq \Box C
            }
            {B \seq C
            }
          }
        }
        }
      \]
      \label{case:D5vsM}
    \item $R_2 = \PP$: Similar the previous case.
    \item $R_2 = \D$: The case of $m = 2$ is similar to the previous
      case, with additional structural rules. The case of $m = 1$ is
      similar to case~\ref{case:Mvs4}, using that in this case the
      calculus also includes the rule $\D\4$.\label{case:D5vsD}.
    \item $R_2 = \4$: Similar to case~\ref{case:D5vsM}.
    \item $R_2 = \D\4$:       Similar to case~\ref{case:D5vsD}.
    \item $R_2 = \D\5$: As in case~\ref{case:D5vsM}.
    \end{enumerate}
  \item $R_1 = \C$:
    \begin{enumerate}
    \item $R_2 = \PP$: Similar to case~\ref{case:MvsP}, using that in
      this case also $\C\D$ is in the rule set:
      \[
        \vcenter{
        \infer[\mcut]{\Box \Gamma \seq \;
        }
        {\infer[\C]{\Box \Gamma \seq \Box A
          }
          {\Gamma \seq A
          }
          &
          \infer[\PP]{\Box A \seq \;
          }
          {A \seq \;
          }
        }
        }
        \quad\leadsto\quad
        \vcenter{
        \infer[\C\D]{\Box \Gamma \seq \;
        }
        {\infer[\mcut]{\Gamma \seq \;
          }
          {\Gamma \seq A
            &
            A \seq \;
          }
        }
        }
      \]\label{case:CvsM}
    \item $R_2 = \D$: Similar to the previous case.
    \item $R_2 = \T$:
      \[
        \vcenter{
        \infer[\mcut]{\Sigma, \Box \Gamma \seq \Pi
        }
        {\infer[\C]{\Box \Gamma \seq \Box A
          }
          {\Gamma \seq A
          }
          &
          \infer[\T]{\Sigma, \Box A^m \seq \Pi
          }
          {\Sigma, \Box A^{m-1}, A \seq \Pi
          }
        }
        }\hspace{15em}
      \]
      \[
        \hspace{5em}\quad\leadsto\quad
        \vcenter{
        \infer=[\Con]{\Sigma, \Box \Gamma, \seq \Pi
        }
        {\infer=[\T]{\Sigma, \Box \Gamma, \Box \Gamma \seq \Pi
          }
          {\infer[\mcut]{\Sigma, \Gamma, \Box \Gamma \seq \Pi
            }
            {\Gamma \seq A
              &
              \infer[\mcut]{\Sigma, \Box \Gamma, A \seq \Pi
              }
              {\infer[\C]{\Box \Gamma \seq \Box A
                }
                {\Gamma \seq A
                }
                &
                \Sigma, \Box A^{m-1}, A \seq \Pi
              }
            }
          }
        }
        }
      \]\label{case:CvsT}
    \item $R_2 = \4$: See case~\ref{case:CvsK4}.
    \item $R_2 = \D\4$: Similar to case~\ref{case:CvsK4}.
    \item $R_2 = \D\5$: See case~\ref{case:CvsKD45}.
    \item $R_2 = \C$:
      \[
        \vcenter{
        \infer[\mcut]{\Box \Gamma, \Box \Sigma \seq \Box B
        }
        {\infer[\C]{\Box \Gamma \seq \Box A
          }
          {\Gamma \seq A
          }
          &
          \infer[\C]{\Box A^m, \Box \Sigma \seq \Box B
          }
          {A^m, \Sigma \seq B
          }
        }
        }
        \quad\leadsto\quad
        \vcenter{
        \infer[\C]{\Box \Gamma, \Box \Sigma \seq \Box B
        }
        {\infer[\mcut]{\Gamma, \Sigma \seq B
          }
          {\Gamma \seq A
            &
            A^m, \Sigma \seq B
          }
        }
        }
      \]\label{case:CvsC}
    \item $R_2 = \C\D$: Similar to the previous case.\label{case:CvsCD}
    \item $R_2 = \C\4$: 
      \[
        \vcenter{
        \infer[\mcut]{\Box \Gamma, \Box \Sigma, \Box \Omega \seq \Box B
        }
        {\infer[\C]{\Box \Gamma \seq \Box A
          }
          {\Gamma \seq A
          }
          &
          \infer[\K\4]{\Box A^m,\Box \Sigma, \Box \Omega \seq \Box B
          }
          {\Box A^{m-k}, \Box \Sigma, A^k, \Omega \seq B
          }
        }
        }\hspace{15em}
      \]
      \[
        \hspace{5em}\quad\leadsto\quad
        \vcenter{
        \infer=[\Con]{\Box \Gamma, \Box \Sigma, \Box \Omega \seq \Box B
        }
        {\infer[\K\4]{\Box \Gamma, \Box \Gamma, \Box \Sigma, \Box \Omega
            \seq \Box B
          }
          {\infer[\mcut]{\Gamma, \Box \Gamma, \Box \Sigma, \Omega
            \seq \Box B
          }
          {\Gamma \seq A
            &
            \infer[\mcut]{\Box \Gamma, \Box \Sigma, A^k, \Omega \seq B
            }
            {\infer[\C]{\Box \Gamma \seq \Box A
              }
              {\Gamma \seq A
              }
              &
              \Box A^{m-k}, \Box \Gamma, \Box \Sigma, A^k, \Omega \seq B
            }
          }
          }
        }
        }
      \]\label{case:CvsC4}
    \item $R_2 = \C\D\4$: Similar to case~\ref{case:CvsC4}.\label{case:CvsCD4}
    \item $R_2 = \K\4$: See case~\ref{case:CvsC4}.
       \label{case:CvsK4}
    \item $R_2 = \K\4\5$: Similar to case~\ref{case:CvsKD45}.
    \item $R_2 = \K\D\4\5$: Similar to case~\ref{case:CvsK4}:
            \label{case:CvsKD45}
    \end{enumerate}\label{case:C}
  \item $R_1 = \C\4$: Similar to case~\ref{case:C}.\label{case:C4}
  \item $R_1 = \K\4$: See case~\ref{case:C4} and case~\ref{case:N}.\label{case:K4}
  \item $R_1 = \K\4\5$: This case only occurs for the logics $\K\4\5$
    and $\K\D\4\5$, limiting the possible cases to the following:
    \begin{enumerate}
    \item $R_2 = \4$: See case~\ref{case:K45vsK45}.
    \item $R_2 = \D\4$: See case~\ref{case:K45vsKD45}.
    \item $R_2 = \D\5$: See case~\ref{case:K45vsKD45}.
    \item $R_2 = \C$: See case~\ref{case:K45vsK45}.
    \item $R_2 = \C\D$: See case~\ref{case:K45vsKD45}.
    \item $R_2 = \C\4$: See case~\ref{case:K45vsK45}.
    \item $R_2 = \C\D\4$: See case~\ref{case:K45vsKD45}.
    \item $R_2 = \K\4$: See case~\ref{case:K45vsK45}.
    \item $R_2 = \K\4\5$: We show the most interesting case, the
      remaining cases are similar.
      \[
        \vcenter{
        \infer[\mcut]{\Box \Gamma, \Box \Sigma, \Box \Omega, \Box
          \Theta \seq \Box \Delta, \Box B, \Box \Pi
        }
        {\infer[\K\4\5]{\Box \Gamma, \Box \Sigma \seq \Box A^n, \Box \Delta
          }
          {\Box \Gamma, \Sigma \seq A, \Box A^{n-1}, \Box \Delta
          }
          &
          \infer[\K\4\5]{\Box A^m, \Box \Omega, \Box \Theta \seq \Box B, \Box \Pi
          }
          {\Box A^{m-k}, \Box \Omega, A^k, \Theta \seq B, \Box \Pi
          }
        }
        }\hspace{15em}
      \]
      \[
        \leadsto\qquad
        \vcenter{
        \infer=[\Con]{\Box \Gamma, \Box \Sigma, \Box \Omega, \Box
          \Theta \seq \Box \Delta, \Box B, \Box \Pi
        }
        {\infer[\K\4\5]{\Box \Gamma, \Box \Sigma, \Box \Omega, \Box
          \Theta,\Box \Gamma, \Box \Sigma, \Box \Omega, \Box
          \Theta \seq \Box \Delta, \Box B, \Box \Pi,\Box \Delta, \Box B, \Box \Pi
          }
          {\infer[\mcut]{\Box \Gamma, \Sigma,
              \Box \Omega, \Box \Theta,\Box \Gamma, \Box \Sigma, \Box \Omega,
            \Theta \seq \Box \Delta,\Box B, \Box \Pi, \Box \Delta, B,
            \Box \Pi
            }
            {\infer*[\mathcal{D}_1]{\Box \Gamma, \Sigma, \Box \Omega, \Box
            \Theta \seq A, \Box \Delta, \Box B, \Box \Pi
              }
              {              }
              &\qquad
              \infer*[\mathcal{D}_2]{\Box \Gamma, \Box \Sigma, \Box \Omega, A^k, 
            \Theta \seq \Box \Delta,  B, \Box \Pi
              }
              {
              }
            }
          }
        }}
      \]
      where $\mathcal{D}_1$ is
      \[
        \infer[\mcut]{\Box \Gamma, \Sigma, \Box \Omega, \Box
            \Theta \seq A, \Box \Delta, \Box B, \Box \Pi
              }
              {\Box \Gamma, \Sigma \seq A, \Box \Delta
                &
                \infer[\K\4\5]{\Box A^m, \Box \Omega, \Box \Theta \seq
                  \Box B, \Box \Pi
                }
                {\Box A^{m-k}, \Box \Omega, A^k, \Theta \seq B, \Box \Pi
                }
              }
            \]
            and $\mathcal{D}_2$ is
            \[
              \infer[\mcut]{\Box \Gamma, \Box \Sigma, \Box \Omega, A^k, 
            \Theta \seq \Box \Delta,  B, \Box \Pi
              }
              {\infer[\K\4\5]{\Box \Gamma, \Box \Sigma \seq \Box A^n,
                  \Box \Delta
                }
                {\Box \Gamma, \Sigma \seq A, \Box A^{n-1}, \Box \Delta
                }
                &
                \Box A^{m-k}, \Box \Omega, A^k, \Theta \seq B, \Box \Pi
              }
            \]\label{case:K45vsK45}
    \item $R_2 = \K\D\4\5$: Similar to the previous case.\label{case:K45vsKD45}
    \end{enumerate}
    \label{case:left-k45}
  \item $R_1 = \K\D\4\5$: Similar to case~\ref{case:left-k45}.\label{case:KD45}
  \end{enumerate}
\end{enumerate}
\end{proof}

%%% Local Variables: 
%%% mode: latex
%%% TeX-master: "TOCL"
%%% End: 

\end{document}